\documentclass[12pt]{article}
\usepackage{amsmath}
\usepackage{graphicx}
\usepackage{enumerate}
\usepackage{natbib}
\usepackage{url} % not crucial - just used below for the URL 

\usepackage{amsthm}
\usepackage{latexsym,epsfig,graphicx,epstopdf,amsmath,amssymb,amscd,undertilde}
\usepackage{multirow,psfrag,paralist,dsfont,url}
\usepackage[titletoc]{appendix}
\usepackage{float}
\usepackage{statmath}
\usepackage{prodint}
\usepackage[dvipsnames]{xcolor}
\usepackage{enumitem}

\usepackage[american]{babel}

\usepackage[section]{placeins}
\usepackage{xr}

\usepackage{mathtools}
\usepackage[mathscr]{euscript}
\mathtoolsset{showonlyrefs=true}

\allowdisplaybreaks

\usepackage{chngcntr}
\usepackage{apptools}
\AtAppendix{\counterwithin{assumption}{section}}
\AtAppendix{\counterwithin{lemma}{section}}
\AtAppendix{\counterwithin{theorem}{section}}
\AtAppendix{\counterwithin{remark}{section}}
\AtAppendix{\counterwithin{coro}{section}}
\AtAppendix{\counterwithin{figure}{section}}
\AtAppendix{\numberwithin{equation}{section}}

\newcommand{\func}{{f}}

\newcommand{\dgoto}{\xrightarrow{\mathrm{d}}}
\newcommand{\Pgoto}{\overset{\mathbb{P}}{\rightarrow}}
\newcommand{\asgoto}{\overset{\mathrm{a.s.}}{\longrightarrow}}

\newcommand{\bigO}{\mathcal{O}}

\theoremstyle{plain}

\newtheorem{theorem}{Theorem}
\newtheorem{lemma}{Lemma}
\newtheorem{coro}{Corollary}

\newtheorem{remark}{Remark}
\newtheorem{assumption}{Assumption}

\usepackage{xpatch}
\xpatchcmd{\proof}{\itshape}{\normalfont\proofnamefont}{}{}
\newcommand{\proofnamefont}{}
\renewcommand{\proofnamefont}{\bfseries}

\usepackage{subcaption}
\usepackage{booktabs}

\begin{document}

	\def\spacingset#1{\renewcommand{\baselinestretch}%
		{#1}\small\normalsize} \spacingset{1}
	
	\setlength{\bibsep}{0pt plus 0.3ex}% condense spacing in bibliography
	
	\setcounter{page}{1}
	
	%%%%%%%%%%%%TITLE%%%%%%%%%%%%%%%%%%%%%%%%%%%%%%%%%%%
	\title{Effect-Wise Inference for Smoothing Spline ANOVA\\ on Tensor-Product Sobolev Space}
	%%%%%%%%%%%%%%%%%%%%%%%%%%%%%%%%%%%%%%%%%%%%%%%%%%%%%%%%%%%%%%%%%%%%%%%%%%%%%%%%%%%%%%%%%%%%%%%%%%%%%%%%%%%%%%%%%
	
	%\iffalse
	\author{
		Youngjin Cho\thanks{Corresponding author. E-mail: \texttt{youngjin.cho@unlv.edu}.}\\[1.5ex]
		{Department of Mathematical Sciences}\\
		{University of Nevada, Las Vegas, Las Vegas, NV 89154}\\[3ex]
		Meimei Liu\\[1.5ex]
		{Department of Statistics}\\
		{Virginia Tech, Blacksburg, VA 24061}
	}
	%\fi
	
	\date{\today}
	\date{}
	
	\maketitle
	
	\begin{abstract}
		Functional ANOVA provides a nonparametric modeling framework for multivariate covariates, enabling flexible estimation and interpretation of effect functions such as main effects and interaction effects. However, effect-wise inference in such models remains challenging. Existing methods focus primarily on inference for entire functions rather than individual effects. Methods addressing effect-wise inference face substantial limitations: the inability to accommodate interactions, a lack of rigorous theoretical foundations, or restriction to pointwise inference.  To address these limitations, we develop a unified framework for effect-wise inference in smoothing spline ANOVA on a subspace of tensor product Sobolev space. For each effect function, we establish rates of convergence, pointwise confidence intervals, and a Wald-type test for whether the effect is zero, with power achieving the minimax distinguishable rate up to a logarithmic factor.  Main effects achieve the optimal univariate rates, and interactions achieve optimal rates up to logarithmic factors. The theoretical foundation relies on an orthogonality decomposition of effect subspaces, which enables the extension of the functional Bahadur representation framework to effect-wise inference in smoothing spline ANOVA with interactions. Simulation studies and real-data application to the Colorado temperature dataset demonstrate superior performance compared to existing methods.

		%Functional ANOVA is a widely used nonparametric modeling framework for multivariate covariates, enabling flexible estimation and interpretation of effect functions such as main effects and interaction effects. An important yet often underexplored aspect of such models is effect-wise inference. Most existing methods focus on inference for the entire function and only rarely address effect-wise inference; those that do often face limitations: some cannot accommodate interactions, others lack rigorous theory, and some are restricted to pointwise inference. In this paper, we develop a general framework for effect-wise inference in smoothing spline ANOVA on a subspace of tensor product Sobolev space, applicable to a broad range of main-effect and interaction-effect scenarios. Our contributions, based on functional Bahadur representation framework, for each effect function, include the rate of convergence, pointwise confidence intervals, the limiting distribution under the null and a power analysis with a distinguishable rate for a Wald-type test of whether the effect function is zero. For each main effect function, the rates in these results match their univariate counterparts, while for interaction, the rates deteriorate only by logarithmic factors. To the best of our knowledge, these are the first results of their kind in the literature. Simulation studies and real-data applications demonstrate that our method delivers more accurate effect-wise inference than existing approaches.
		
		\textbf{Key Words:} {Effect-wise inference, Functional Bahadur representation,\\Main/interaction effect function, Nonparametric testing, Smoothing spline ANOVA, Tensor-product Sobolev space}
	\end{abstract}
	\newpage
	\spacingset{1.8}
	\section{Introduction}
	In nonparametric modeling with multivariate covariates, functional ANOVA \citep{hoeffding_1948,stone_1994,wahba_1995} provides a framework for examining covariate effects by decomposing the regression function into an intercept, main effects for each covariate, and interactions of increasing orders. This structure enables the estimation of each effect function individually, allowing for the inclusion or exclusion of specific components. Smoothing spline ANOVA \citep{gu2013smoothing}, one specific implementation constructed on tensor products of univariate reproducing kernel Hilbert spaces, establishes a computational foundation for this approach. However, determining whether a specific component contributes to the model remains a challenge. While classical ANOVA and linear regression utilize hypothesis tests for model selection, statistical inference for individual effect functions, specifically pointwise inference and testing whether an effect is zero, remains underdeveloped in the nonparametric setting.
	
	The literature on nonparametric inference has focused on the entire regression function rather than its individual components. In the univariate setting, \cite{shang_2013, shang_2015} established local and global inference using a functional Bahadur representation on univariate Sobolev spaces, providing power analyses with minimax distinguishable rates. However, these developments target the entire function without partitioning it into individual effects. For multivariate covariates, several extensions utilize the functional Bahadur representation: \cite{Zhao02012021} addressed panel-data inference on multivariate domains; \cite{liu_2020} proposed a Wald-type test within a divide-and-conquer framework for additive and thin-plate splines; and \cite{Dette_2024} studied bivariate function-on-function regression. Beyond this framework, \cite{tuo2024asymptotictheorylinearfunctionals} developed an asymptotic theory for linear functionals, such as point evaluations and derivatives, on Sobolev spaces. Despite these advancements, these methods prioritize global rather than effect-wise inference.
	
	Effect-wise inference in nonparametric modeling remains limited. \cite{Fan_2005} proposed global tests for main effects in additive models but omitted interactions. In the smoothing spline ANOVA context, \cite{Gu01031993, gu2013smoothing} introduced Kullback–Leibler measures and Bayesian confidence intervals to assess effect importance, yet these lack a large-sample theoretical framework for hypothesis testing. \cite{SSA_Guo} developed tests for individual effect functions, but the results are restricted to null distributions without power analysis. While the \texttt{mgcv} approach \citep{wood2017generalized} performs tests for main and interaction effects, these are essentially parametric Wald tests on penalized coefficients since it is based on optimization in finite-dimensional spaces via basis expansions. Furthermore, \cite{Liu_Ruiqi_2023} proposed derivative-based tests to determine if the derivative of the entire function with respect to a subset of covariates is zero across the corresponding domain, which is less flexible than testing individual components directly. In the high-dimensional setting, \cite{COSSO} handled both main and interaction effects, while \cite{Zhang01092011} considered only main effects divided into linear and nonparametric parts. Both developed effect-wise selection methods and established convergence rates in empirical norm and selection consistency, but neither provided limiting distributions or testing procedures under a fixed-dimensional framework.
	
	Recent research has addressed effect-wise inference for specific interaction terms involving categorical variables. \cite{JMLR:v21:19-800} proposed a parallelism test for interactions between a continuous and a categorical covariate, while \cite{xing_2024} tested the invariance of multivariate densities across groups. While these works conduct effect-wise inference, they are restricted to interactions with categorical components. Therefore, there remains a need for an inference framework  for both main effects and general interactions on continuous domains, supported by asymptotic theory and power analysis.
	
	This paper develops a unified framework for effect-wise inference in smoothing spline ANOVA on a tensor product Sobolev space. Consider the nonparametric regression model $Y = f^*(X) + \epsilon$, where $X = (X_{[1]}, \ldots, X_{[d]})^\top$ is a $d$-dimensional covariate vector on $[0,1]^d$ and $f^*$ belongs to a reproducing kernel Hilbert space with a functional ANOVA decomposition. Specifically, $f^*$ can be expressed as $
	f^* = \sum_{S \in \mathbb{S}} f^*_S,$ 
	where $\mathbb{S}$ is a collection of subsets of $\{1, \ldots, d\}$ specifying the effects in the model, and each $f^*_S$ represents the effect function associated with the covariates indexed by $S$. The set $S = \emptyset$ corresponds to the intercept, singleton sets correspond to main effects, and sets with two or more elements correspond to interactions of the corresponding order.
	
	Our goal is inference on individual effect functions $f^*_S$. For each effect, whether main effect or interaction, we establish: (i) rates of convergence, (ii) pointwise limiting distributions for constructing confidence intervals, and (iii) a Wald-type test for the null hypothesis $\mathrm{H}_{0,S}: f^*_S = 0$, together with the limiting distribution under the null and power analysis with distinguishable rate. For local inference on effect $S$, let $\hat{f}_S$ denote the estimated effect function and let $\mathcal{V}_{S,\lambda}(x_S)/n \to 0$ denote the asymptotic variance at point $x_S$, where $\mathcal{V}_{S,\lambda}(x_S) \to \infty$. We establish
	\begin{equation*}
		\frac{\sqrt{n}\,\big(\hat{f}_S(x_S) - f^*_S(x_S)\big)}{\sqrt{\mathcal{V}_{S,\lambda}(x_S)}} \dgoto \mathrm{N}(0,1).
	\end{equation*}For global inference, we derive the limiting distribution of Wald-type test statistic $\mathcal{T}_{S,\lambda}$ under $\mathrm{H}_{0,S}$, namely,
	\begin{equation*}
		\mathcal{T}_{S,\lambda} \dgoto \mathrm{N}(0,1),
	\end{equation*}and establish its power with a minimax distinguishable rate up to a logarithmic factor \citep{ingster1993asymptotically}.
	
	The convergence rates we obtain reflect the structure of effect functions. The intercept achieves the parametric rate. Main effects attain the same rates as in the univariate inference of \cite{shang_2013}, while interaction effects deteriorate only by a logarithmic factor as the interaction order increases. Each effect function achieves the optimal convergence rate up to a logarithmic factor under a conventional tuning order, and the minimax optimal testing rate up to a logarithmic factor is attained under another such order \citep{ingster1993asymptotically}.
	
	The theoretical foundation rests on the orthogonality of effect subspaces with respect to a non-negative definite bivariate functional $V$ central to the functional Bahadur representation framework. This orthogonality arises from the effect decomposition induced by the averaging operators and their connection to $V$. As a consequence, the eigensystem for the entire function space decomposes as the union of eigensystems for individual effect spaces, enabling derivation of all theoretical results on an effect-by-effect basis. While this orthogonality holds straightforwardly in additive models \citep{liu_2020}, establishing it for models with arbitrary interaction effects requires a more delicate analysis of the tensor product structure. This extension constitutes a key theoretical contribution, allowing the functional Bahadur representation framework of \cite{shang_2013} and \cite{liu_2020} to be applied to smoothing spline ANOVA with interactions.
	
	To our knowledge, these are the first results enabling rigorous effect-wise inference in nonparametric functional ANOVA with interactions, extending the classical ANOVA paradigm to flexible nonparametric modeling. This provides practitioners with principled tools for determining which effects to include in a model. In simulations with up to two-factor interactions, we compare our Wald-type test with \texttt{mgcv}, demonstrating higher power and superior Type I error control across all scenarios. We also show that our pointwise confidence intervals more accurately reflect the limiting distribution and are narrower than the Bayesian intervals of \cite{gu2013smoothing}. In an application to the Colorado temperature dataset, our method effectively selects effect functions and provides local interpretations that align with expected geographical and seasonal patterns.

	The paper is organized as follows. Section~\ref{sec:model_RKHS} presents the model specification, the tensor product reproducing kernel Hilbert space with its eigensystem, and the estimation procedure. Section~\ref{sec:inference} develops the main theoretical results and their implementation. Section~\ref{sec:simul} contains simulation studies, and Section~\ref{sec:real} presents real-data application to the Colorado temperature dataset. Section~\ref{sec:conclusion} concludes with future research directions. Notations, technical details, and additional numerical results are provided in the Supplementary Material.
	
	%{\color{red}{I think you may need to change the Appendix to Supplementary.}}
	
	%The paper is organized as follows. Section~\ref{sec:New_Idea} includes the model specification (Section~\ref{sec:model_structure}), the tensor product reproducing kernel Hilbert space and its eigensystem (Section~\ref{sec:RKHS}), and the estimation procedure (Section~\ref{sec:estimation}). Section~\ref{sec:inference} presents the main theoretical results (Section~\ref{sec:main_theoretical_result}) and its implementation (Section~\ref{sec:implement}). Section~\ref{sec:simul} contains the simulation studies, Section~\ref{sec:real} presents the real-data applications {\color{black}on the Colorado temperature dataset,} and Section~\ref{sec:conclusion} concludes with a discussion of future research directions. Notations and technical details are collected in Sections \ref{sec:notations} and \ref{sec:Tech_detail} of the Appendix, respectively.

	\section{Model Specification and Reproducing Kernel Hilbert Space Framework}\label{sec:model_RKHS}
	
	\subsection{Model Specification}\label{sec:model_structure}
	
	Consider a nonparametric regression model with response variable \( Y \in \mathbb{R} \) and \( d \)-dimensional covariate vector \( X = (X_{[1]}, \dots, X_{[d]})^{\top} \in \mathcal{X} = \prod_{j=1}^d \mathcal{X}_{[j]} = [0,1]^d \):
	\[
	Y = \func^\ast(X) + \epsilon,
	\]
	where \( \func^\ast(X) = \mathbb{E}(Y|X) \) is the unknown regression function and \( \epsilon \) is a mean-zero noise term with \( \mathrm{Var}(\epsilon) = \sigma^2 \), bounded fourth moment \( \mathbb{E}(\epsilon^4) < \infty \), and independent of \( X \). We assume \( X_{[1]}, \dots, X_{[d]} \) are independent, with each \( X_{[j]} \) following a uniform distribution on \( [0,1] \), ensuring sufficient data coverage across the domain \( \mathcal{X} \), which is a standard assumption in smoothing spline ANOVA with interaction terms. Given a sample of size \( n \), \( \{Z_i\}_{i=1}^n = \{X_i, Y_i, \epsilon_i\}_{i=1}^n \) are independent and identically distributed copies of \( Z = (X, Y, \epsilon) \).
	
	The regression function \( \func^\ast \) is assumed to belong to a reproducing kernel Hilbert space (RKHS) \( \mathcal{H} \) that admits a functional ANOVA decomposition \( \func^\ast = \sum_{S \in \mathbb{S}} \func^\ast_S \), where \( \mathbb{S} \) is a collection of subsets of \( \{1, \ldots, d\} \) specifying the effects in the model. Here, \( S = \emptyset \) corresponds to the intercept, singleton sets correspond to main effects, and sets with two or more elements correspond to interactions of the corresponding order. Our goal is to develop a unified framework for effect-wise inference, enabling estimation, pointwise confidence intervals, and hypothesis testing for each individual effect function \( \func^\ast_S \). The following subsection formalizes the tensor product RKHS structure that makes this decomposition precise.
	
	\subsection{Tensor Product Reproducing Kernel Hilbert Space}\label{sec:RKHS}
	
	We construct the RKHS $\mathcal{H}$ via tensor products of univariate RKHSs, following the smoothing spline ANOVA framework of \cite{gu2013smoothing}. To formalize the decomposition structure of the regression function $f^\ast = \sum_{S \in \mathbb{S}} f^\ast_S$, where each effect function $f^\ast_S$ belongs to an appropriate function space, we begin with the univariate Sobolev space for each covariate axis.
	
	For each \( j \in \{1, \ldots, d\} \), the Sobolev space of order \( m \geq 2 \) on \( \mathcal{X}_{[j]} = [0,1] \) is defined as 
	\begin{align}\label{eqn_Sobolev_detail}
		\mathcal{H}_{[j]} = \left\{ f \in \mathcal{L}_2(\mathcal{X}_{[j]}) : f, f^{(1)}, \ldots, f^{(m-1)} \text{ are absolutely continuous, } f^{(m)} \in \mathcal{L}_2(\mathcal{X}_{[j]}) \right\},
	\end{align} 
	where \( \mathcal{L}_2(\mathcal{X}_{[j]}) = \left\{ f : \mathcal{X}_{[j]} \rightarrow \mathbb{R}, \int_{\mathcal{X}_{[j]}} |f(x_{[j]})|^2 dx_{[j]} < \infty \right\} \). Let \( id \) denote the identity operator and \( \mathcal{A}_{[j]} \) denote the averaging operator on $\mathcal{X}_{[j]}$:
	\begin{align}\label{eqn:avj_op}
		\mathcal{A}_{[j]} f(x) = \int_{\mathcal{X}_{[j]}} f(x_{[1]}, \ldots, x_{[d]}) dx_{[j]}.
	\end{align} 
	The space \( \mathcal{H}_{[j]} \) admits the orthogonal decomposition \( \mathcal{H}_{[j]} = \mathcal{H}_{\emptyset[j]} \oplus \mathcal{H}_{\{j\}} \), where \( \mathcal{H}_{\emptyset[j]} = \mathrm{Span}\{1\} \) is the constant space and \( \mathcal{H}_{\{j\}} = \{ f_{[j]} \in \mathcal{H}_{[j]} : \mathcal{A}_{[j]} f_{[j]} = 0 \} \) is the centered effect space. Any \( f_{[j]} \in \mathcal{H}_{[j]} \) decomposes as \( f_{[j]} = f_{\emptyset[j]} + f_{\{j\}} \), where \( f_{\emptyset[j]} = \mathcal{A}_{[j]} f_{[j]} \in \mathcal{H}_{\emptyset[j]} \) and \( f_{\{j\}} = (id - \mathcal{A}_{[j]}) f_{[j]} \in \mathcal{H}_{\{j\}} \).
	
	Each subspace is equipped with a reproducing kernel (RK) and a corresponding inner product. For the constant space \( \mathcal{H}_{\emptyset[j]} \), the RK and inner product are
	\[
	\mathcal{K}_{\emptyset[j]}(x_{[j]}, x_{[j]}') = 1, \quad \langle f_{\emptyset[j]}, g_{\emptyset[j]} \rangle_{\emptyset[j]} = f_{\emptyset[j]} g_{\emptyset[j]}.
	\]
	For the effect space \( \mathcal{H}_{\{j\}} \), the RK and inner product are
	\[
	\mathcal{K}_{\{j\}}(x_{[j]}, x_{[j]}') = \sum_{l=1}^{m} \kappa_{l[j]}(x_{[j]}) \kappa_{l[j]}(x_{[j]}') + (-1)^{m-1} \kappa_{2m[j]}(|x_{[j]} - x_{[j]}'|),
	\]
	\[
	\langle f_{\{j\}}, g_{\{j\}} \rangle_{\{j\}} = \sum_{l=1}^{m-1} \int_{\mathcal{X}_{[j]}} f_{\{j\}}^{(l)}(x_{[j]}) dx_{[j]} \int_{\mathcal{X}_{[j]}} g_{\{j\}}^{(l)}(x_{[j]}) dx_{[j]} + \int_{\mathcal{X}_{[j]}} f_{\{j\}}^{(m)}(x_{[j]}) g_{\{j\}}^{(m)}(x_{[j]}) dx_{[j]},
	\]
	where \( \kappa_{l[j]}(\cdot) = B_{l[j]}(\cdot)/l! \) is the \( l \)th order scaled Bernoulli polynomial on \( \mathcal{X}_{[j]} \). The full space \( \mathcal{H}_{[j]} \) then has RK \( \mathcal{K}_{[j]} = \mathcal{K}_{\emptyset[j]} + \mathcal{K}_{\{j\}} \) and inner product \( \langle f_{[j]}, g_{[j]} \rangle_{[j]} = \langle f_{\emptyset[j]}, g_{\emptyset[j]} \rangle_{\emptyset[j]} + \langle f_{\{j\}}, g_{\{j\}} \rangle_{\{j\}} \).
	
	The tensor product Sobolev space over all axes decomposes as
	\begin{align}
		\otimes_{j=1}^d \mathcal{H}_{[j]} = \otimes_{j=1}^d \left\{ \mathcal{H}_{\emptyset[j]} \oplus \mathcal{H}_{\{j\}} \right\} = \oplus_{S \in \mathcal{P}_d} \mathcal{H}_S,
	\end{align} 
	where \( \mathcal{P}_d \) denotes the power set of \( \{1, \ldots, d\} \). For each \( S \in \mathcal{P}_d \), the effect space is \( \mathcal{H}_S = \otimes_{j \in S} \mathcal{H}_{\{j\}} \), defined on \( \mathcal{X}_S = \prod_{j \in S} \mathcal{X}_{[j]} \). The intercept space corresponds to \( S = \emptyset \) with \( \mathcal{H}_\emptyset = \mathrm{Span}\{1\} \). Any function \( f \in \otimes_{j=1}^d \mathcal{H}_{[j]} \) admits the functional ANOVA decomposition \( f = \sum_{S \in \mathcal{P}_d} f_S \), where 
	\begin{align}\label{eqn:proj} 
		f_S = \big\{ \prod_{j \in S} (id - \mathcal{A}_{[j]}) \prod_{j \in \{1,\ldots,d\} \setminus S} \mathcal{A}_{[j]} \big\} f \in \mathcal{H}_S.
	\end{align}
	
	Correspondingly, each effect space \( \mathcal{H}_S \) is equipped with RK 
	\[
	\mathcal{K}_S(x_S, x_S') = \prod_{j \in S} \mathcal{K}_{\{j\}}(x_{[j]}, x_{[j]}') 
	\]
	and a corresponding inner product $ \langle \cdot, \cdot \rangle_S$, 
	where \( x_S = \{x_{[j]}\}_{j \in S} \). For the intercept space, \( \mathcal{K}_\emptyset = 1 \) and \( \langle f_\emptyset, g_\emptyset \rangle_\emptyset = f_\emptyset g_\emptyset \).
	
	As introduced in Section~\ref{sec:model_structure}, we consider the model space $$ \mathcal{H} = \oplus_{S \in \mathbb{S}} \mathcal{H}_S $$ for a specified subset \( \emptyset \in \mathbb{S} \subseteq \mathcal{P}_d \). The choice of \( \mathbb{S} \) determines the model structure: \( \mathbb{S} = \{\emptyset, \{1\}, \ldots, \{d\}\} \) yields an additive model, while including sets of cardinality two or higher incorporates interaction effects. The RK and inner product on \( \mathcal{H} \) are
	\[
	\mathcal{K}(x, x') = \sum_{S \in \mathbb{S}} \mathcal{K}_S(x_S, x_S'), \quad \langle f, g \rangle = \sum_{S \in \mathbb{S}} \langle f_S, g_S \rangle_S,
	\]
	where \( f = \sum_{S \in \mathbb{S}} f_S \) and \( g = \sum_{S \in \mathbb{S}} g_S \) with \( f_S, g_S \in \mathcal{H}_S \). %We assume that $f^\ast =\sum_{S \in \mathbb{S}} f^\ast_S \in \mathcal{H}$ with $f^\ast_S \in \mathcal{H}_S$.
	
	The inner products on effect spaces naturally define the penalty for smoothing spline estimation. Let \( J_S(\cdot, \cdot) \equiv \langle \cdot, \cdot \rangle_S \) for \( S \in \mathbb{S} \setminus \{\emptyset\} \) and \( J_\emptyset(\cdot, \cdot) \equiv 0 \), so that the intercept is unpenalized. The overall penalty is
	\[
	J(f, g) = \sum_{S \in \mathbb{S}} J_S(f_S, g_S) = \sum_{S \in \mathbb{S} \setminus \{\emptyset\}} \langle f_S, g_S \rangle_S,
	\]
	which is an inner product on \( \mathcal{H}_J \equiv \oplus_{S \in \mathbb{S} \setminus \{\emptyset\}} \mathcal{H}_S \) with paired RK $$ \mathcal{K}_J(x, x') = \sum_{S \in \mathbb{S} \setminus \{\emptyset\}} \mathcal{K}_S(x_S, x_S').$$ For notational convenience, we write \( J(f) = J(f, f) \) and \( J_S(f_S) = J_S(f_S, f_S) \). The quantity \( J_S(f_S) \) measures the roughness of the effect function \( f_S \), with larger values indicating less smooth functions. The following assumption ensures that the true effect functions have bounded roughness.
	
	\begin{assumption}\label{asp:smooth}
		For each \( S \in \mathbb{S} \setminus \{\emptyset\} \), \( J_S(\func^\ast_S) \leq \mathcal{C}_{J,S} \) for a constant \( \mathcal{C}_{J,S} \in (0, \infty) \).
	\end{assumption}
	
	Assumption~\ref{asp:smooth} imposes sufficient smoothness on each effect function, requiring that \( \func^\ast_S \) lies in the interior of \( \mathcal{H}_S \) with respect to the roughness measure \( J_S \). This is a standard regularity condition in smoothing spline theory and ensures that the penalty term is well-defined and finite for the true function.

	\subsection{Orthogonality of Effect Subspaces}\label{sec:orthogonality}
	
	A central question for effect-wise inference is whether the decomposition \( \func^\ast = \sum_{S \in \mathbb{S}} \func^\ast_S \) yields components that can be estimated and tested independently. This depends on the geometric relationship between the effect spaces \( \mathcal{H}_S \). Following the framework of \cite{shang_2013} and \cite{liu_2020}, we introduce the bilinear form \( V \) on \( \mathcal{H} \): for any \( f, g \in \mathcal{H} \),
	\begin{align}
		V(f,g) = \mathbb{E}_X \left( f(X) g(X) \right) = \langle f, g \rangle_{\mathcal{L}_2(\mathcal{X})} = \int_{\mathcal{X}} f(x) g(x) \, dx.
	\end{align}
	The following lemma establishes that distinct effect spaces are orthogonal with respect to both \( V \) and the penalty \( J \).
	
	\begin{lemma}\label{lem:ortho}  
		For all \( S \neq S' \) in \( \mathcal{P}_d \), and for any \( f_S \in \mathcal{H}_S \) and \( g_{S'} \in \mathcal{H}_{S'} \), we have  
		\[
		V(f_S, g_{S'}) = J(f_S, g_{S'}) = 0.
		\]  
	\end{lemma}
	
	The proof is provided in Section~\ref{sec:Lemma_cor_proof} of the Supplementary Material. While this orthogonality for additive models is straightforward \citep{liu_2020}, we establish that it extends to models with arbitrary interaction effects, which is a key theoretical foundation for effect-wise inference. This orthogonality is essential: it allows the eigensystem of \( \mathcal{H} \) to be constructed as the union of the eigensystems of the individual effect spaces \( \mathcal{H}_S \), providing the foundation for effect-wise inference developed in Section~\ref{sec:inference}. The assumption of independent and uniformly distributed covariates is crucial here, ensuring that \( V \) coincides with the \( \mathcal{L}_2 \) inner product on \( \mathcal{X} \), under which the orthogonality holds due to the centering property of the averaging operator \eqref{eqn:avj_op} and the projection \eqref{eqn:proj}.
	
	For each \( S \in \mathbb{S} \), we define the restriction of \( V \) to \( \mathcal{H}_S \): for \( f_S, g_S \in \mathcal{H}_S \),
	\[
	V_S(f_S, g_S) = \mathbb{E}_{X_S} \left( f_S(X_S) g_S(X_S) \right) = \langle f_S, g_S \rangle_{\mathcal{L}_2(\mathcal{X}_S)} =\int_{\mathcal{X}_S} f_S(x_S) g_S(x_S) \, dx_S,
	\] 
	where \( X_S = \{X_{[j]}\}_{j \in S} \). Similarly, \( J_S \) is the restriction of \( J \) to \( \mathcal{H}_S \), satisfying \( J_S(f_S, g_S) = J(f_S, g_S) \). For notational convenience, we write \( V(f) = V(f,f) \) and \( V_S(f_S) = V_S(f_S, f_S) \).

	\subsection{Eigensystem and Inner Product Structure}\label{sec:eigensystem}
	
	We next define a $\lambda$-weighted inner product on \( \mathcal{H} \) by \( \langle \cdot, \cdot \rangle_\lambda = (V + \lambda J)(\cdot, \cdot) \), with corresponding RK \( \mathcal{R}_\lambda \) and norm \( \|\cdot\|_\lambda \), where $\lambda>0$ is the tuning parameter. For each \( S \in \mathbb{S} \), the restriction to \( \mathcal{H}_S \) yields inner product \( \langle \cdot, \cdot \rangle_{S,\lambda} = (V_S + \lambda J_S)(\cdot, \cdot) \), with corresponding RK \( \mathcal{R}_{S,\lambda} \) and norm \( \|\cdot\|_{S,\lambda} \). For the intercept space $\mathcal{H}_\emptyset$, since \( J_\emptyset = 0 \), we have \( \langle f_\emptyset, g_\emptyset \rangle_{\emptyset,\lambda} = V_\emptyset(f_\emptyset, g_\emptyset) = f_\emptyset g_\emptyset \) and \( \mathcal{R}_{\emptyset,\lambda} = 1 \).
	
	The orthogonality of effect spaces extends to the \( \lambda \)-weighted inner product, as stated in the following corollary.
	
	\begin{coro}\label{cor:Inner_prod_RK_ortho}
		For \( f = \sum_{S \in \mathbb{S}} f_S, g = \sum_{S \in \mathbb{S}} g_S \in \mathcal{H} \) with \( f_S, g_S \in \mathcal{H}_S \),
		\[
		\langle f, g \rangle_\lambda = \sum_{S \in \mathbb{S}} \langle f_S, g_S \rangle_{S,\lambda}, \quad \mathcal{R}_\lambda(x, x') = \sum_{S \in \mathbb{S}} \mathcal{R}_{S,\lambda}(x_S, x'_S).
		\]
	\end{coro}
	
	The proof is provided in Section~\ref{sec:Lemma_cor_proof} of the Supplementary Material. This decomposition shows that the \( \lambda \)-weighted inner product and RK inherit the additive structure of the effect decomposition. As a consequence, the squared norm \( \|f\|_\lambda^2 = \sum_{S \in \mathbb{S}} \|f_S\|_{S,\lambda}^2 \) separates into effect-wise contributions, which is essential for analyzing each effect function independently.
	
	%The proof is provided in Section~\ref{sec:Tech_detail} of the Supplementary Material. This decomposition is essential: it implies that the eigensystem of \( \mathcal{H} \) can be constructed as the union of the eigensystems of the individual effect spaces \( \mathcal{H}_S \), enabling effect-wise theoretical analysis.
	
	We next introduce the eigensystem of each effect space. For each \( S \in \mathbb{S} \setminus \{\emptyset\} \), let \( \{\mu_{S,v}, \psi_{S,v}\}_{v \in \mathbb{N}} \) denote the eigensystem on \( \mathcal{H}_S \), where \( \mu_{S,v} \geq 0 \) are eigenvalues arranged in nonincreasing order and \( \psi_{S,v} \in \mathcal{H}_S \) are the corresponding eigenfunctions. These eigenpairs are defined with respect to the bilinear forms \( V_S \) and \( J_S \): the eigenfunctions are orthonormal under \( V_S \), and their \( J_S \)-inner products are determined by the inverse eigenvalues. For the intercept space \( \mathcal{H}_\emptyset \), the eigensystem is trivial: \( \{\mu_{\emptyset,0}, \psi_{\emptyset,0}\} \) with \( \mu_{\emptyset,0}^{-1} = 0 \) and \( \psi_{\emptyset,0} = 1 \), where $ V_\emptyset(\psi_{\emptyset,0},\psi_{\emptyset,0}) = 1$, $J_\emptyset(\psi_{\emptyset,0},\psi_{\emptyset,0}) =0 = \mu_{\emptyset,0}^{-1}$, and $f_\emptyset=V_\emptyset(f_\emptyset,\psi_{\emptyset,0})\psi_{\emptyset,0}$ for any $f_\emptyset \in \mathcal{H}_\emptyset$. The following lemma formalizes these properties.
	
	\begin{lemma}\label{lem:eigen}
		Let \( \mathcal{C}_{\psi} \in (1, \infty) \) be a constant. For each \( S \in \mathbb{S} \setminus \{\emptyset\} \), the eigensystem \( \{\mu_{S,v}, \psi_{S,v}\}_{v \in \mathbb{N}} \) on \( \mathcal{H}_S \) satisfies \( \sup_{v \in \mathbb{N}} \|\psi_{S,v}\|_{\sup} \leq \mathcal{C}_{\psi} \), and for all \( v, v' \in \mathbb{N} \),
		\[
		V_S(\psi_{S,v}, \psi_{S,v'}) = \delta_{v,v'}, \quad J_S(\psi_{S,v}, \psi_{S,v'}) = \mu_{S,v}^{-1} \delta_{v,v'}.
		\]
		The RK admits the Mercer expansion
		\begin{align}\label{eqn:mercer}
			\mathcal{K}_S(x_S, x'_S) = \sum_{v \in \mathbb{N}} \mu_{S,v} \psi_{S,v}(x_S) \psi_{S,v}(x'_S),
		\end{align}
		and any \( f_S \in \mathcal{H}_S \) has the expansion \( f_S = \sum_{v \in \mathbb{N}} V_S(f_S, \psi_{S,v}) \psi_{S,v} \).
		
		The eigensystem on \( \mathcal{H} \) is the union of the effect-wise eigensystems:
		\begin{align}\label{eqn:eigen_H}
			\{\mu_v, \psi_v\}_{v \in \mathbb{N}} = \{\mu_{\emptyset,0}, \psi_{\emptyset,0}\} \bigcup \Big\{ \bigcup_{S \in \mathbb{S} \setminus \{\emptyset\}} \{\mu_{S,v}, \psi_{S,v}\}_{v \in \mathbb{N}} \Big\},
		\end{align}
		aligned in nonincreasing order of eigenvalues, which yields $\{\mu_1,\psi_1\}=\{\mu_{\emptyset,0},\psi_{\emptyset,0}\}$. For all \( v, v' \in \mathbb{N} \), \( V(\psi_v, \psi_{v'}) = \delta_{v,v'} \) and \( J(\psi_v, \psi_{v'}) = \mu_v^{-1} \delta_{v,v'} \), and any \( f \in \mathcal{H} \) has the expansion \( f = \sum_{v \in \mathbb{N}} V(f, \psi_v) \psi_v \).
	\end{lemma}
	
	The proof is provided in Section~\ref{sec:Lemma_cor_proof} of the Supplementary Material. The Mercer expansion \eqref{eqn:mercer} is a standard result in RKHS theory. The union structure \eqref{eqn:eigen_H} is a direct consequence of the orthogonality in Lemma~\ref{lem:ortho}: because distinct effect spaces are orthogonal under both \( V \) and \( J \), the eigenfunctions from different spaces remain orthogonal when combined. This structure is crucial for our effect-wise inference, as it allows the theoretical analysis of each effect \( \func_S^\ast \) to proceed using only the eigensystem \( \{\mu_{S,v}, \psi_{S,v}\}_{v \in \mathbb{N}} \) of its own space \( \mathcal{H}_S \).
	
	Later in Section~\ref{sec:inference}, we show the convergence rates of our estimators depend on the asymptotic behavior of eigenvalue-based summations as the tuning parameter \( \lambda \) tends to zero. Define \( S_{\sup} \equiv \argmax_{S \in \mathbb{S}} |S| \) as the highest-order effect in the model.
	
	\begin{lemma}\label{lemma:eigen_order}
		As \( \lambda \rightarrow 0 \),
		\begin{align}
			\sum_{v \in \mathbb{N}} \frac{\lambda/\mu_v}{(1 + \lambda/\mu_v)^2}, \quad \sum_{v \in \mathbb{N}} \frac{1}{(1 + \lambda/\mu_v)^2}, \quad \sum_{v \in \mathbb{N}} \frac{1}{1 + \lambda/\mu_v} \quad \asymp \quad \lambda^{-1/(2m)} (-\log \lambda)^{|S_{\sup}|-1}.
		\end{align}
		For each \( S \in \mathbb{S} \setminus \{\emptyset\} \),
		\begin{align}
			\sum_{v \in \mathbb{N}} \frac{\lambda/\mu_{S,v}}{(1 + \lambda/\mu_{S,v})^2}, \quad \sum_{v \in \mathbb{N}} \frac{1}{(1 + \lambda/\mu_{S,v})^2}, \quad \sum_{v \in \mathbb{N}} \frac{1}{1 + \lambda/\mu_{S,v}} \quad \asymp \quad \lambda^{-1/(2m)} (-\log \lambda)^{|S|-1}.
		\end{align}
	\end{lemma}
	
	The factor \( \lambda^{-1/(2m)} \) reflects the smoothness order \( m \) of the Sobolev space, while the logarithmic factor \( (-\log \lambda)^{|S|-1} \) captures the complexity associated with effect \( S \). For main effects (\( |S| = 1 \)), there is no logarithmic factor, and the rates match those of univariate smoothing splines. For higher-order interactions, the logarithmic factor increases with \( |S| \), but this degradation is mild compared to the exponential curse of dimensionality that would arise without the functional ANOVA structure. The proof is provided in Section~\ref{sec:Lemma_cor_proof} of the Supplementary Material. 
	
	%To characterize the bias introduced by penalization, 
	We further define the self-adjoint operator \( \mathcal{W}_\lambda: \mathcal{H} \rightarrow \mathcal{H} \) by the relation
	\[
	\langle \mathcal{W}_\lambda f, g \rangle_\lambda = \lambda J(f, g)
	\]
	for any \( f, g \in \mathcal{H} \). Intuitively, \( \mathcal{W}_\lambda f \) captures the bias component in the penalized estimator due to the penalty \( \lambda J(f) \). For each \( S \in \mathbb{S} \), the restriction \( \mathcal{W}_{S,\lambda}: \mathcal{H}_S \rightarrow \mathcal{H}_S \) is defined analogously by \( \langle \mathcal{W}_{S,\lambda} f_S, g_S \rangle_{S,\lambda} = \lambda J_S(f_S, g_S) \) for any \( f_S, g_S \in \mathcal{H}_S \), and satisfies \( \mathcal{W}_\lambda f_S = \mathcal{W}_{S,\lambda} f_S \). The following lemma provides eigensystem representations for the inner products, RKs, and bias operators.
	
	\begin{lemma}\label{lem:operator_eigen}
		For all \( f, g \in \mathcal{H} \) and \( x \in \mathcal{X} \),
		\begin{align}
			\langle f, g \rangle_\lambda &= \sum_{v \in \mathbb{N}} V(f, \psi_v) V(g, \psi_v) (1 + \lambda/\mu_v), \\
			\mathcal{R}_\lambda(x, \cdot) &= \sum_{v \in \mathbb{N}} \frac{\psi_v(x)}{1 + \lambda/\mu_v} \psi_v, \quad \mathcal{W}_\lambda f = \sum_{v \in \mathbb{N}} V(f, \psi_v) \frac{\lambda/\mu_v}{1 + \lambda/\mu_v} \psi_v.
		\end{align}
		For each \( S \in \mathbb{S} \setminus \{\emptyset\} \), for all \( f_S, g_S \in \mathcal{H}_S \) and \( x_S \in \mathcal{X}_S \),
		\begin{align}
			\langle f_S, g_S \rangle_{S,\lambda} &= \sum_{v \in \mathbb{N}} V_S(f_S, \psi_{S,v}) V_S(g_S, \psi_{S,v}) (1 + \lambda/\mu_{S,v}), \\
			\mathcal{R}_{S,\lambda}(x_S, \cdot) &= \sum_{v \in \mathbb{N}} \frac{\psi_{S,v}(x_S)}{1 + \lambda/\mu_{S,v}} \psi_{S,v}, \quad \mathcal{W}_{S,\lambda} f_S = \sum_{v \in \mathbb{N}} V_S(f_S, \psi_{S,v}) \frac{\lambda/\mu_{S,v}}{1 + \lambda/\mu_{S,v}} \psi_{S,v}.
		\end{align}
		For the intercept space, for all $f_\emptyset, g_\emptyset \in \mathcal{H}_\emptyset$, $\langle f_\emptyset, g_\emptyset \rangle_{\emptyset,\lambda} = V_\emptyset(f_\emptyset,\psi_{\emptyset,0})V_\emptyset(g_\emptyset,\psi_{\emptyset,0})(1+\lambda/\mu_{\emptyset,0})$, $$\mathcal{R}_{\emptyset,\lambda}=\frac{\psi_{\emptyset,0}}{1+\lambda/\mu_{\emptyset,0}}\psi_{\emptyset,0} = 1,\quad \mathcal{W}_{\emptyset,\lambda}f_\emptyset=V_\emptyset(f_\emptyset,\psi_{\emptyset,0})\frac{\lambda/\mu_{\emptyset,0}}{1+\lambda/\mu_{\emptyset,0}}\psi_{\emptyset,0}=0.$$
	\end{lemma}
	
	These representations reveal how the tuning parameter \( \lambda \) modulates each eigencomponent. In the RK \( \mathcal{R}_\lambda \), the factor \( (1 + \lambda/\mu_v)^{-1} \) downweights eigenfunctions with small eigenvalues (i.e., high-frequency components), providing regularization. In the bias operator \( \mathcal{W}_\lambda \), the factor \( \lambda/\mu_v (1 + \lambda/\mu_v)^{-1} \) shows that the bias is most pronounced for eigenfunctions with small eigenvalues, where the penalty has the greatest effect. For the intercept, \( \mathcal{W}_{\emptyset,\lambda} f_\emptyset = 0 \) confirms that the intercept incurs no bias, consistent with \( J_\emptyset = 0 \). The proof follows from derivations similar to those in \cite{xing_2024} and \cite{liu_2020}. 
	
	The additive structure of the effect decomposition extends to the bias operator.
	
	\begin{coro}\label{coro:Wdecomp}
		For all \( f = \sum_{S \in \mathbb{S}} f_S \in \mathcal{H} \) with \( f_S \in \mathcal{H}_S \), we have \( \mathcal{W}_\lambda f = \sum_{S \in \mathbb{S}} \mathcal{W}_{S,\lambda} f_S \).
	\end{coro}
	
	The proof is provided in Section~\ref{sec:Lemma_cor_proof} of the Supplementary Material. This decomposition ensures that the bias for each effect function \( \func_S^\ast \) depends only on its own penalty \( J_S \) and can be analyzed separately, which is fundamental to our effect-wise inference framework.
	
	\subsection{Estimation Procedure}\label{sec:estimation}
	
	With the RKHS framework established, we now describe the penalized least squares estimation of \( \func^\ast \). Define the empirical squared loss function by $$\mathscr{L}_{n}(f) = \frac{1}{2n} \sum_{i=1}^n \left( Y_i - f(X_i) \right)^2.$$ The estimator is obtained by minimizing the penalized loss function
	\[
	\mathscr{L}_{n,\lambda}(f) = \mathscr{L}_{n}(f) + \frac{\lambda}{2} J(f),
	\]
	which balances the squared loss against the roughness penalty \( J(f) \). The tuning parameter \( \lambda > 0 \) controls the bias-variance trade-off. The estimated function is
	\begin{align}\label{eqn:opt}
		\hat{\func} = \underset{f \in \mathcal{H}}{\argmin} \, \mathscr{L}_{n,\lambda}(f).
	\end{align}
	
	By the representer theorem \citep{gu2013smoothing}, the solution to \eqref{eqn:opt} has the finite-dimensional representation
	\[
	\hat{\func} = \hat{\func}_\emptyset + \sum_{i=1}^{n} \hat{c}_i \mathcal{K}_J(X_i, \cdot) = \sum_{S \in \mathbb{S}} \hat{\func}_S,
	\]
	where \( \mathcal{K}_J = \sum_{S \in \mathbb{S} \setminus \{\emptyset\}} \mathcal{K}_S \) is the RK associated with the penalty \( J \) as defined in Section~\ref{sec:RKHS}. For each \( S \in \mathbb{S} \setminus \{\emptyset\} \), the estimated effect function is \( \hat{\func}_S = \sum_{i=1}^n \hat{c}_i \mathcal{K}_S(X_{iS}, \cdot) \), where \( X_{iS} = \{X_{i[j]}\}_{j \in S} \). This shows that the estimator \( \hat{\func} \) inherits the additive structure of the model space: just as \( \func^\ast = \sum_{S \in \mathbb{S}} \func^\ast_S \in \mathcal{H} \) with \( \func^\ast_S \in \mathcal{H}_S \), the estimator decomposes as \( \hat{\func} = \sum_{S \in \mathbb{S}} \hat{\func}_S \in \mathcal{H} \) with \( \hat{\func}_S \in \mathcal{H}_S \).
	
	The coefficients \( \hat{\func}_\emptyset \in \mathbb{R} \) and \( \hat{\boldsymbol{c}} = (\hat{c}_1, \dots, \hat{c}_n)^{\top} \in \mathbb{R}^n \) are obtained by solving
	\begin{align}\label{eqn:opt2}
		(\hat{\func}_\emptyset, \hat{\boldsymbol{c}}^{\top})^{\top} = \underset{\func_\emptyset \in \mathbb{R}, \, \boldsymbol{c} \in \mathbb{R}^n}{\argmin} \, \frac{1}{n} \left\| \boldsymbol{y} - \boldsymbol{1} \func_\emptyset - \boldsymbol{\mathcal{K}}_J \boldsymbol{c} \right\|_2^2 + \lambda \boldsymbol{c}^{\top} \boldsymbol{\mathcal{K}}_J \boldsymbol{c},
	\end{align}
	where \( \boldsymbol{y} = (Y_1, \ldots, Y_n)^{\top} \), \( \boldsymbol{1} \) is the \( n \)-vector of ones, \( \boldsymbol{\mathcal{K}}_J = \{ \mathcal{K}_J(X_i, X_{i'}) \}_{i, i' \in \{1, \ldots, n\}} \) is the \( n \times n \) kernel matrix, and \( \|\cdot\|_2 \) denotes the Euclidean norm. The closed-form solution is
	\begin{align}\label{eqn:opt_sol}
		\hat{\func}_\emptyset = \frac{\boldsymbol{1}^{\top}\boldsymbol{y} - \boldsymbol{1}^{\top}\boldsymbol{\mathcal{K}}_J(\boldsymbol{\mathcal{K}}_J + n\lambda\boldsymbol{I})^{-1}\boldsymbol{y}}{n - \boldsymbol{1}^{\top}\boldsymbol{\mathcal{K}}_J(\boldsymbol{\mathcal{K}}_J + n\lambda\boldsymbol{I})^{-1}\boldsymbol{1}}, \quad \hat{\boldsymbol{c}} = (\boldsymbol{\mathcal{K}}_J + n\lambda\boldsymbol{I})^{-1}(\boldsymbol{y} - \boldsymbol{1}\hat{\func}_\emptyset),
	\end{align}
	where \( \boldsymbol{I} \) denotes the \( n \times n \) identity matrix. The fitted values at the observed design points can be expressed as \( (\hat{\func}(X_1), \ldots, \hat{\func}(X_n))^{\top} = \boldsymbol{A}(\lambda) \boldsymbol{y} \), where the smoother matrix is
	\begin{align}
		\boldsymbol{A}(\lambda) = \boldsymbol{\mathcal{K}}_J(\boldsymbol{\mathcal{K}}_J + n\lambda\boldsymbol{I})^{-1} + \frac{\left(\boldsymbol{1} - \boldsymbol{\mathcal{K}}_J(\boldsymbol{\mathcal{K}}_J + n\lambda\boldsymbol{I})^{-1}\boldsymbol{1}\right)\left(\boldsymbol{1} - \boldsymbol{\mathcal{K}}_J(\boldsymbol{\mathcal{K}}_J + n\lambda\boldsymbol{I})^{-1}\boldsymbol{1}\right)^{\top}}{n - \boldsymbol{1}^{\top}\boldsymbol{\mathcal{K}}_J(\boldsymbol{\mathcal{K}}_J + n\lambda\boldsymbol{I})^{-1}\boldsymbol{1}}.
	\end{align}
	The matrix \( \boldsymbol{A}(\lambda) \) is symmetric and plays a central role in tuning parameter selection and variance estimation.
	
	The tuning parameter \( \lambda \) is selected by minimizing the generalized cross-validation (GCV) criterion \citep{gu2013smoothing}:
	\begin{align}
		\mathrm{GCV}(\lambda) = \frac{\boldsymbol{y}^{\top} \left(\boldsymbol{I} - \boldsymbol{A}(\lambda)\right)^2 \boldsymbol{y} / n}{\left(\mathrm{tr}\left(\boldsymbol{I} - \gamma \boldsymbol{A}(\lambda)\right) / n\right)^2},
	\end{align}
	where \( \gamma > 1 \) is a parameter that controls the bias-variance trade-off in tuning parameter selection. Following \cite{gu2013smoothing}, we set \( \gamma = 1.4 \). The parameter \( \gamma \) inflates the effective degrees of freedom in the denominator, which encourages the selection of larger \( \lambda \) values and thus smoother fits. This adjustment helps guard against overfitting, particularly in settings with complex model structures involving multiple interaction effects.

	\section{Effect-Wise Inference in Smoothing Spline ANOVA}\label{sec:inference}
	
	Building on the RKHS framework and orthogonality results established in Section~\ref{sec:model_RKHS}, we now develop the theoretical foundation for effect-wise inference. The key insight is that the orthogonality of effect spaces (Lemma~\ref{lem:ortho}) allows us to analyze each effect function \( \func^\ast_S \) separately, leading to effect-specific convergence rates, confidence intervals, and hypothesis tests. The results are organized as follows: functional Bahadur representation and convergence rates (Section~\ref{sec:bahadur_rates}), local inference for pointwise confidence intervals (Section~\ref{sec:local_inference}), global inference for hypothesis testing (Section~\ref{sec:global_inference}), inference for multiple effects (Section~\ref{sec:multiple_effects}), and implementation of inference (Section~\ref{sec:implement}). 
	
	\subsection{Functional Bahadur Representation and Convergence Rates}\label{sec:bahadur_rates}
	
	The functional Bahadur representation decomposes the estimation error \( \hat{\func}_S - \func^\ast_S \) into interpretable components: a leading stochastic term driven by the noise, a bias term due to penalization, and a negligible remainder. This decomposition is fundamental to both local and global inference, as it reveals how the tuning parameter \( \lambda \) plays the role of the bias-variance trade-off for each effect.
	
	\begin{theorem}\label{thm:fbr} (Effect-Wise Functional Bahadur Representation)\\
		Suppose that Assumption~\ref{asp:smooth} holds. If \( \lambda \) satisfies {\color{black}\iffalse\( \sqrt{n} \lambda = o(1) \),\fi\( \lambda = o(1) \),} \( n^{-1} \lambda^{-1/m} (-\log \lambda)^{2|S_{\sup}|-2} = o(1) \), and \( \sqrt{n}\alpha_n = o(1) \), where
		\begin{align}
			\alpha_{n} &= \beta_n \left( n^{-1/2} \lambda^{-1/(4m)} (-\log\lambda)^{(|S_{\sup}|-1)/2} + \lambda^{1/2}J^{1/2}(\func^\ast) \right), \\
			\beta_n &= n^{-1/2} \lambda^{-1/(2m)}\lambda^{-1/(4m-2\tau-2)} \lambda^{1/(8m^2-4m\tau-4m)} \nonumber \\ 
			&\quad \quad  \cdot (-\log \lambda)^{(|S_{\sup}|-1)(1-1/(4m-2\tau-2))} \left(\log\left(\lambda^{1/(4m)-1/2} (-\log\lambda)^{(1-|S_{\sup}|)/2}\right) \right)^{1/2} \left(\log n\right)^{1/2},
		\end{align} 
		then for each \( S \in \mathbb{S} \), we have 
		\[
		\left\| \hat{\func}_S - \func^\ast_S - \frac{1}{n}\sum_{i=1}^n \epsilon_i \mathcal{R}_{S,\lambda}(X_{iS},\cdot) + \mathcal{W}_{S,\lambda}\func^\ast_S \right\|_{S,\lambda} = \bigO_{\mathbb{P}}(\alpha_{n}).
		\]  
		Here, \( \tau = 0 \) when \( |S_{\sup}| = 1 \), and \( \tau \) can take any value in \( (0, 2m-2) \) when \( |S_{\sup}| > 1 \).
	\end{theorem}
	
	As shown in the representation in Theorem~\ref{thm:fbr}, the term \( n^{-1} \sum_{i=1}^n \epsilon_i \mathcal{R}_{S,\lambda}(X_{iS}, \cdot) \) is the leading stochastic component, representing the effect of noise on estimation; its magnitude determines the variance of the estimator. The term \( \mathcal{W}_{S,\lambda} \func^\ast_S \) is the bias introduced by penalization, as characterized by the bias operator defined in Section~\ref{sec:eigensystem}. The remainder \( \alpha_n \) is asymptotically negligible under the stated conditions on \( \lambda \). The proof is provided in Section~\ref{sec:Tech_detail} of the Supplementary Material. 
	
	The conditions on \( \lambda \) align with those in the existing literature \citep{shang_2013, liu_2020}. {\color{black}\iffalse The condition \( \sqrt{n} \lambda = o(1) \) ensures that the bias term vanishes in the limiting distribution\fi The condition $\lambda=o(1)$ controls the bias term, while the condition} \( n^{-1} \lambda^{-1/m} (-\log \lambda)^{2|S_{\sup}|-2} = o(1) \) controls the variance term. The condition \( \sqrt{n}\alpha_n = o(1) \) guarantees that the remainder is negligible. The following remark identifies the range of smoothness orders \( m \) and tuning parameter rates that satisfy these conditions. 
	
	{\color{black}
		\begin{remark}\label{rmk:order}
			Suppose that Assumption~\ref{asp:smooth} holds. When \( m \geq 2 \), for any fixed \( d \in \mathbb{N} \) with any \( |S_{\sup}| \in \{1, \ldots, d\} \), the conditions \( \lambda = o(1) \), \( n^{-1} \lambda^{-1/m} (-\log \lambda)^{2|S_{\sup}|-2} = o(1) \), and \( \sqrt{n}\alpha_n = o(1) \) hold for \( \lambda \asymp n^{-2m/(2m+1)} \), \( \lambda \asymp n^{-2m/(4m+1)} \), or \( \lambda \asymp n^{-4m/(4m+1)} \), while the condition \( \sqrt{n} \lambda = o(1) \) additionally holds for \( \lambda \asymp n^{-2m/(2m+1)} \) or \( \lambda \asymp n^{-4m/(4m+1)} \).
	\end{remark}}
	
	The {\color{black}three} rates of \( \lambda \) in Remark~\ref{rmk:order} correspond to different inferential goals: \( \lambda \asymp n^{-2m/(2m+1)} \) is optimal for estimation, {\color{black}\( \lambda \asymp n^{-2m/(4m+1)} \) is optimal for estimation under the additional smoothness condition on the effect function, and} \( \lambda \asymp n^{-4m/(4m+1)} \) is optimal for hypothesis testing, as will be shown in Section~\ref{sec:global_inference}. {\color{black}All three choices of $\lambda$} achieve optimal theoretical rates up to logarithmic factors. The requirement \( m \geq 2 \) is mild, as this is standard in smoothing spline applications. {\color{black}Note that the condition \( \sqrt{n} \lambda = o(1) \) is only required for local inference, where it ensures that the bias term vanishes in the limiting distribution; see Section~\ref{sec:local_inference}.} 
	
	The rate of convergence follows directly from the functional Bahadur representation by bounding the stochastic and bias terms.
	
	\begin{theorem}\label{thm:rate_of_conv} (Effect-Wise Rate of Convergence)\\
		{\color{black}Suppose that Assumption~\ref{asp:smooth} holds and that {\color{black}\iffalse\( \sqrt{n} \lambda = o(1) \),\fi\( \lambda = o(1) \),} \( n^{-1} \lambda^{-1/m} (-\log \lambda)^{2|S_{\sup}|-2} = o(1) \), and \( \sqrt{n}\alpha_n = o(1) \). For each \( S \in \mathbb{S} \setminus \{\emptyset\} \), we have \begin{align}\label{eqn:rate_general}
				\|\hat{\func}_S - \func^\ast_S\|_{S,\lambda} = \bigO_{\mathbb{P}}\left( n^{-1/2} \lambda^{-1/(4m)} (-\log\lambda)^{(|S|-1)/2} + \lambda^{1/2} J^{1/2}_S(\func^\ast_S) \right),
			\end{align}
			and if there exists a constant \( \mathcal{C}_{S}^\ast \in (0, \infty) \) such that \begin{align}\label{eqn:supersmooth_in_rate_of_conv}
				\sum_{v \in \mathbb{N}} \mu_{S,v}^{-2} V_S^2(\func^\ast_S, \psi_{S,v}) \leq \mathcal{C}_{S}^\ast,\end{align} we further have \begin{align}\label{eqn:rate_supersmooth}
				\|\hat{\func}_S - \func^\ast_S\|_{S,\lambda} = \bigO_{\mathbb{P}}\left( n^{-1/2} \lambda^{-1/(4m)} (-\log\lambda)^{(|S|-1)/2} + \lambda \right).
			\end{align} For the intercept, we have \( \| \hat{\func}_\emptyset - \func^\ast_\emptyset \|_{\emptyset,\lambda} = \bigO_{\mathbb{P}}\left( n^{-1/2} \right) \).}
	\end{theorem}
	The convergence rate {\color{black}\eqref{eqn:rate_general}} in Theorem~\ref{thm:rate_of_conv} consists of two terms reflecting the bias-variance trade-off. The first term \( n^{-1/2} \lambda^{-1/(4m)} (-\log\lambda)^{(|S|-1)/2} \) arises from the variance of the stochastic component and decreases with \( n \) but increases as \( \lambda \) decreases. The second term \( \lambda^{1/2} J^{1/2}_S(\func^\ast_S) \) represents the bias due to penalization and decreases as \( \lambda \) decreases. Balancing these terms by choosing \( \lambda \asymp n^{-2m/(2m+1)} \) yields the optimal rate of convergence up to a logarithmic factor, provided \( \func^\ast_S \neq 0 \) so that \( J^{1/2}_S(\func^\ast_S) \asymp 1 \). {\color{black}Meanwhile, under the additional condition \eqref{eqn:supersmooth_in_rate_of_conv}, the convergence rate improves to \eqref{eqn:rate_supersmooth}, as similarly observed in \cite{gu2013smoothing}. In this case, balancing the two terms yields \( \lambda \asymp n^{-2m/(4m+1)} \), attaining the optimal rate of convergence up to a logarithmic factor. Condition \eqref{eqn:supersmooth_in_rate_of_conv} requires supersmoothness \citep{gu2013smoothing} on each effect function, which is stronger than Assumption~\ref{asp:smooth} and ensures that the coefficients of \( \func^\ast_S \) in the eigenbasis decay sufficiently fast \citep{shang_2015, gu2013smoothing}.} The proof is provided in Section~\ref{sec:Tech_detail} of the Supplementary Material. We note firstly that the convergence rate depends on \( |S| \) only through the logarithmic factor \( (-\log\lambda)^{(|S|-1)/2} \), so higher-order interactions incur only a mild logarithmic penalty compared to main effects. Second, the intercept \( \hat{\func}_\emptyset \) achieves the parametric rate \( n^{-1/2} \) because it is not penalized (\( J_\emptyset = 0 \)). These results demonstrate that the effect-wise framework avoids the curse of dimensionality that would arise without the functional ANOVA structure.
	
	\subsection{Local Inference: Pointwise Confidence Intervals}\label{sec:local_inference}
	
	The functional Bahadur representation enables local inference at specific covariate values. Based on the decomposition in Theorem~\ref{thm:fbr}, we study the distributional behavior of the leading stochastic term \( n^{-1} \sum_{i=1}^n \epsilon_i \mathcal{R}_{S,\lambda}(X_{iS}, \cdot) \), which yields asymptotic normality for the pointwise estimator \( \hat{\func}_S(x_S) \).
	
	\begin{theorem}\label{thm:local} (Effect-Wise Local Inference)\\
		Suppose that Assumption~\ref{asp:smooth} holds. {\color{black}For each \( S \in \mathbb{S} \setminus \{\emptyset\} \), if $\lambda$ satisfies \( \sqrt{n} \lambda = o(1) \),} \( n^{-1} \lambda^{-1/m} (-\log \lambda)^{2|S_{\sup}|-2} = o(1) \), and \( \sqrt{n}\alpha_n = o(1) \),  {\color{black}and if \iffalse there exists a constant \( \mathcal{C}_{S}^\ast \in (0, \infty) \) such that
			\begin{align}\label{eqn:supersmooth}
				\sum_{v \in \mathbb{N}} \mu_{S,v}^{-2} V_S^2(\func^\ast_S, \psi_{S,v}) \leq \mathcal{C}_{S}^\ast,
			\end{align}\fi condition \eqref{eqn:supersmooth_in_rate_of_conv} in Theorem~\ref{thm:rate_of_conv} is satisfied,}
		then for any fixed \( x_S \in \mathcal{X}_S \) satisfying
		\begin{align}\label{eqn:local_order}
			\sum_{v \in \mathbb{N}} \frac{\psi_{S,v}^2(x_S)}{(1+\lambda/\mu_{S,v})^2} \asymp \lambda^{-1/(2m)} (-\log\lambda)^{|S|-1},\quad \text{we have}
		\end{align}\begin{align}
			\frac{\sqrt{n \lambda^{1/(2m)} (-\log\lambda)^{1-|S|}} \left( \hat{\func}_S(x_S) - \func^\ast_S(x_S) \right)}{\sqrt{\sigma^2 \lambda^{1/(2m)} (-\log\lambda)^{1-|S|} \sum_{v \in \mathbb{N}} \psi_{S,v}^2(x_S) / (1+\lambda/\mu_{S,v})^2}} \dgoto \mathrm{N}(0,1) \quad \text{as } n \rightarrow \infty.
		\end{align}
		For the intercept, {\color{black}if \( \lambda \) satisfies {\color{black}\iffalse\( \sqrt{n} \lambda = o(1) \),\fi\( \lambda = o(1) \),} \( n^{-1} \lambda^{-1/m} (-\log \lambda)^{2|S_{\sup}|-2} = o(1) \), and \( \sqrt{n}\alpha_n = o(1) \),} we have ${\sqrt{n} ( \hat{\func}_\emptyset - \func^\ast_\emptyset )}/{\sigma} \dgoto \mathrm{N}(0,1) \quad \text{as } n \rightarrow \infty.$
	\end{theorem}
	
	The proof is provided in Section~\ref{sec:Tech_detail} of the Supplementary Material.
	
	%{\color{red}{Check the concept of the supersmoothing. I remember we use undersmoothing to remove the bias.}}
	\begin{remark}\label{rmk:local_conditions}
		The conditions in Theorem~\ref{thm:local} have clear interpretations. \iffalse Condition \eqref{eqn:supersmooth} requires supersmoothness \citep{gu2013smoothing} on each effect function, which is stronger than Assumption~\ref{asp:smooth} and ensures that the coefficients of \( \func^\ast_S \) in the eigenbasis decay sufficiently fast \citep{shang_2015, gu2013smoothing}.\fi Combined with the undersmoothing requirement \( \sqrt{n} \lambda = o(1) \), {\color{black}condition \eqref{eqn:supersmooth_in_rate_of_conv}, which imposes supersmoothness on each effect function,} guarantees that the bias term vanishes in the limiting distribution. Condition \eqref{eqn:local_order} is a technical requirement expected to hold for most points in \( \mathcal{X}_S \), since \( \sup_{v \in \mathbb{N}} \|\psi_{S,v}\|_{\sup} \leq \mathcal{C}_\psi \) and \( \sum_{v \in \mathbb{N}} (1 + \lambda/\mu_{S,v})^{-2} \asymp \lambda^{-1/(2m)} (-\log\lambda)^{|S|-1} \) by Lemma~\ref{lemma:eigen_order}.
	\end{remark}
	
	The asymptotic normality in Theorem~\ref{thm:local} allows construction of pointwise confidence intervals. For \( S \in \mathbb{S} \setminus \{\emptyset\} \), an asymptotic \( (1-\alpha) \) confidence interval for \( \func^\ast_S(x_S) \) is
	\[
	\hat{\func}_S(x_S) \pm z_{1-\alpha/2} \sqrt{\frac{\sigma^2}{n} \sum_{v \in \mathbb{N}} \frac{\psi_{S,v}^2(x_S)}{(1+\lambda/\mu_{S,v})^2}},
	\]
	where \( z_{1-\alpha/2} \) is the \( (1-\alpha/2) \) quantile of the standard normal distribution. For the intercept, the confidence interval simplifies to \( \hat{\func}_\emptyset \pm z_{1-\alpha/2} \sigma / \sqrt{n} \), reflecting the parametric rate. The order of the asymptotic variance of \( \hat{\func}_S(x_S) \) depends on \( |S| \) only through a logarithmic factor, confirming that local inference for higher-order interactions incurs only a mild penalty compared to main effects.
	
	\subsection{Global Inference: Hypothesis Testing}\label{sec:global_inference}
	
	We now study global inference to test whether an effect function is identically zero, which determines the significance of each effect in the model. For each \( S \in \mathbb{S} \setminus \{\emptyset\} \), we consider testing
	\[
	\mathrm{H}_{0,S}: \func^\ast_S = 0 \quad \text{versus} \quad \mathrm{H}_{1,S}: \func^\ast_S \in \mathcal{H}_S \setminus \{0\}.
	\]
	We propose a Wald-type test statistic \( \mathcal{T}_{S,\lambda} \) based on \( \|\hat{\func}_S\|_{S,\lambda}^2 \), which measures the distance between the estimated effect function and the null function \( \func^\ast_S = 0 \).
	
	\begin{theorem}\label{thm:global_null} (Effect-Wise Global Inference: Null Distribution)\\
		Suppose that Assumption~\ref{asp:smooth} holds and that {\color{black}\iffalse\( \sqrt{n} \lambda = o(1) \),\fi\( \lambda = o(1) \),} \( n^{-1} \lambda^{-1/m} (-\log \lambda)^{2|S_{\sup}|-2} = o(1) \), and \( \sqrt{n}\alpha_n = o(1) \). For each \( S \in \mathbb{S} \setminus \{\emptyset\} \), when \( \mathrm{H}_{0,S} \) is true,
		\begin{align}
			\mathcal{T}_{S,\lambda} \equiv \frac{n^2 \left( \|\hat{\func}_S\|_{S,\lambda}^2 - \sigma^2 \sum_{v \in \mathbb{N}} (1+\lambda/\mu_{S,v})^{-1} / n \right)}{\sqrt{2 \sigma^4 n(n-1) \sum_{v \in \mathbb{N}} (1+\lambda/\mu_{S,v})^{-2}}} \dgoto \mathrm{N}(0,1) \quad \text{as } n \rightarrow \infty.
		\end{align}
	\end{theorem}
	
	The test statistic \( \mathcal{T}_{S,\lambda} \) is constructed by centering \( \|\hat{\func}_S\|_{S,\lambda}^2 \) at its expected value under the null and standardizing by its asymptotic standard deviation. The centering term \( \sigma^2 \sum_{v \in \mathbb{N}} (1+\lambda/\mu_{S,v})^{-1} / n \) accounts for the noise contribution to the squared norm even when \( \func^\ast_S = 0 \). The proof is provided in Section~\ref{sec:Tech_detail} of the Supplementary Material. 
	
	Based on Theorem~\ref{thm:global_null}, for \( S \in \mathbb{S} \setminus \{\emptyset\} \), we reject \( \mathrm{H}_{0,S} \) at significance level \( \alpha \) if
	\[
	\mathcal{J}_{S,\lambda} \equiv \mathds{1} \left( |\mathcal{T}_{S,\lambda}| \geq z_{1-\alpha/2} \right) = 1.
	\] The following theorem characterizes the power of this test.
	
	\begin{theorem}\label{thm:global_power} (Effect-Wise Global Inference: Power Analysis)\\
		Suppose that Assumption~\ref{asp:smooth} holds and that {\color{black}\iffalse\( \sqrt{n} \lambda = o(1) \),\fi\( \lambda = o(1) \),} \( n^{-1} \lambda^{-1/m} (-\log \lambda)^{2|S_{\sup}|-2} = o(1) \), and \( \sqrt{n}\alpha_n = o(1) \). Fix \( S \in \mathbb{S} \setminus \{\emptyset\} \) and allow \( \func^\ast_S \) to depend on \( n \), denoted \( \func^\ast_{S(n)} \neq 0 \). Define the distinguishable rate
		\[
		\mathcal{D}_{S,\lambda} = n^{-1/2} \lambda^{-1/(8m)} (-\log\lambda)^{(|S|-1)/4} + \lambda^{1/2} J^{1/2}_S(\func^\ast_{S(n)}).
		\]
		For any \( \delta \in (0,1) \), there exist constants \( \mathcal{C}_{S,\delta} \in (0,\infty) \) and \( N_{S,\delta} \in \mathbb{N} \) such that, if \( \|\func^\ast_{S(n)}\|_{S,\lambda} \geq \mathcal{C}_{S,\delta} \mathcal{D}_{S,\lambda} \) for all \( n \geq N_{S,\delta} \), then
		\[
		\mathbb{P}\left( \mathcal{J}_{S,\lambda} = 1 \right) \geq 1 - \delta \quad \text{for all } n \geq N_{S,\delta}.
		\]
	\end{theorem}
	
	The proof is provided in Section~\ref{sec:Tech_detail} of the Supplementary Material. Theorem~\ref{thm:global_power} shows that the test has power approaching \( 1 - \delta \) whenever the signal \( \|\func^\ast_{S(n)}\|_{S,\lambda} \) exceeds a constant multiple of the distinguishable rate \( \mathcal{D}_{S,\lambda} \). When \( J^{1/2}_S(\func^\ast_{S(n)}) \asymp 1 \), choosing \( \lambda \asymp n^{-4m/(4m+1)} \) yields the minimax optimal rate of testing up to a logarithmic factor \citep{ingster1993asymptotically}. This rate differs from the optimal rate for estimation (\( \lambda \asymp n^{-2m/(2m+1)} \)), reflecting different bias-variance trade-offs: the optimal testing rate balances the bias of the estimator against the standard deviation of the test statistic, while the optimal estimation rate balances the squared bias against the variance of the estimator.
	
	\begin{remark}\label{rmk:effect_order}
		The theoretical results in Theorems~\ref{thm:rate_of_conv}--\ref{thm:global_power} share a common feature: the rates depend on the effect order \( |S| \) only through logarithmic factors. Consequently, inference for higher-order interactions incurs only a mild penalty compared to main effects. In particular, for main effects (\( |S| = 1 \)), the rates coincide with those for univariate smoothing splines in \citet{shang_2013}, regardless of the presence of higher-order interactions in the model. This demonstrates that the functional ANOVA structure effectively circumvents the curse of dimensionality.
	\end{remark}
	
	%The proofs of Theorems~\ref{thm:fbr}--\ref{thm:global_power} are provided in Section~\ref{sec:Tech_detail} of the Supplementary Material, with supporting technical details in Lemmas~\ref{lem:support} and \ref{lem:frechet}. The development builds on \cite{shang_2013, shang_2015, liu_2020} for Theorems~\ref{thm:fbr}--\ref{thm:rate_of_conv}, \cite{shang_2013, shang_2015} for Theorem~\ref{thm:local}, and \cite{liu_2020} for Theorems~\ref{thm:global_null}--\ref{thm:global_power}.
	
	\subsection{Inference for Multiple Effects}\label{sec:multiple_effects}
	
	The effect-wise inference framework extends naturally to simultaneous inference on multiple effects. Let \( \emptyset \notin \tilde{\mathbb{S}} \subseteq \mathbb{S} \) be a collection of effects of interest. Theorems~\ref{thm:fbr}--\ref{thm:global_power} generalize by replacing each function and eigensystem-related quantity over a single \( \mathcal{H}_S \) with a summation over all \( S \in \tilde{\mathbb{S}} \), applied inside the norm or square root where appropriate. Rates and constants associated with a specific \( S \) are replaced by those corresponding to the highest-order effect in \( \tilde{\mathbb{S}} \). This extension enables testing whether a group of effects jointly contributes to the model and provides a nonparametric generalization of classical ANOVA.
	
	\subsection{Implementation}\label{sec:implement}
	
	The inference procedures developed in Sections~\ref{sec:local_inference} and \ref{sec:global_inference} involve unknown quantities that must be estimated from data: the error variance \( \sigma^2 \) and the eigensystem \( \{\mu_{S,v}, \psi_{S,v}\}_{v \in \mathbb{N}} \) for each \( S \in \mathbb{S} \setminus \{\emptyset\} \). We describe their estimation and the resulting computational procedures.
	
	The error variance \( \sigma^2 \) is estimated and substituted by
	\[
	\hat{\sigma}^2 = \frac{\boldsymbol{y}^\top \left(\boldsymbol{I} - \boldsymbol{A}(\lambda)\right)^2 \boldsymbol{y}}{\mathrm{tr}\left(\boldsymbol{I} - \boldsymbol{A}(\lambda)\right)},
	\]
	following \cite{gu2013smoothing}, where \( \boldsymbol{A}(\lambda) \) is the smoother matrix defined in Section~\ref{sec:estimation}. This estimator adjusts the residual sum of squares by the effective degrees of freedom.
	
	The eigensystem \( \{\mu_{S,v}, \psi_{S,v}\}_{v \in \mathbb{N}} \) is estimated via spectral decomposition of the empirical kernel matrix. For each \( S \in \mathbb{S} \setminus \{\emptyset\} \), let \( \boldsymbol{\mathcal{K}}_S = \{ \mathcal{K}_S(X_{iS}, X_{i'S}) \}_{i, i' \in \{1, \ldots, n\}} \) be the \( n \times n \) kernel matrix. Based on Mercer's theorem \eqref{eqn:mercer}, we compute the spectral decomposition
	\begin{align}
		\frac{1}{n} \boldsymbol{\mathcal{K}}_S = \sum_{v=1}^n \hat{\mu}_{S,v} \left( \frac{1}{\sqrt{n}} \hat{\boldsymbol{\psi}}_{S,v} \right) \left( \frac{1}{\sqrt{n}} \hat{\boldsymbol{\psi}}_{S,v} \right)^{\top},
	\end{align}
	where \( \hat{\mu}_{S,v} \) are the empirical eigenvalues and \( \hat{\boldsymbol{\psi}}_{S,v} = (\hat{\psi}_{S,v}(X_{1S}), \ldots, \hat{\psi}_{S,v}(X_{nS}))^{\top} \) are the empirical eigenvectors. The normalization ensures that
	\[
	\frac{1}{n} \sum_{i=1}^n \hat{\psi}_{S,v}(X_{iS}) \hat{\psi}_{S,v'}(X_{iS}) = \delta_{v,v'},
	\]
	which serves as the empirical analogue of \( V_S(\psi_{S,v}, \psi_{S,v'}) = \delta_{v,v'} \). Additionally,
	\[
	\sum_{v=1}^n \hat{\mu}_{S,v} = \mathrm{tr}\left( \frac{1}{n} \boldsymbol{\mathcal{K}}_S \right) = \frac{1}{n} \sum_{i=1}^n \mathcal{K}_S(X_{iS}, X_{iS})
	\]
	is the sample mean analogue of \( \mathbb{E}_{X_S}( \mathcal{K}_S(X_S, X_S) ) = \sum_{v \in \mathbb{N}} \mu_{S,v} \). The empirical eigensystem \( \{\hat{\mu}_{S,v}, \hat{\boldsymbol{\psi}}_{S,v}\}_{v=1}^n \) is then substituted into the eigensystem-based summations in the test statistics, with the infinite sums truncated at \( v = n \).
	
	For local inference, the confidence intervals in Theorem~\ref{thm:local} are computed at the observed points \( \{X_{iS}\}_{i=1}^n \) using the empirical eigensystem. Extension to unobserved locations can be achieved via interpolation or smoothing if necessary.
	
	For global inference, the test statistic \( \mathcal{T}_{S,\lambda} \) in Theorem~\ref{thm:global_null} requires computing \( \|\hat{\func}_S\|_{S,\lambda}^2 \). By definition of the \( \lambda \)-weighted norm,
	\[
	\|\hat{\func}_S\|_{S,\lambda}^2 = (V_S + \lambda J_S)(\hat{\func}_S) = \int_{\mathcal{X}_S} \hat{\func}_S^2(x_S) \, dx_S + \lambda J_S(\hat{\func}_S) = \int_{\mathcal{X}_S} \hat{\func}_S^2(x_S) \, dx_S + \lambda \hat{\boldsymbol{c}}^{\top} \boldsymbol{\mathcal{K}}_S \hat{\boldsymbol{c}},
	\]
	where the second term follows from the representer theorem representation of \( \hat{\func}_S \). The integral in the first term can be computed numerically or via closed-form expressions depending on the kernel structure.
	
	\section{Simulation Studies}\label{sec:simul}
	
	We conduct simulation studies to evaluate the finite-sample performance of the proposed effect-wise inference procedures. The studies have two specific goals: (1) to assess the pointwise confidence intervals in Theorem~\ref{thm:local} by comparing with the Bayesian pointwise confidence intervals of \cite{gu2013smoothing}, and (2) to evaluate the Wald-type test in Theorems~\ref{thm:global_null} and \ref{thm:global_power} by comparing with the \texttt{mgcv} approach \citep{mgcv_2012, wood2017generalized}. The order of presentation follows the theoretical development in Section~\ref{sec:inference}.
	
	We set \( d = 3 \), \( m = 3 \), \( \mathbb{S} = \{\emptyset, \{1\}, \{2\}, \{3\}, \{1,2\}, \{1,3\}, \{2,3\}\} \), \( \epsilon \sim \mathrm{N}(0, \sigma^2 = 1) \), and
	\[
	f^{\ast} = \sum_{S \in \mathbb{S}} f_{S}^{\ast} = \sum_{S \in \mathbb{S}} \rho_{S} g_{S}^{\ast} \in \mathcal{H},
	\]
	where, for each \( S \in \mathbb{S} \), \( g_S^{\ast} \in \mathcal{H}_S \) with \( V_S^{1/2}(g_S^{\ast}) = 0.35 \), and the effect size \( \rho_{S} \) takes values in \( [0, 1] \). The effect functions are
	\begin{align}
		g^\ast_{\emptyset} &= 0.35, \\
		g^\ast_{\{1\}}(x_{[1]}) &= 3.063x_{[1]}^2 - 2.144x_{[1]} + 0.051, \\
		g^\ast_{\{2\}}(x_{[2]}) &= 4.202\exp(-x_{[2]}) - 5.883x_{[2]}^2 + 8.236x_{[2]} - 4.813, \\
		g^\ast_{\{3\}}(x_{[3]}) &= 0.407\log(0.667x_{[3]}^2 - 1.333x_{[3]} + 0.767) + 3.052x_{[3]}^2 - 3.052x_{[3]} + 1.052, \\
		g^\ast_{\{1,2\}}(x_{[1]}, x_{[2]}) &= -11.502x_{[1]}^2 x_{[2]} + 11.502x_{[1]} x_{[2]}^2 + 5.751x_{[1]}^2 - 5.751x_{[2]}^2 \nonumber \\
		&\quad - 3.834x_{[1]} + 3.834x_{[2]}, \\
		g^\ast_{\{1,3\}}(x_{[1]}, x_{[3]}) &= -7.484x_{[1]}^2 x_{[3]} + 8.315x_{[1]} x_{[3]}^2 - 3.881x_{[1]} x_{[3]} + 3.742x_{[1]}^2 \nonumber \\
		&\quad - 4.158x_{[3]}^2 - 0.832x_{[1]} + 4.435x_{[3]} - 0.832, \\
		g^\ast_{\{2,3\}}(x_{[2]}, x_{[3]}) &= -1.353x_{[2]}^2 x_{[3]}^2 - 6.226x_{[2]}^2 x_{[3]} + 8.933x_{[2]} x_{[3]}^2 + 1.805x_{[2]} x_{[3]} \nonumber \\
		&\quad + 3.564x_{[2]}^2 - 4.015x_{[3]}^2 - 3.880x_{[2]} + 1.173x_{[3]} + 0.752,
	\end{align}
	which are illustrated in Figure~\ref{fig:true_components} of the Supplementary Material.
	
	\subsection{Effect-Wise Confidence Intervals}\label{sec:simul_local}
	
	We compare the proposed smoothing spline ANOVA effect-wise confidence interval (\textbf{ssaec}) in Theorem~\ref{thm:local} against the smoothing spline ANOVA effect-wise Bayesian confidence interval (\textbf{ssaebc}) in \citet{gu2013smoothing} for each \( S \in \mathbb{S} \).
	
	The Bayesian approach of \citet{gu2013smoothing} assumes independent prior distributions for \( f^\ast_S \) across \( S \in \mathbb{S} \). The intercept \( f^\ast_\emptyset \) is assigned a diffuse prior \( \mathrm{N}(0, \rho \sigma^2 / (n\lambda)) \) with \( \rho \to \infty \). For \( S \in \mathbb{S} \setminus \{\emptyset\} \), \( f^\ast_S \) is assigned a Gaussian process prior with mean \( 0 \) and covariance \( \sigma^2 \mathcal{K}_S / (n\lambda) \). The posterior distribution of \( f^\ast_\emptyset \) is
	$
	\mathrm{N}( \hat{f}_\emptyset, {\sigma^2}/{(n\lambda \boldsymbol{1}^{\top}(\boldsymbol{\mathcal{K}}_J + n\lambda\boldsymbol{I})^{-1}\boldsymbol{1})} ),
	$ 
	and for \( S \in \mathbb{S} \setminus \{\emptyset\} \) and any fixed \( x_S \in \mathcal{X}_S \), the posterior distribution of \( f^\ast_S(x_S) \) is Gaussian with mean \( \hat{f}_S(x_S) \) and variance
	\[
	\frac{\sigma^2}{n\lambda} \Big( \mathcal{K}_S(x_S, x_S) - \boldsymbol{\mathcal{K}}_S(x_S)^\top \big( (\boldsymbol{\mathcal{K}}_J + n\lambda\boldsymbol{I})^{-1} - \frac{(\boldsymbol{\mathcal{K}}_J + n\lambda\boldsymbol{I})^{-1} \boldsymbol{1}\boldsymbol{1}^\top (\boldsymbol{\mathcal{K}}_J + n\lambda\boldsymbol{I})^{-1}}{\boldsymbol{1}^\top (\boldsymbol{\mathcal{K}}_J + n\lambda\boldsymbol{I})^{-1} \boldsymbol{1}} \big) \boldsymbol{\mathcal{K}}_S(x_S) \Big),
	\]
	where \( \boldsymbol{\mathcal{K}}_S(x_S) = (\mathcal{K}_S(X_{1S}, x_S), \ldots, \mathcal{K}_S(X_{nS}, x_S))^\top \) and \( \sigma^2 \) is replaced by \( \hat{\sigma}^2 \) in practice.
	
	We set \( \rho_S = 1 \) for all \( S \in \mathbb{S} \) and consider sample sizes \( n \in \{250, 500, 750, 1000, 1250, 1500\} \), each with 1,000 replicates, resulting in 6 simulation scenarios. The significance level is \( \alpha = 0.05 \). For each \( S \in \mathbb{S} \), we evaluate interval length and coverage, averaged over grid points defined as \( \{(\ell-1)/99\}_{\ell=1}^{100} \) for main effects and \( \{(\ell-1)/99, (\ell'-1)/99\}_{\ell, \ell' \in \{1,\ldots,100\}} \) for two-factor interactions.
	
	The empirical length is defined as the average of
	$
	\big(\mathrm{Q}_{1-\alpha/2}(\hat{f}_S(x_S)) - \mathrm{Q}_{\alpha/2}(\hat{f}_S(x_S))\big)/2
	$ 
	across grid points, where \( \mathrm{Q}_p(\hat{f}_S(x_S)) \) denotes the \( p \)th empirical quantile of \( \hat{f}_S(x_S) \) over replicates. An estimated length close to this empirical length indicates accurate calibration of the confidence interval.
	
	Before presenting the confidence interval results, we verify that the convergence rates behave as expected. For the main and interaction effects, Figure~\ref{fig_RMISE} presents the root mean integrated squared error (RMISE) of \( \hat{f}_S \) across replicates. The corresponding results for the intercept are provided in Figure~\ref{fig_intercept} of the Supplementary Material. The results show that the intercept converges faster than the main effects, which in turn converge faster than the interaction effects, and all RMISE values decrease as \( n \) increases. These patterns align with Theorem~\ref{thm:rate_of_conv}, confirming that the estimator performs in accordance with the theoretical convergence rates.
	
	\begin{figure}[H]
		\centering
		\includegraphics[scale=0.35]{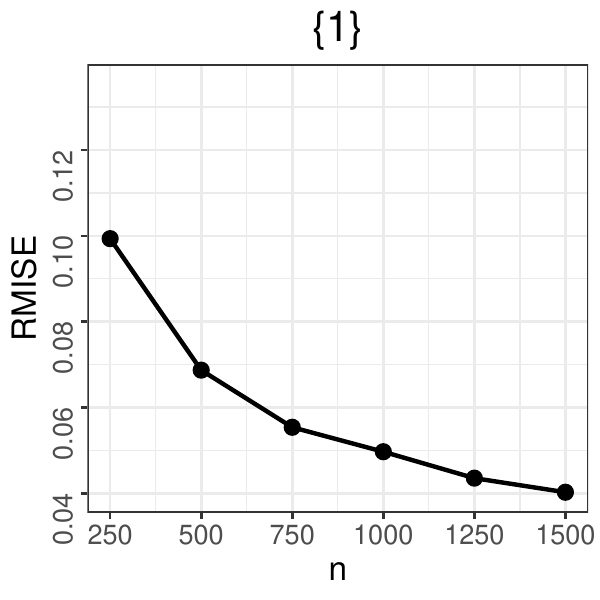}
		\includegraphics[scale=0.35]{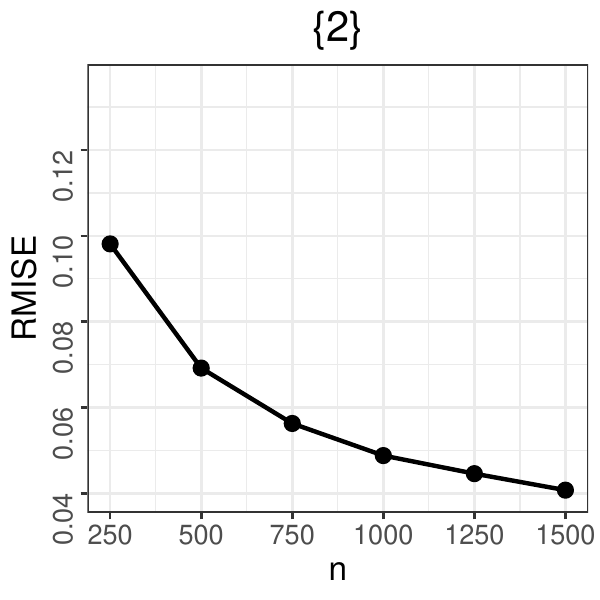}
		\includegraphics[scale=0.35]{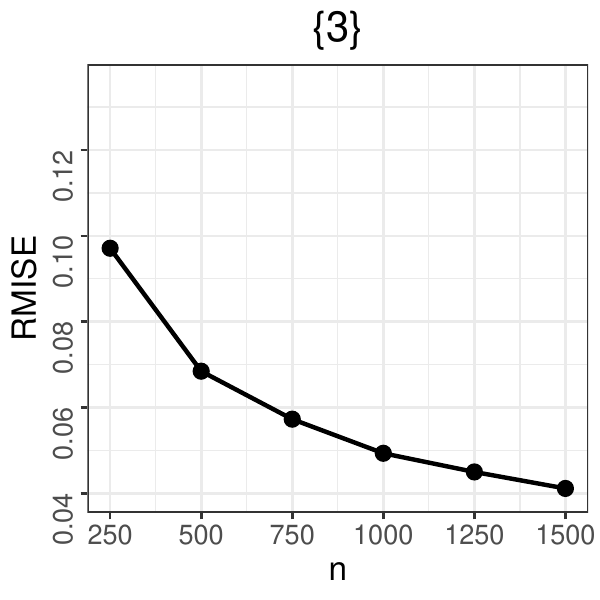}\\
		\includegraphics[scale=0.35]{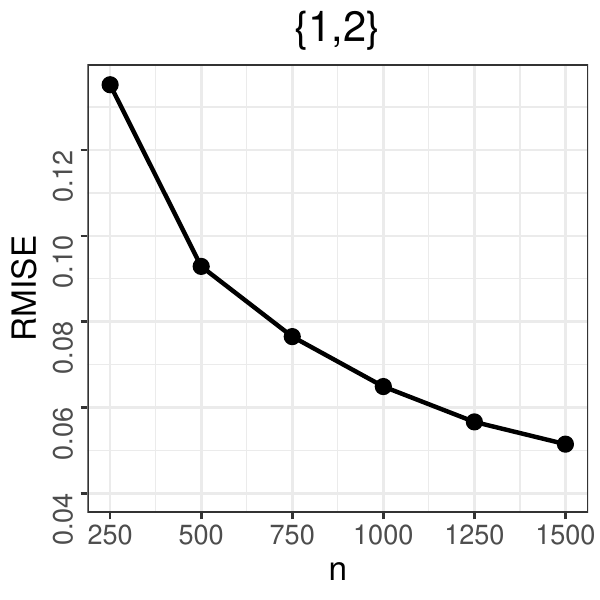}
		\includegraphics[scale=0.35]{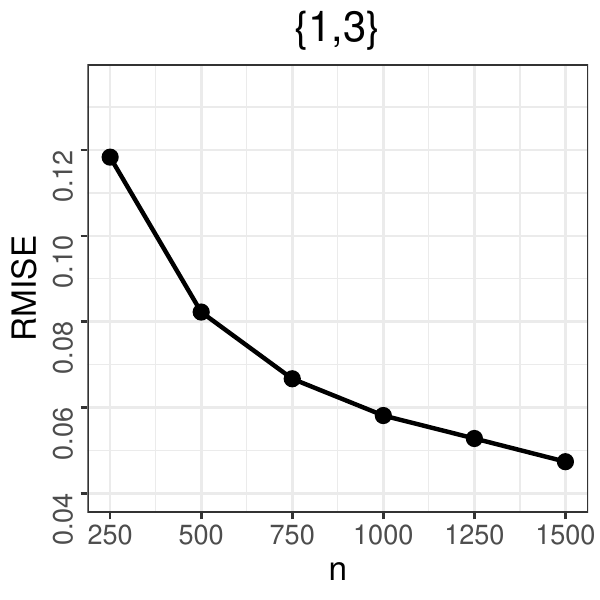}
		\includegraphics[scale=0.35]{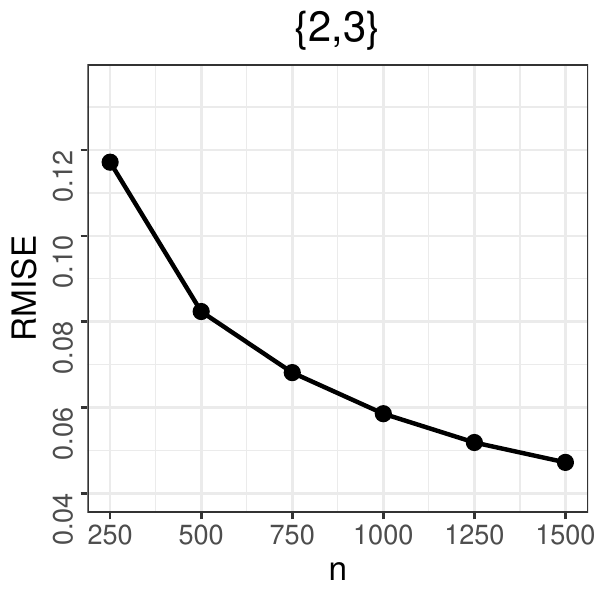}
		\caption{RMISE of effect-wise estimators. Each panel corresponds to a specific effect \( S \).}\label{fig_RMISE}
	\end{figure}
	
	The simulation results for interval length and coverage are presented in Figures~\ref{fig_length} and~\ref{fig_coverage}, respectively, for the main and interaction effects, with the corresponding intercept results reported in Figure~\ref{fig_intercept} of the Supplementary Material. For all effects and sample sizes, \textbf{ssaec} yields shorter intervals than \textbf{ssaebc}. This advantage is more pronounced for interaction effects than for main effects, and for main effects than for the intercept, as well as for smaller \( n \) compared to larger \( n \). As \( n \) increases, the interval length of \textbf{ssaec} approaches the empirical length, suggesting accurate calibration. In contrast, the interval length of \textbf{ssaebc} exceeds the empirical length as \( n \) increases, indicating overly conservative intervals.
	
	Because \textbf{ssaec} produces narrower intervals than \textbf{ssaebc}, its coverage is slightly lower. For both methods, as \( n \) increases, the interval length becomes consistently narrower and the coverage approaches the nominal level \( 1 - \alpha = 0.95 \). These patterns are consistent with the numerical results for univariate smoothing splines in \cite{shang_2013}.
	
	\begin{figure}[H]
		\centering
		\includegraphics[scale=0.35]{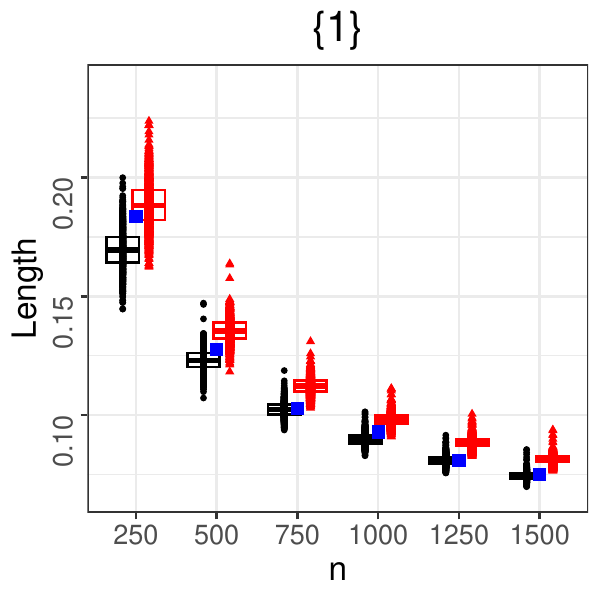}
		\includegraphics[scale=0.35]{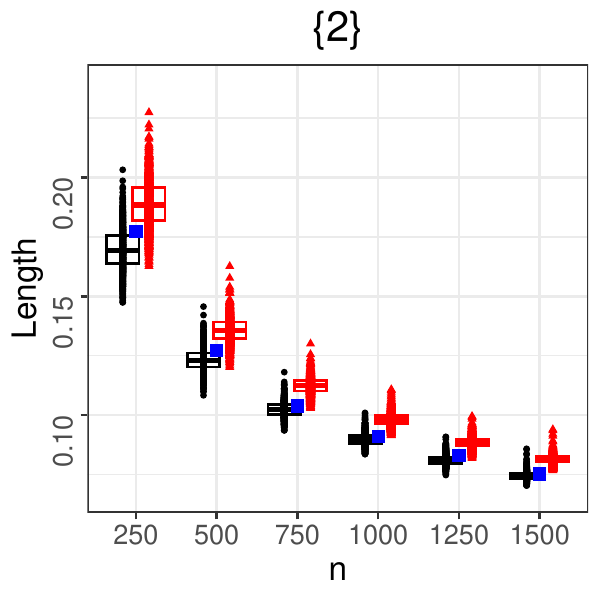}
		\includegraphics[scale=0.35]{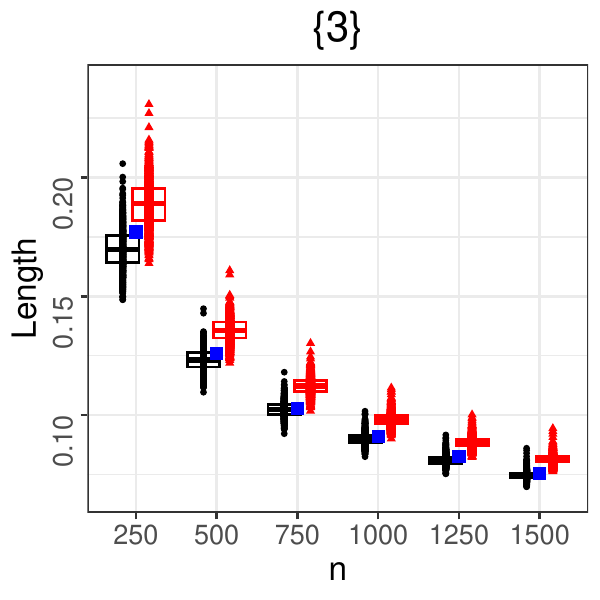}\\
		\includegraphics[scale=0.35]{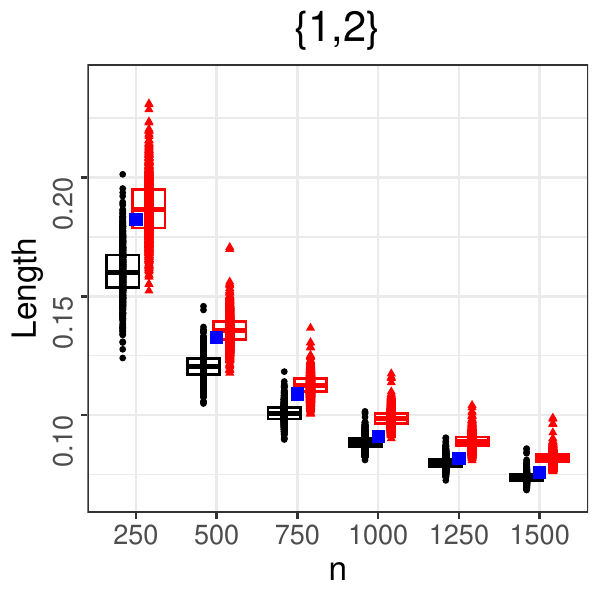}
		\includegraphics[scale=0.35]{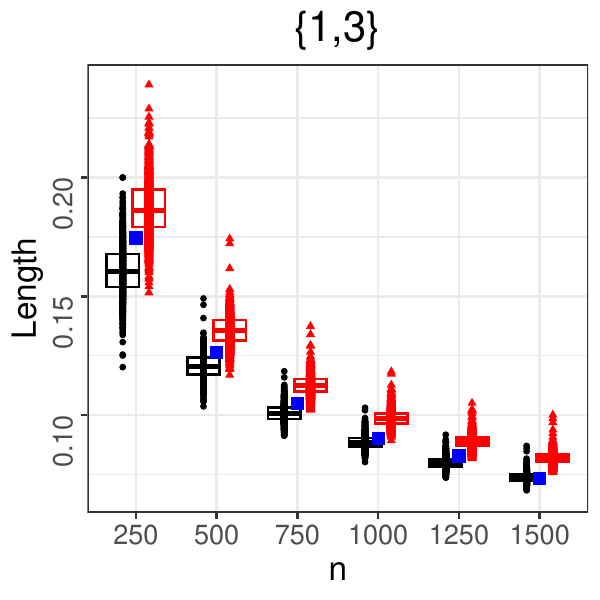}
		\includegraphics[scale=0.35]{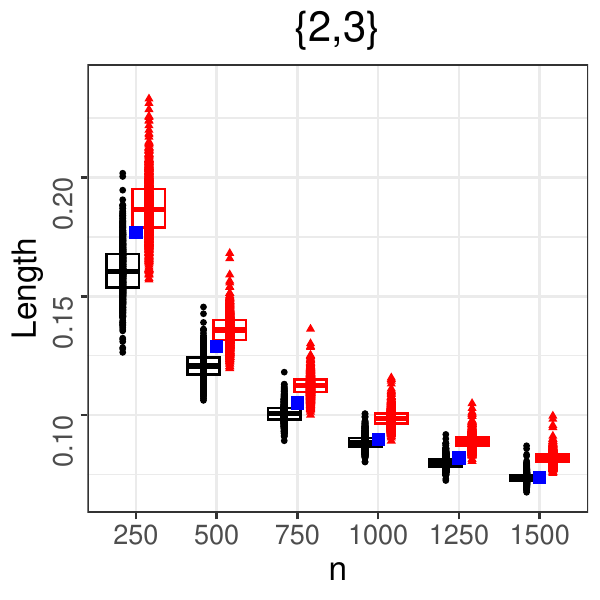}
		\includegraphics[scale=0.7]{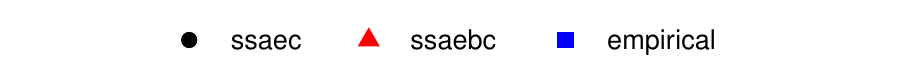}
		\caption{Interval length of effect-wise confidence intervals. Each panel corresponds to a specific effect \( S \). Boxplots display the distribution of interval lengths across replicates for \textbf{ssaec} and \textbf{ssaebc}, and the empirical length is indicated by a dot.}\label{fig_length}
	\end{figure}
	
	\begin{figure}[H]
		\centering
		\includegraphics[scale=0.35]{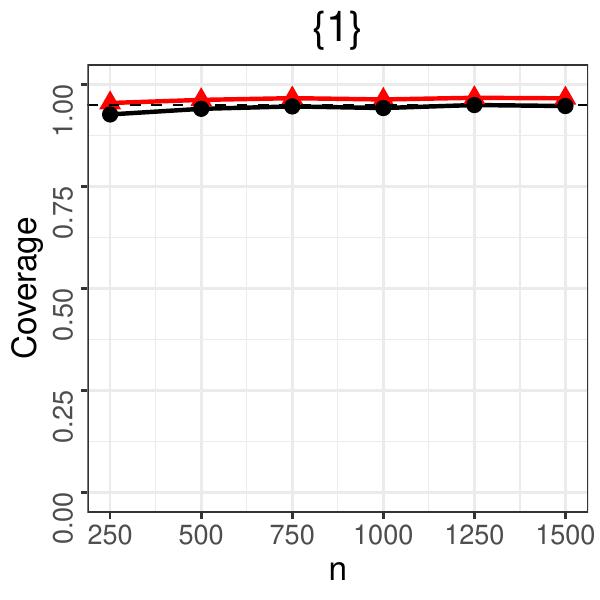}
		\includegraphics[scale=0.35]{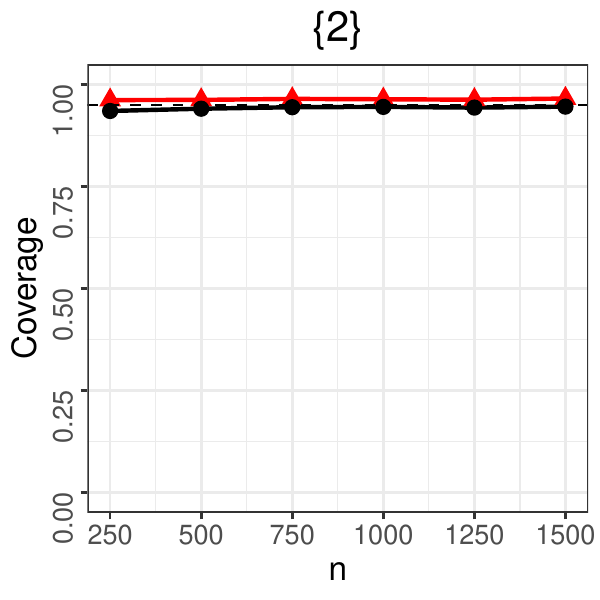}
		\includegraphics[scale=0.35]{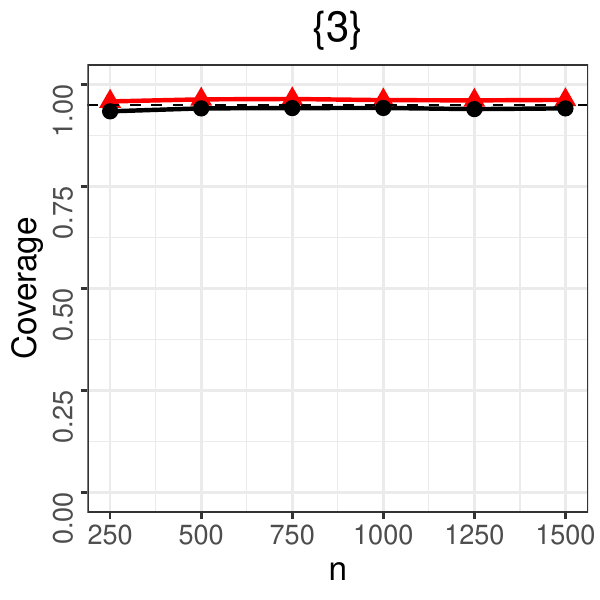}\\
		\includegraphics[scale=0.35]{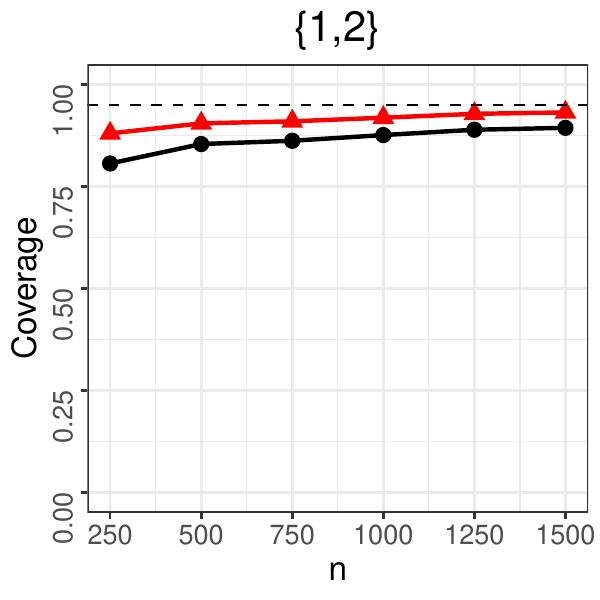}
		\includegraphics[scale=0.35]{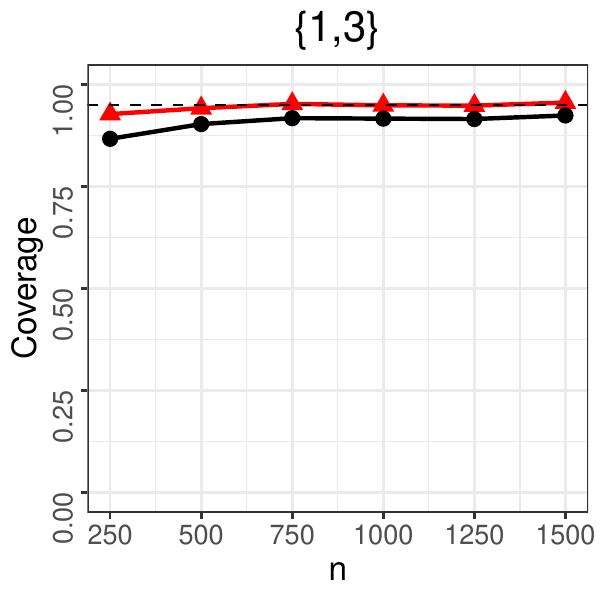}
		\includegraphics[scale=0.35]{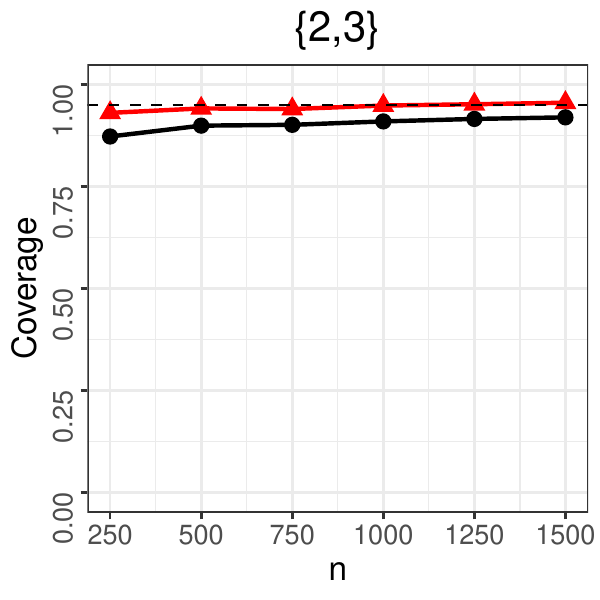}
		\includegraphics[scale=0.7]{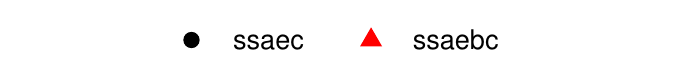}
		\caption{Coverage of effect-wise confidence intervals. Each panel corresponds to a specific effect \( S \). Coverages for \textbf{ssaec} and \textbf{ssaebc} are averaged over replicates. The dashed horizontal line indicates the nominal level \( 1 - \alpha = 0.95 \).}\label{fig_coverage}
	\end{figure}
	
	\subsection{Effect-Wise Wald-Type Tests}\label{sec:simul_global}
	
	We evaluate the performance of the proposed Wald-type test in Theorems~\ref{thm:global_null} and \ref{thm:global_power} for testing \( \mathrm{H}_{0,S}: f^\ast_S = 0 \) for each \( S \in \mathbb{S} \setminus \{\emptyset\} \). We compare the proposed smoothing spline ANOVA effect-wise Wald-type test (\textbf{ssaew}) with effect-wise tests from two functional ANOVA model specifications implemented in the \texttt{mgcv} package in \textsf{R}:
	\begin{itemize}
		\setlength{\itemsep}{3pt}
		\setlength{\parskip}{0pt}
		\setlength{\parsep}{0pt}
		\item[] \textbf{mgcv1}: both main and two-factor interaction effects are specified using \texttt{ti(\(\cdot\))};
		\item[] \textbf{mgcv2}: main effects are specified using \texttt{s(\(\cdot\))} and two-factor interactions using \texttt{ti(\(\cdot\))},
	\end{itemize}
	where both specifications are recommended in \cite{mgcv_2012}.
	
	For each $S \in \mathbb{S}\setminus\{\emptyset\}$, we fix \( \rho_{S'} = 1 \) for all \( S' \neq S \) and vary the target effect size \( \rho_S \in \{0, 0.3, 0.4, 0.5\} \) under sample sizes \( n \in \{250, 500, 750, 1000, 1250, 1500\} \). The case \( \rho_S = 0 \) corresponds to the null hypothesis and is used to evaluate type I error control, while \( \rho_S > 0 \) corresponds to the alternative and is used to evaluate power. Hence, there are \( 6 \times 4 \times 6 = 144 \) simulation scenarios in total, and each scenario is based on 1,000 replicates. The significance level is \( \alpha = 0.05 \).
	
	The simulation results for empirical size (\( \rho_S = 0 \)) are presented in Figure~\ref{fig_size}. The proposed \textbf{ssaew} consistently achieves better type I error control than the \texttt{mgcv}-based competitors across all scenarios. As \( n \) increases, the rejection rate of \textbf{ssaew} approaches the nominal level \( \alpha = 0.05 \), whereas the \texttt{mgcv}-based methods do not exhibit similar improvement and often perform worse with larger \( n \).
	
	\begin{figure}[H]
		\centering
		\includegraphics[scale=0.35]{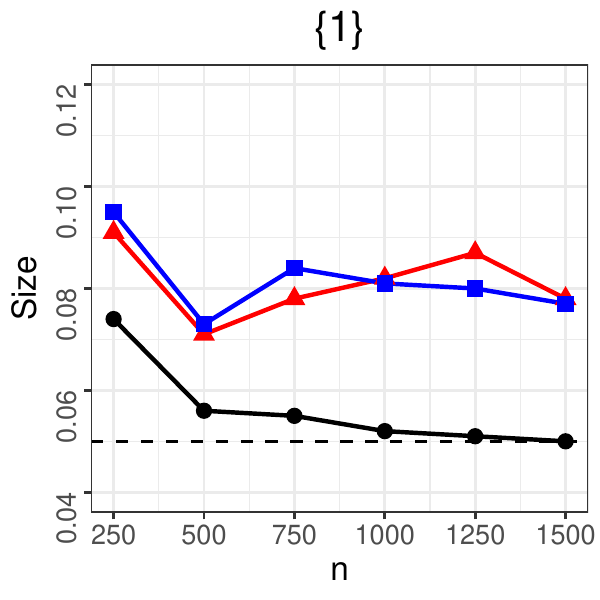}
		\includegraphics[scale=0.35]{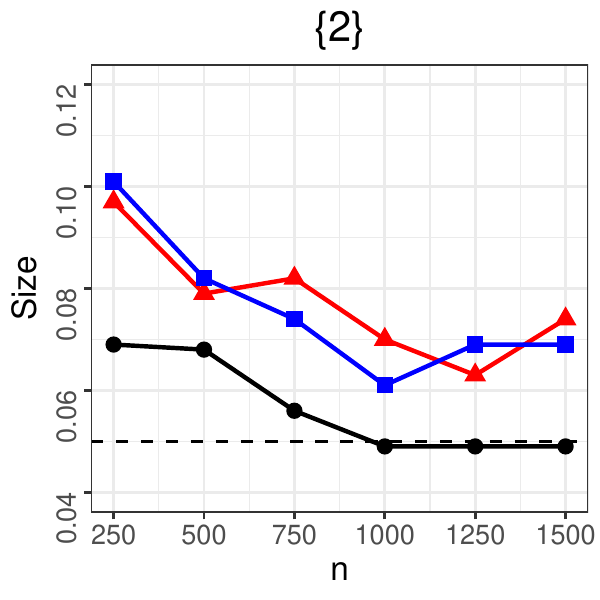}
		\includegraphics[scale=0.35]{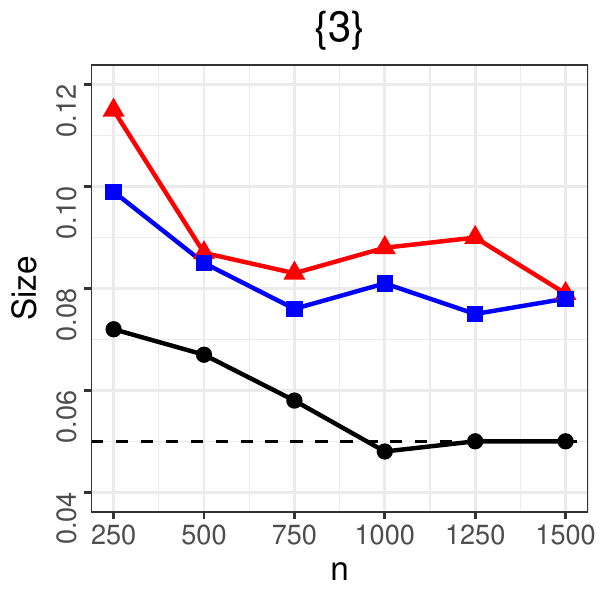}\\
		\includegraphics[scale=0.35]{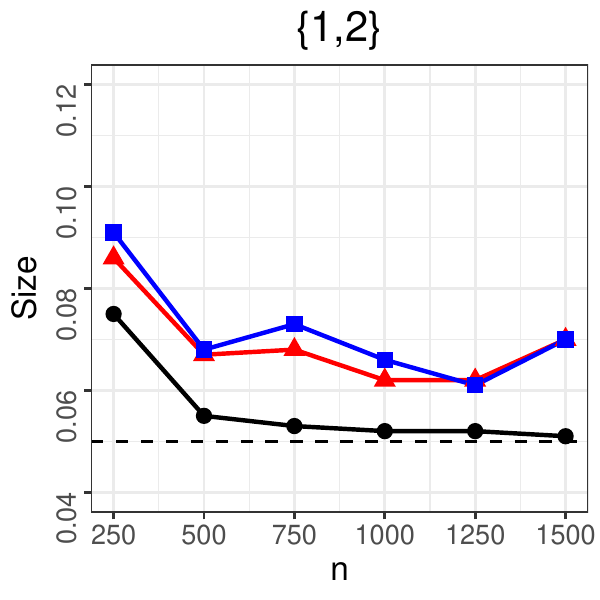}
		\includegraphics[scale=0.35]{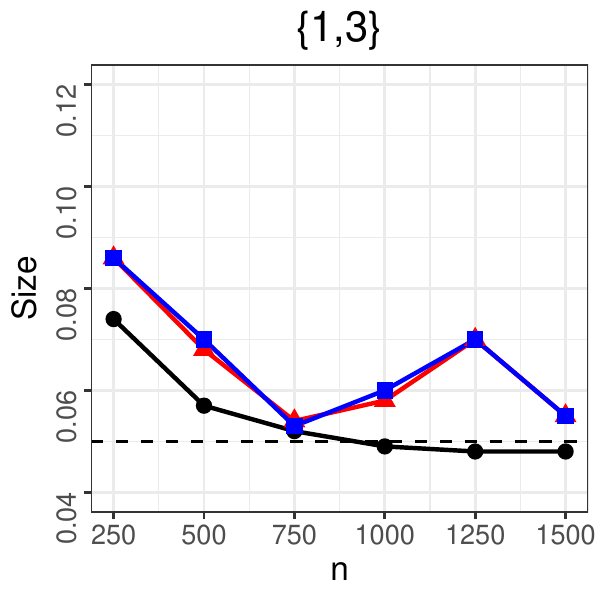}
		\includegraphics[scale=0.35]{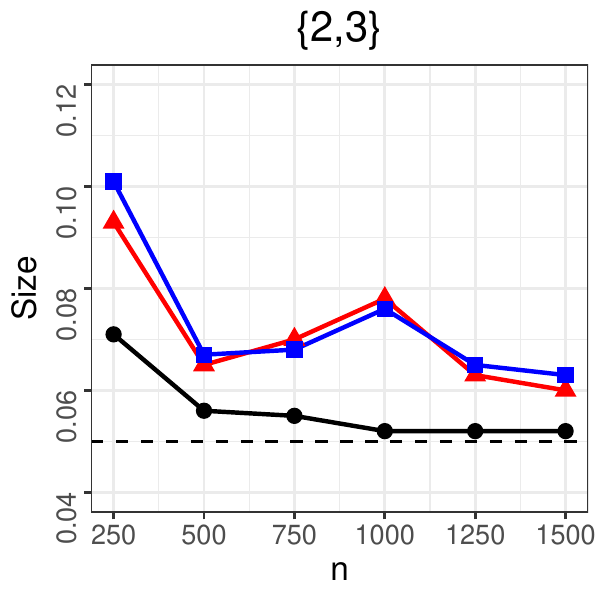}
		\includegraphics[scale=0.7]{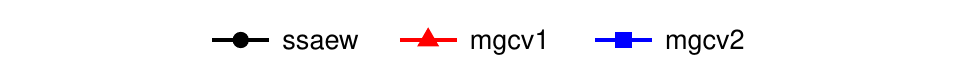}
		\caption{Empirical size of effect-wise Wald-type tests. Each panel corresponds to a specific effect \( S \), representing scenarios for testing \( \mathrm{H}_{0,S}: f^\ast_S = 0 \) with \( \rho_S = 0 \). The dashed horizontal line indicates the significance level \( \alpha = 0.05 \).}\label{fig_size}
	\end{figure}
	
	The simulation results for empirical power (\( \rho_S \in \{0.3, 0.4, 0.5\} \)) are shown in Figure~\ref{fig_power_main} for the main effects and Figure~\ref{fig_power_int} for the interaction effects. The proposed \textbf{ssaew} consistently outperforms the \texttt{mgcv}-based methods across all scenarios. The improvement is most pronounced for effects \( S = \{2\} \) and \( S = \{1,2\} \), and remains moderate yet distinctly present for the other effects. These results demonstrate that the proposed method effectively detects both main and interaction effects across all effect sizes, even with small signals. As either \( \rho_S \) or \( n \) increases, the power improves for all methods, with \textbf{ssaew} maintaining a clear advantage.
	
	\begin{figure}[H]
		\centering
		\includegraphics[scale=0.35]{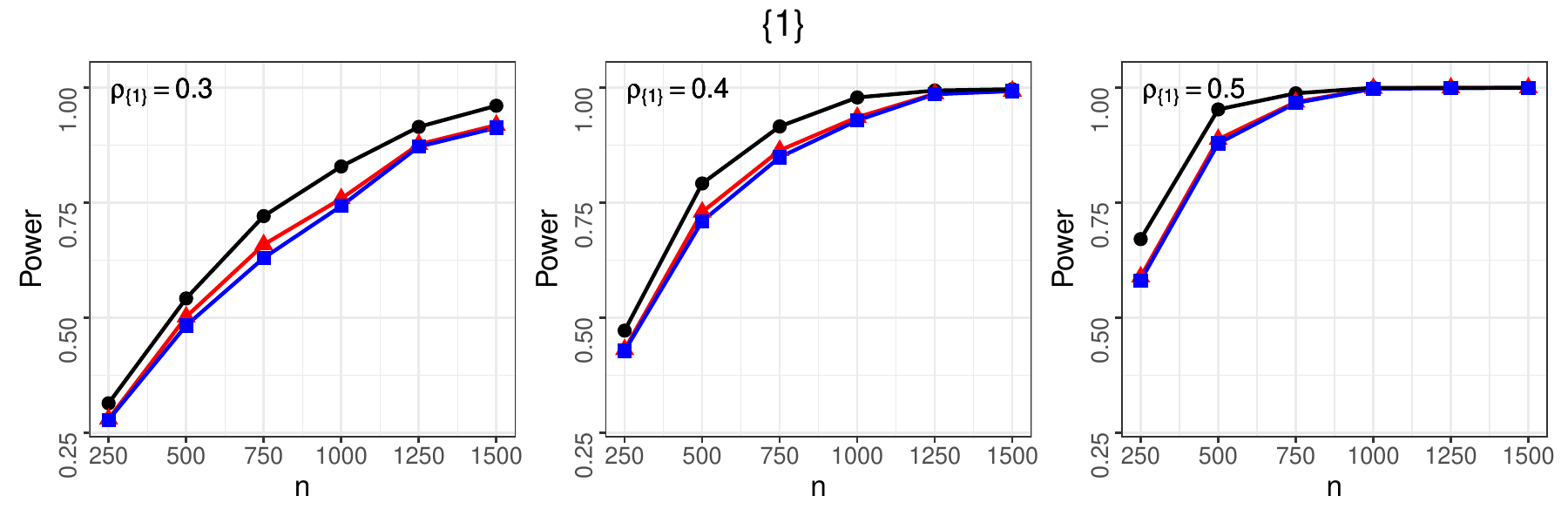}
		\includegraphics[scale=0.35]{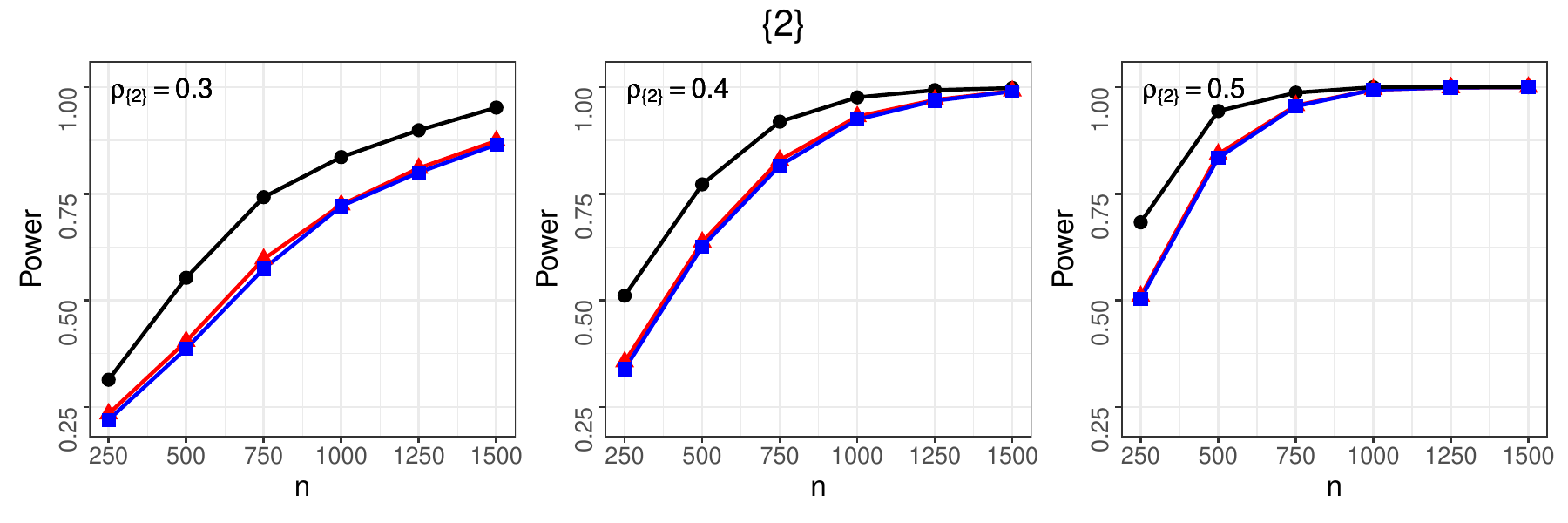}
		\includegraphics[scale=0.35]{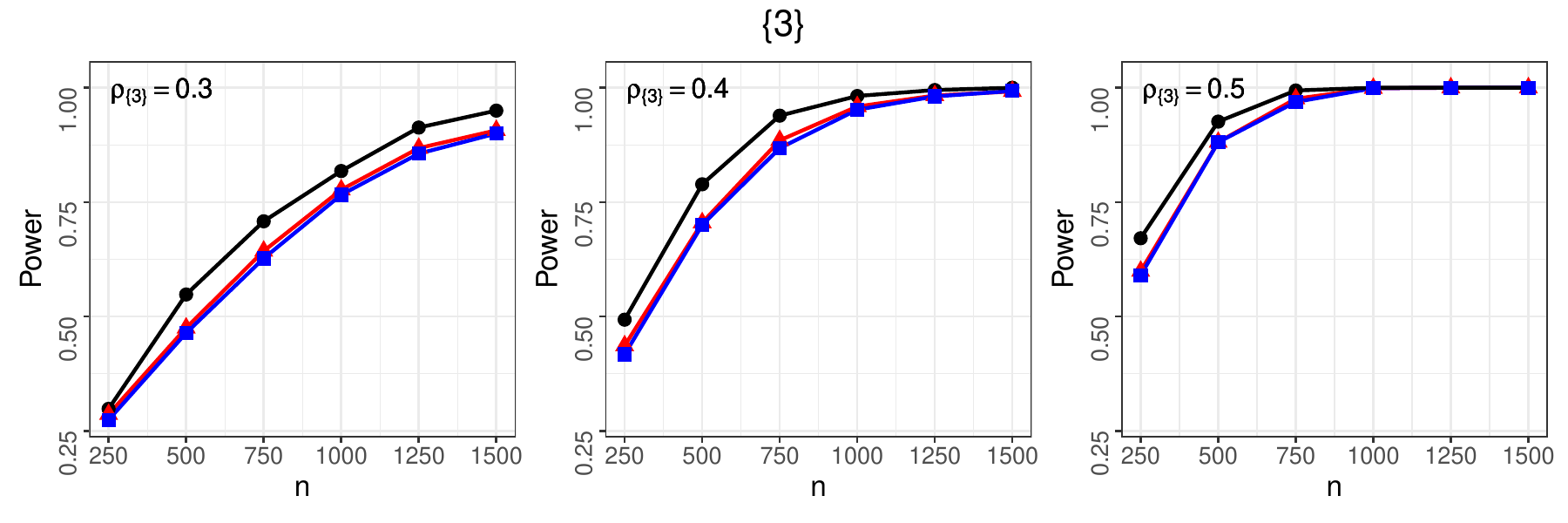}
		\includegraphics[scale=0.7]{Graphics/legend_horizontal.pdf}
		\caption{Empirical power of effect-wise Wald-type tests for main effects. Each row corresponds to a specific main effect \( S \), representing scenarios for testing \( \mathrm{H}_{0,S}: f^\ast_S = 0 \) with \( \rho_S \in \{0.3, 0.4, 0.5\} \).}\label{fig_power_main}
	\end{figure}
	Overall, the proposed method exhibits superior performance in effect-wise global inference, achieving more accurate type I error control and higher power than the \texttt{mgcv}-based approaches.
	\begin{figure}[H]
		\centering
		\includegraphics[scale=0.35]{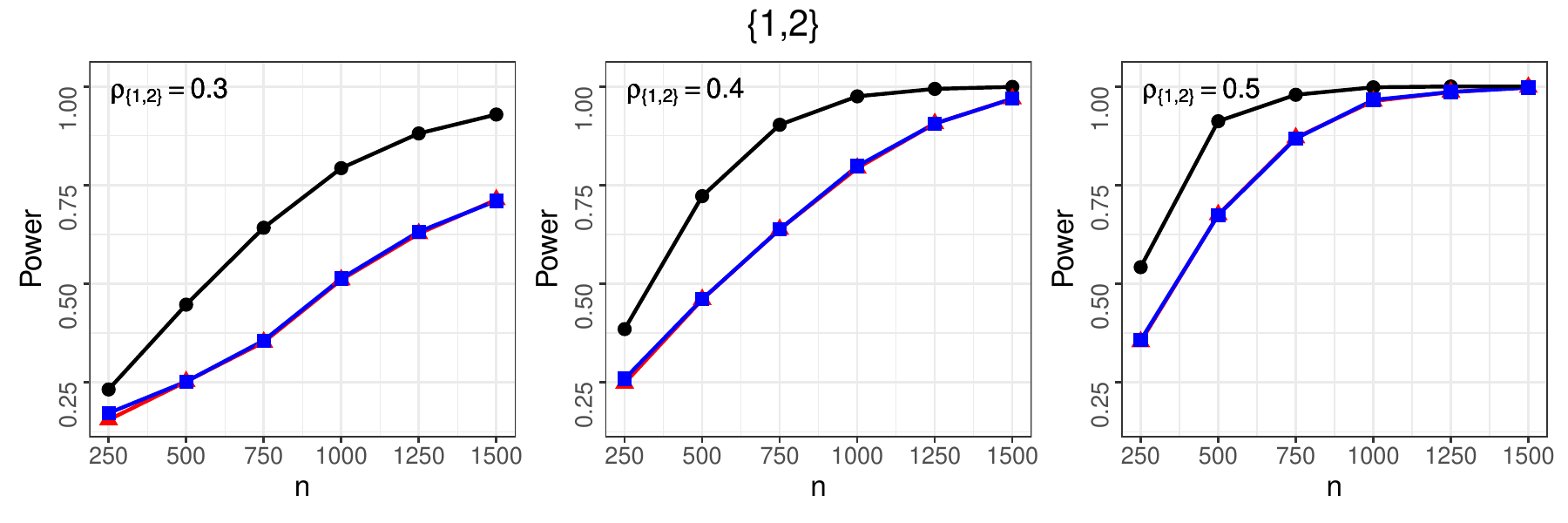}
		\includegraphics[scale=0.35]{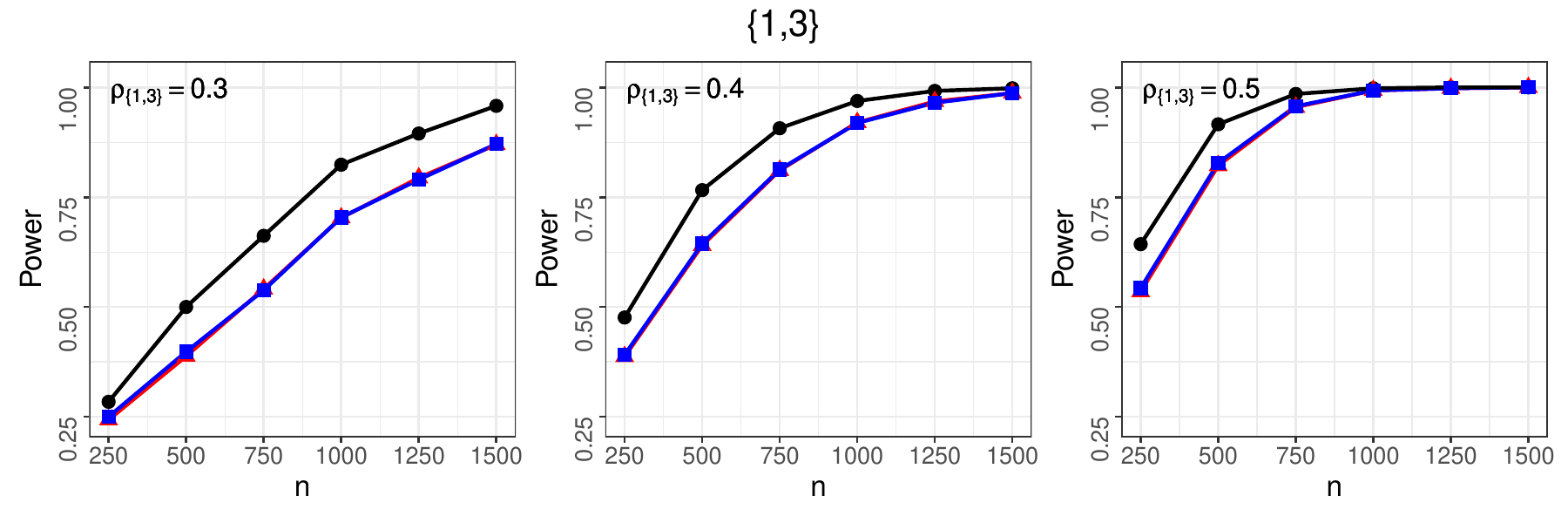}
		\includegraphics[scale=0.35]{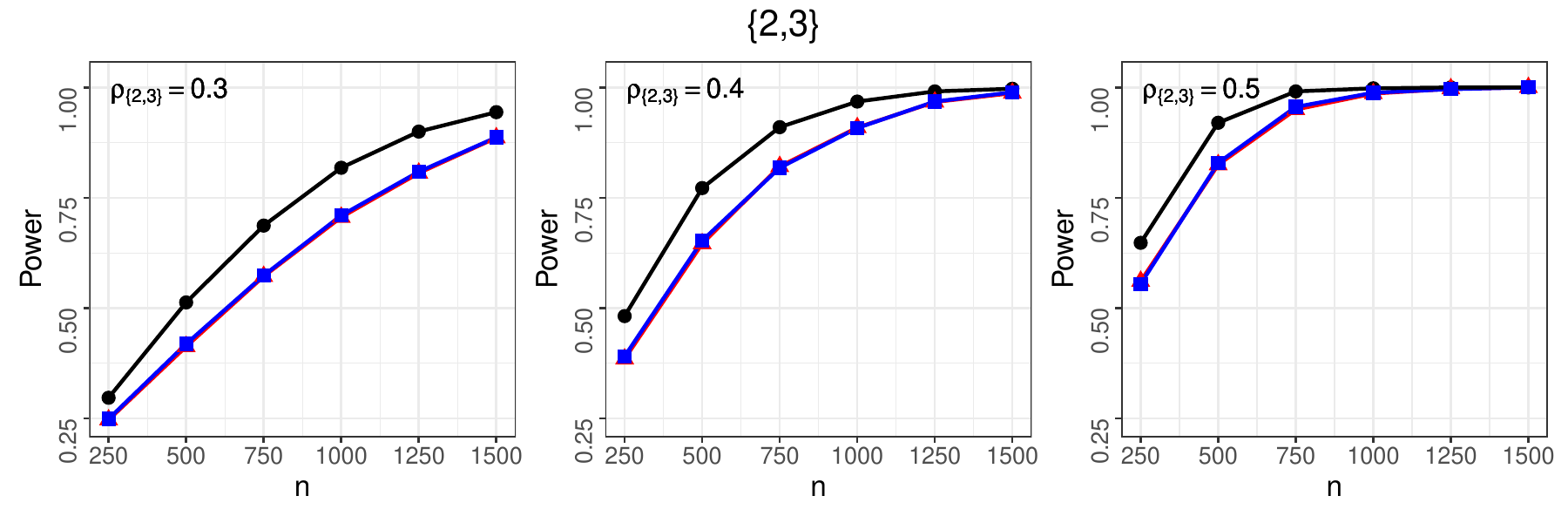}
		\includegraphics[scale=0.7]{Graphics/legend_horizontal.pdf}
		\caption{Empirical power of effect-wise Wald-type tests for two-factor interactions. Each row corresponds to a specific interaction effect \( S \), representing scenarios for testing \( \mathrm{H}_{0,S}: f^\ast_S = 0 \) with \( \rho_S \in \{0.3, 0.4, 0.5\} \).}\label{fig_power_int}
	\end{figure}

	\section{Real Data Application}\label{sec:real}
	
	We apply the proposed effect-wise inference procedures to the Colorado temperature dataset to examine how geographic and temporal variables affect temperature. This analysis demonstrates the practical utility of the methods developed in Section~\ref{sec:inference} and complements the simulation studies in Section~\ref{sec:simul}. Consistent with the simulations, we compare the pointwise confidence interval (\textbf{ssaec}) from Theorem~\ref{thm:local} with the Bayesian pointwise confidence interval (\textbf{ssaebc}), and compare the Wald-type test (\textbf{ssaew}) in Theorems~\ref{thm:global_null} and \ref{thm:global_power} with the \texttt{mgcv} approach (\textbf{mgcv1}, \textbf{mgcv2}).
	
	\subsection{Data Description}
	
	The dataset is obtained from the Global Surface Summary of the Day (GSOD) database, accessed via the \texttt{GSODR} package \citep{GSODR} in R, which provides daily meteorological records worldwide. We extract all observations from the state of Colorado (US) in the year 2020. The response variable \( Y \) is the average daily temperature (in degrees Celsius), and the \( d = 3 \) covariates are:
	\begin{itemize}
		\item[] \( X_{[1]} \): latitude of the recording station, ranging from 37.14 to 40.97;
		\item[] \( X_{[2]} \): longitude of the recording station, ranging from \(-108.97\) to \(-102.27\);
		\item[] \( X_{[3]} \): day of year, taking values from 1 to 366.
	\end{itemize}
	The original dataset contains 24,203 observations. To alleviate computational burden, we randomly select \( n = 500 \) observations for analysis, preserving the distributional pattern and nonlinear relationships in the data. All covariates are scaled to \( [0,1] \) during model fitting and transformed back to their original scales for presentation.
	
	Figure~\ref{fig_COtemp_eda} presents exploratory pairwise plots of temperature and covariates. The covariates appear approximately uniformly distributed over \( [0,1]^d \), consistent with our modeling assumption. The relationships between temperature and each covariate exhibit nonlinearity: most prominently with day of year, showing high temperatures in summer and low temperatures in winter; followed by longitude, showing that the west-central Rocky Mountain region is noticeably cooler than other parts of Colorado; and lastly latitude, which shows a weaker but still nonlinear pattern.
	
	\begin{figure}[H]
		\centering
		\includegraphics[scale=0.35]{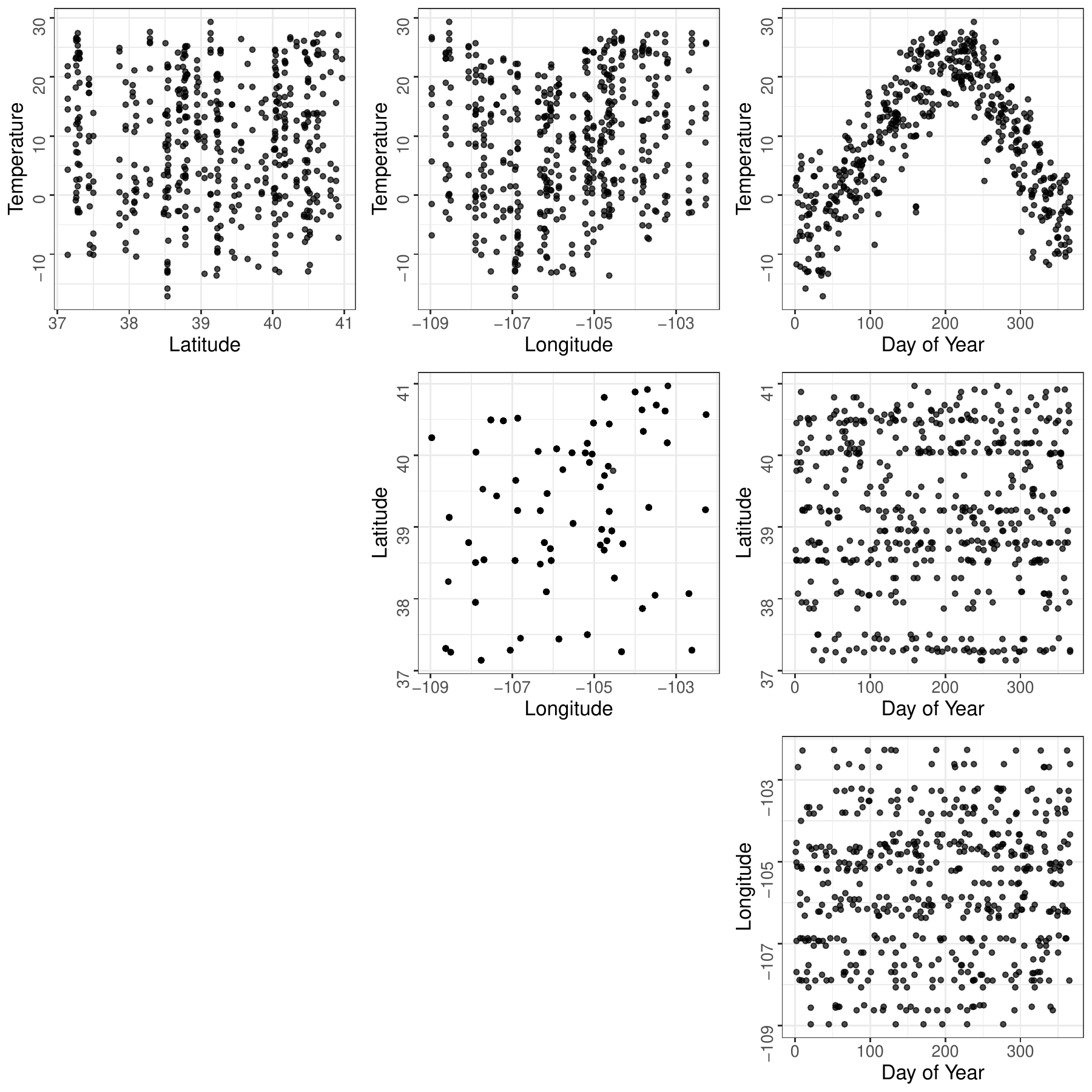}
		\caption{Pairwise plots among temperature, latitude, longitude, and day of year for the Colorado temperature dataset.}\label{fig_COtemp_eda}
	\end{figure}
	
	\subsection{Model Specification and Initial Analysis}
	
	We fit the model with \( m = 3 \) and
	\[
	\mathbb{S} = \{\emptyset, \{1\}, \{2\}, \{3\}, \{1,2\}, \{1,3\}, \{2,3\}\},
	\]
	which includes the intercept, all three main effects, and all two-factor interactions. The regression function is
	\[
	f^{\ast} = \sum_{S \in \mathbb{S}} f_{S}^{\ast} \in \mathcal{H},
	\]
	where \( f_{S}^{\ast} \in \mathcal{H}_{S} \) for each \( S \in \mathbb{S} \). The significance level is \( \alpha = 0.05 \).
	
	Table~\ref{tab:COtemp_first} presents the global inference results from \textbf{ssaew}, \textbf{mgcv1}, and \textbf{mgcv2}. For main effects, \textbf{ssaew} identifies all three as significant, with the ordering of test statistic magnitudes aligning with the exploratory analysis: day of year has the largest \( |\mathcal{T}_{S,\lambda}| \), followed by longitude, then latitude. For interaction effects, \textbf{ssaew} identifies the longitude-day interaction as significant, while the other two interactions are not.
	
	The \texttt{mgcv} methods yield substantially different results. Neither \textbf{mgcv1} nor \textbf{mgcv2} detects the longitude-day interaction as significant; instead, both identify the latitude-longitude interaction as significant. Furthermore, \textbf{mgcv1} fails to identify the latitude main effect as significant, unlike \textbf{ssaew} and \textbf{mgcv2}.
	
	\begin{table}[H]
		\renewcommand{\arraystretch}{0.7}
		\centering
		\caption{Test statistic and p-value (in parentheses) for each effect in the Colorado temperature analysis with full \( \mathbb{S} \). Here, \( 1 \) denotes latitude, \( 2 \) denotes longitude, and \( 3 \) denotes day of year.}\label{tab:COtemp_first}
		\begin{tabular}{lcccccc}
			\hline
			Method & $\{1\}$ & $\{2\}$ & $\{3\}$ & $\{1,2\}$ & $\{1,3\}$ & $\{2,3\}$ \\
			\hline
			\hline
			\multirow{2}{*}{ssaew} &
			4.971 & 56.444 & 684.771 & 1.306 & $-0.376$ & 3.127 \\
			& (0.000) & (0.000) & (0.000) & (0.192) & (0.707) & (0.002) \\
			\hline
			\multirow{2}{*}{mgcv1} &
			-- & -- & -- & -- & -- & -- \\
			& (0.246) & (0.000) & (0.000) & (0.000) & (0.690) & (0.124) \\
			\hline
			\multirow{2}{*}{mgcv2} &
			-- & -- & -- & -- & -- & -- \\
			& (0.000) & (0.000) & (0.000) & (0.000) & (0.544) & (0.279) \\
			\hline
		\end{tabular}
	\end{table}
	
	\subsection{Reduced Model Analysis}
	
	Based on the \textbf{ssaew} results in Table~\ref{tab:COtemp_first}, we re-fit the model with only the significant effects:
	\[
	\mathbb{S} = \{\emptyset, \{1\}, \{2\}, \{3\}, \{2,3\}\}.
	\]
	Table~\ref{tab:COtemp_second} shows that \textbf{ssaew} rejects all effects in the reduced model, with test statistic magnitudes closely mirroring those in Table~\ref{tab:COtemp_first}. In contrast, \textbf{mgcv1} still fails to reject the latitude main effect and the longitude-day interaction, and \textbf{mgcv2} also fails to reject the longitude-day interaction.
	
	\begin{table}[H]
		\renewcommand{\arraystretch}{0.7}
		\centering
		\caption{Test statistic and p-value (in parentheses) for each effect in the Colorado temperature analysis with reduced \( \mathbb{S} \).}\label{tab:COtemp_second}
		\begin{tabular}{lcccc}
			\hline
			Method & $\{1\}$ & $\{2\}$ & $\{3\}$ & $\{2,3\}$ \\
			\hline
			\hline
			\multirow{2}{*}{ssaew} &
			4.450 & 54.901 & 672.323 & 3.076 \\
			& (0.000) & (0.000) & (0.000) & (0.002) \\
			\hline
			\multirow{2}{*}{mgcv1} &
			-- & -- & -- & -- \\
			& (0.340) & (0.000) & (0.000) & (0.062) \\
			\hline
			\multirow{2}{*}{mgcv2} &
			-- & -- & -- & -- \\
			& (0.004) & (0.000) & (0.000) & (0.203) \\
			\hline
		\end{tabular}
	\end{table}
	
	For the reduced model, we obtain \( \hat{\sigma}^2 = 20.526 \) and \( \hat{f}_{\emptyset} = 8.967 \), with confidence interval lengths of 0.397 for \textbf{ssaec} and 0.456 for \textbf{ssaebc}. Both intervals indicate that the intercept is significantly different from zero.
	
	Table~\ref{tab:COtemp_second_lengths} compares the average pointwise confidence interval lengths for each effect. The proposed \textbf{ssaec} produces substantially narrower intervals than \textbf{ssaebc} for all effects, with the difference most pronounced for the longitude-day interaction.
	
	\begin{table}[H]
		\renewcommand{\arraystretch}{0.7}
		\centering
		\caption{Average pointwise confidence interval length for each effect in the Colorado temperature analysis with reduced \( \mathbb{S} \).}\label{tab:COtemp_second_lengths}
		\begin{tabular}{lcccc}
			\hline
			Method & $\{1\}$ & $\{2\}$ & $\{3\}$ & $\{2,3\}$ \\
			\hline
			\hline
			ssaec & 0.888 & 0.890 & 0.852 & 1.021 \\
			\hline
			ssaebc & 1.041 & 1.041 & 0.972 & 1.190 \\
			\hline
		\end{tabular}
	\end{table}
	
	Figure~\ref{fig:COtemp_second_plot} displays the fitted effect functions and their pointwise confidence intervals. Among main effects, day of year shows the strongest signal and is locally significant across most of its domain under \textbf{ssaec}. The fitted effect increases temperature during summer and decreases it during winter, matching the seasonal pattern in Figure~\ref{fig_COtemp_eda}. Longitude exhibits the second strongest effect and is also mostly locally significant; the fitted effect is notably negative in the west-central region of Colorado, indicating colder temperatures there. Latitude shows the weakest main effect, with local significance only in regions corresponding to local extrema of the fitted function.
	
	For the longitude-day interaction, the fitted surface shows a sharply negative regions for western locations near the end of the year. The narrower intervals from \textbf{ssaec} are sufficient to declare this interaction locally significant in these regions, whereas the wider intervals from \textbf{ssaebc} are not.
	
	The results from \textbf{ssaew} and \textbf{ssaec} align closely with the known physical characteristics of Colorado's climate. Western Colorado is colder than the east because the Rocky Mountains occupy the western half of the state at substantially higher elevations. Since elevation changes more dramatically from west to east than from north to south, the east-west temperature difference naturally exceeds the north-south difference, explaining the stronger longitude effect relative to latitude.
	
	Seasonal behavior follows directly from this geography: summer is warmer and winter is colder everywhere, but winter processes such as snowpack, radiative cooling, and inversions amplify the cold in the mountains. This makes western Colorado disproportionately colder during winter, producing the longitude-day interaction detected by \textbf{ssaew} and the corresponding local significance identified by \textbf{ssaec}.
	
	In contrast, \textbf{mgcv1} and \textbf{mgcv2} fail to identify these physically well-established effects. For global inference, both methods miss the longitude-day interaction and instead identify a latitude-longitude interaction that lacks clear physical justification. For local inference, \textbf{ssaebc} produces wider confidence intervals that fail to detect the expected local significance in the longitude-day interaction.
	
	Overall, the proposed \textbf{ssaew} and \textbf{ssaec} procedures yield results that are consistent with the natural temperature structure of Colorado, providing more reliable global and local inference than the competing methods.
	
	\begin{figure}[H]
		\centering
		\begin{subfigure}{0.27\textwidth}
			\centering
			\includegraphics[width=0.85\linewidth,height=0.85\linewidth]{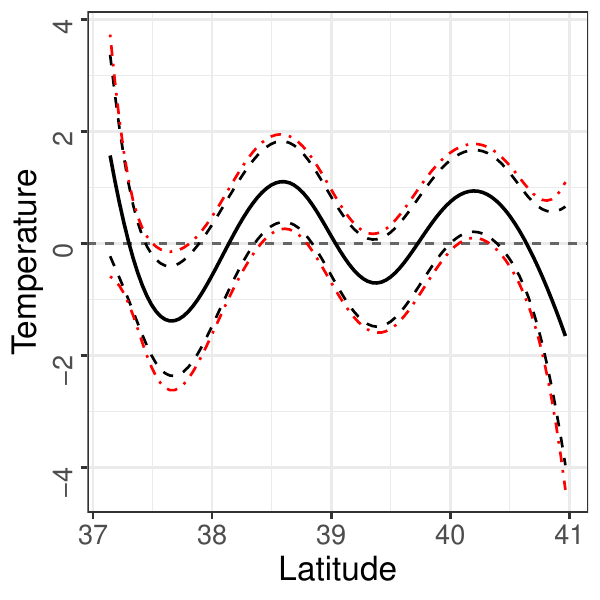}
			\caption{Latitude main effect.}
		\end{subfigure}
		\begin{subfigure}{0.27\textwidth}
			\centering
			\includegraphics[width=0.85\linewidth,height=0.85\linewidth]{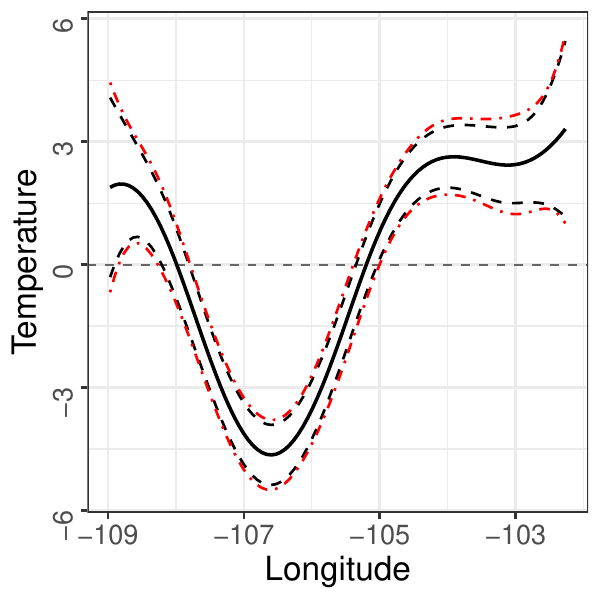}
			\caption{Longitude main effect.}
		\end{subfigure}
		\begin{subfigure}{0.27\textwidth}
			\centering
			\includegraphics[width=0.85\linewidth,height=0.85\linewidth]{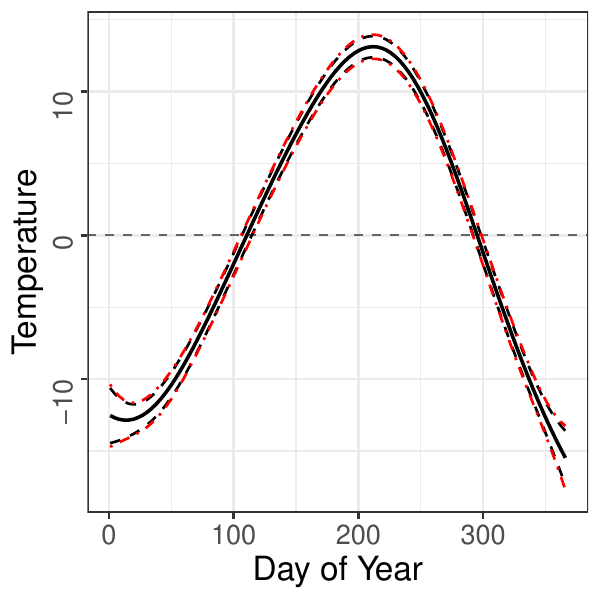}
			\caption{Day main effect.}
		\end{subfigure}
		\begin{subfigure}{0.81\textwidth}
			\centering
			\includegraphics[scale=0.7]{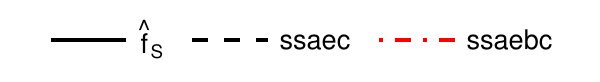}
		\end{subfigure}
		\vspace{0.3ex}
		\begin{subfigure}{0.81\textwidth}
			\centering
			\includegraphics[width=\linewidth,height=0.3\linewidth]{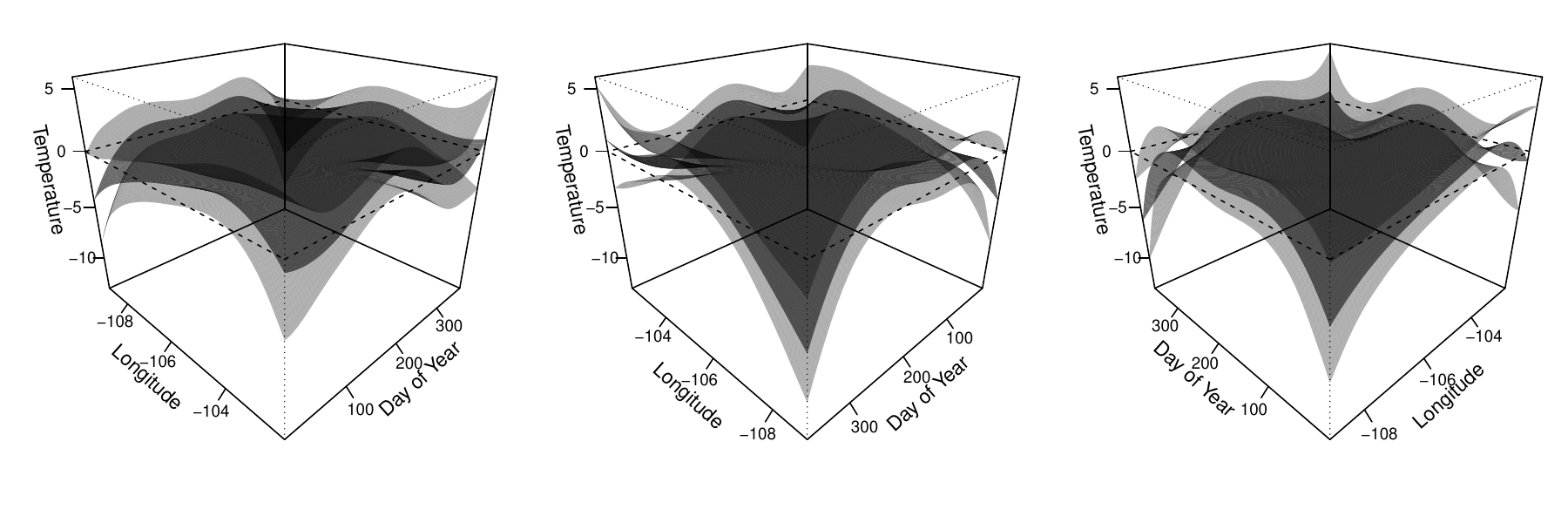}
			\includegraphics[width=\linewidth,height=0.3\linewidth]{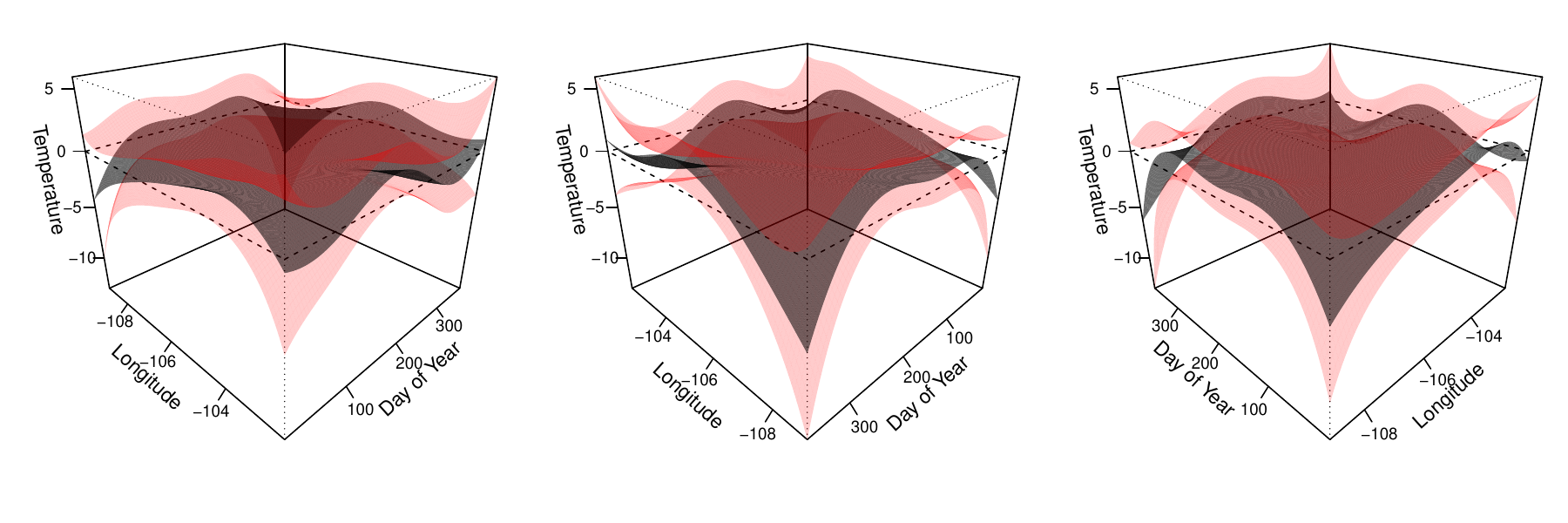}
			\caption{Longitude-day interaction effect.}
		\end{subfigure}
		\caption{Fitted effect functions and pointwise confidence intervals in the Colorado temperature analysis with reduced \( \mathbb{S} \). Panels (a)--(c) show fitted main effects with \textbf{ssaec} and \textbf{ssaebc}. Panel (d) shows the interaction effect: upper row displays \( \hat{f}_S \) (darker surface) with \textbf{ssaec} bounds (lighter surface); lower row displays \( \hat{f}_S \) (darker surface) with \textbf{ssaebc} bounds (lighter surface).}\label{fig:COtemp_second_plot}
	\end{figure}
	
	\section{Conclusion}\label{sec:conclusion}
	
	We developed a unified framework for effect-wise inference in smoothing spline ANOVA on tensor-product Sobolev spaces. The key structural insight is the orthogonality between effect spaces with respect to both \( V \) and \( J \), which enables functional Bahadur representations, convergence rates, pointwise confidence intervals, and Wald-type hypothesis tests for individual main and interaction effects. The theoretical results achieve minimax optimal rates up to logarithmic factors: main effects match classical univariate rates, while interactions incur only mild logarithmic penalties. Simulations demonstrate that the proposed confidence intervals are shorter and better calibrated than Bayesian intervals, and the proposed Wald-type test achieves accurate type I error control and higher power than \texttt{mgcv} alternatives. The Colorado temperature analysis confirms these advantages, with our methods correctly identifying the expected geographic and seasonal effects—including the longitude-day interaction missed by competing approaches.
	
	Several extensions present promising directions for future work. When testing multiple effects simultaneously, formal procedures for controlling family-wise error rate or false discovery rate would provide rigorous multiplicity adjustments. In high-dimensional settings where \( d \) is large, integrating effect selection with inference through sparse functional ANOVA models presents both theoretical and computational challenges. A further direction is functional response models, where the outcome is a curve rather than a scalar, enabling effect-wise inference on coefficient surfaces in function-on-scalar regression. Yet another extension is function-on-function regression, where the bivariate functional coefficient surface admits decomposition into main and interaction effects, allowing effect-wise inference in fully functional data settings.

	\section*{Supplementary Materials}
	
	The following supplementary material is available online.
	
	\begin{description}
		\item[\textbf{Additional details:}] Notations, Technical details, Additional numerical results (pdf file).
	\end{description}
	\par
	%%%%%%%%%%%%%%%%%%%%%%%%%%%%%%%%%%%%%%%%%%%%%%%%%%%%%%%%%%%%%%%%%%%%%%%%%%%%%%%%%%%%%%%%%%%%%%%%%%%%%%%%%%%%%%%%%%%%%%%%%%%%
	\section*{Acknowledgements}
	
	The authors acknowledge the Advanced Research Computing program at Virginia Tech for providing computational resources. 
	\par
	
	\bibliographystyle{Chicago}
	\bibliography{main}
	\newpage
	\spacingset{1}
	\setcounter{page}{1}
	
	% Sections: S1, S2, ...
	\setcounter{section}{0}
	\renewcommand{\thesection}{S\arabic{section}}
	
	% Figures: S1, S2, ...
	\setcounter{figure}{0}
	\renewcommand{\thefigure}{S\arabic{figure}}
	
	% Tables (if needed)
	\setcounter{table}{0}
	\renewcommand{\thetable}{S\arabic{table}}
	
	% Equations (if needed)
	\setcounter{equation}{0}
	\renewcommand{\theequation}{S\arabic{equation}}
	
	% Theorems, Lemmas, Corollaries
	\setcounter{theorem}{0}
	\renewcommand{\thetheorem}{S\arabic{theorem}}
	
	\setcounter{lemma}{0}
	\renewcommand{\thelemma}{S\arabic{lemma}}
	
	\setcounter{coro}{0}
	\renewcommand{\thecoro}{S\arabic{corollary}}
	
	\setcounter{remark}{0}
	\renewcommand{\theremark}{S\arabic{remark}}

	\vspace*{1.5cm}   % ← adjust: 2.5cm–4cm is typical
	
	\begin{center}
		
		{\LARGE
			Supplement to ``Effect-Wise Inference for Smoothing Spline ANOVA on Tensor-Product Sobolev Space''
		}
		
		\vspace{2em}
		
		{\large
			Youngjin Cho\\[1ex]
			Department of Mathematical Sciences\\
			University of Nevada, Las Vegas, Las Vegas, NV 89154}
		
		\vspace{2em}
		
		{\large
			Meimei Liu\\[1ex]
			Department of Statistics\\
			Virginia Tech, Blacksburg, VA 24061}
	\end{center}
	
	\vspace{1em}
	
	\begin{abstract}
		This supplement to ``Effect-Wise Inference for Smoothing Spline ANOVA on Tensor-Product Sobolev Space'' contains additional materials supporting the main results of the paper. 
		Specifically, the supplement includes notations in Section~\ref{sec:notations} and proofs of lemmas and corollaries in Section~\ref{sec:Lemma_cor_proof}. 
		Section~\ref{sec:support_theory} presents supporting theoretical results used in proving the main theorems, while Section~\ref{sec:Tech_detail} contains proofs of the main theorems. Technical details on reproducing kernel Hilbert space and its entropy bound used in proving the main theorems are provided in Section~\ref{sec:detail_RKHS}. Additional results from the simulation studies are reported in Section~\ref{sec:additional_simul}.
	\end{abstract}

	\vfill
	\clearpage
	\newpage
	\spacingset{1.8}
	\section{Notations}\label{sec:notations}
	Here, we introduce notations for our theory. For ease of reference, notations already introduced in the manuscript is summarized and explained in Table~\ref{tab:notations}.
	
	\begin{table}[H]
		\renewcommand{\arraystretch}{0.59}
		\centering
		\caption{Summary of notations introduced in the manuscript.}
		\label{tab:notations}
		\begin{tabular}{l|l}
			\hline
			Notation & Description \\
			\hline
			\hline
			$Y$ 
			& response \\
			$\epsilon$ & random error \\
			$n$ & sample size\\
			$d$ & number of covariates\\
			$\mathbb{S}$ & effects of interest (subset of the power set of $\{1,\ldots,d\}$) \\
			$S$ & element of $\mathbb{S}$ (an effect)\\
			$X_{[j]}$ & $j$th covariate \\
			$\mathcal{X}_{[j]}$ & domain of $X_{[j]}$ \\
			$X_S$ & $=\{X_{[j]}\}_{j\in S}$, collection of covariates indexed by $S$ \\
			$\mathcal{X}_S$ & domain of $X_S$ \\
			$X$ & $=\{X_{[j]}\}_{j=1}^d$, collection of all covariates \\
			$\mathcal{X}$ & domain of $X$ \\
			$\lambda$ & tuning parameter\\
			$id$ & identity operator\\
			$\mathcal{A}_{[j]}$& averaging operator on $\mathcal{X}_{[j]}$  \\
			$\mathcal{H}_S$ & RKHS associated with effect $S$ defined on $\mathcal{X}_S$\\
			$V_S$ 
			&  $\mathcal{L}_2$ inner product defined on $\mathcal{X}_S$  \\
			$J_S$ 
			& penalty term on $\mathcal{H}_S$ ($0$ if $S=\emptyset$, inner product if $S\in\mathbb{S}\setminus\{\emptyset\}$)\\
			$\langle \cdot,\cdot \rangle_{S,\lambda}$, $\|\cdot\|_{S,\lambda}$
			& $\lambda$-weighted inner product and norm on $\mathcal{H}_S$ \\
			$\mathcal{R}_{S,\lambda}$ & $\lambda$-weighted RK on $\mathcal{H}_S$\\
			$\mathcal{W}_{S,\lambda}$ 
			& self-adjoint operator on $\mathcal{H}_S$ satisfying \( \langle \mathcal{W}_{S,\lambda} f_S, g_S \rangle_{S,\lambda} = \lambda J_S(f_S, g_S) \) \\
			$\{\mu_{\emptyset,0}, \psi_{\emptyset,0}\}$ 
			& eigensystem on $\mathcal{H}_\emptyset$ with respect to $V_\emptyset$ and $J_\emptyset$ \\
			$\{\mu_{S,v},\psi_{S,v}\}_{v \in \mathbb{N}}$ 
			& eigensystem on $\mathcal{H}_S$ with respect to $V_S$ and $J_S$, $S \in \mathbb{S}\setminus\{\emptyset\}$ \\
			$\mathcal{H} $
			& $=\oplus_{S \in \mathbb{S}}\mathcal{H}_S$: RKHS for all effects defined on $\mathcal{X}$ \\
			$V$ 
			&  $\mathcal{L}_2$ inner product defined on $\mathcal{X}$  \\
			$J$ 
			&  penalty term on $\mathcal{H}$ (semi-inner product)  \\
			$\langle \cdot,\cdot \rangle_{\lambda}$, $\|\cdot\|_\lambda$
			&  $\lambda$-weighted inner product and norm on $\mathcal{H}$ \\
			$\mathcal{R}_{\lambda}$ & $\lambda$-weighted RK on $\mathcal{H}$\\
			$\mathcal{W}_{\lambda}$ 
			& self-adjoint operator on $\mathcal{H}$ satisfying \( \langle \mathcal{W}_\lambda f, g \rangle_\lambda = \lambda J(f, g) \) \\
			$\{\mu_{v},\psi_{v}\}_{v \in \mathbb{N}}$ 
			& eigensystem on $\mathcal{H}$ with respect to $V$ and $J$ \\
			$f^\ast$ 
			& $=\sum_{S \in \mathbb{S}} f^\ast_S$, true regression function\\
			$\hat{f}$ 
			& $=\sum_{S \in \mathbb{S}} \hat{f}_S$, estimated regression function\\
			$\mathcal{T}_{S,\lambda}$ 
			& Wald-type test statistic for testing $f^\ast_S = 0$, $S \in \mathbb{S}\setminus\{\emptyset\}$\\
			$\mathcal{D}_{S,\lambda}$ 
			& distinguishable rate of $\mathcal{T}_{S,\lambda}$, $S \in \mathbb{S}\setminus\{\emptyset\}$\\
			\hline
		\end{tabular}
	\end{table}

	\newpage
	Let \((\Omega,\mathcal{A},\mathbb{P})\) be the underlying probability space, where \(\Omega\) is the sample space, \(\mathcal{A}\) is the sigma-algebra, and \(\mathbb{P}\) is the probability measure. Note that \((\Omega,\mathcal{A},\mathbb{P})\) is fixed and does not depend on \(n\). The random variables \(Z = \{X,Y,\epsilon\}\) and \(Z_i = \{X_i,Y_i,\epsilon_i\}\) are mappings from \((\Omega,\mathcal{A},\mathbb{P})\) to \((\mathcal{Z},\mathcal{B}(\mathcal{Z}))\), where \(\mathcal{Z} = \mathcal{X} \times \mathbb{R}^2\) and \(\mathcal{B}(\mathcal{Z})\) is the Borel sigma-algebra on \(\mathcal{Z}\).

	Let $\mathrm{P}_Z = \mathbb{P}(Z^{-1}(\cdot))$ denote the distribution of $Z$, and let $\mathrm{P}_{Z^n} = \mathbb{P}(\{Z_1, \ldots, Z_n\}^{-1}(\cdot))$ represent the distribution of $\{Z_i\}_{i=1}^n$. The operator $\mathbb{E}_Z$ is defined to compute the Lebesgue integral with respect to the distribution of either $Z$ or $\{Z_i\}_{i=1}^n$. Specifically, $$\mathbb{E}_Z\left[g(Z)\right] = \int_{\mathcal{Z}} g \, d\mathrm{P}_Z \enspace\text{and}\enspace \mathbb{E}_Z\left[g(\{Z_i\}_{i=1}^n)\right] = \int_{\mathcal{Z}^n} g \, d\mathrm{P}_{Z^n}$$ for an arbitrary function \( g \) with a domain of \( \mathcal{Z} \) or \( \mathcal{Z}^n \). Note that $g$ can be either a deterministic function or a stochastic process. When $g$ is a deterministic function, we have $\mathbb{E}_Z\left[g(Z)\right] = \mathbb{E}\left[g(Z)\right]$ and $\mathbb{E}_Z\left[g(\{Z_i\}_{i=1}^n)\right] = \mathbb{E}\left[g(\{Z_i\}_{i=1}^n)\right]$. 
	
	Similarly, for any $S \in \mathcal{P}_d$, with $\mathrm{P}_{X_S}=\mathbb{P}(X_S^{-1}(\cdot))$ and $\mathrm{P}_{X_S^n} = \mathbb{P}(\{X_{1S}, \ldots, X_{nS}\}^{-1}(\cdot))$, \begin{align}
		\mathbb{E}_{X_S}\left[g_S(X_S)\right] &\equiv \int_{\mathcal{X}_S} g_S \, d\mathrm{P}_{X_S}  = \int_{\mathcal{X}_S} g_S(x_S) \, dx_S \enspace\text{and}\\ \mathbb{E}_{X_S}\left[g_S(\{X_{iS}\}_{i=1}^n)\right] &\equiv \int_{\mathcal{X}_S^n} g_S \, d\mathrm{P}_{X_S^n}=\int_{\mathcal{X}_S}\ldots\int_{\mathcal{X}_S} g_S(x_{1S},\ldots,x_{nS}) \, dx_{1S}\ldots dx_{nS}
	\end{align} for an arbitrary function (deterministic function or stochastic process) \( g_S \) with a domain of \( \mathcal{X}_S \) or \( \mathcal{X}_S^n \). Note that \(\mathbb{E}_X\) is defined as \(\mathbb{E}_{X_S}\) by setting \(S = \{1, \ldots, d\}\). For {\color{black}$f: \mathcal{X} \rightarrow \mathbb{R}$,} let $\mathsf{P}f=\mathbb{E}_X(f(X))$, and $\mathsf{P}_n f = \sum_{i=1}^n f(X_i)/n$.
	
	For any deterministic sequences $a_n, b_n \in \mathbb{R}$ for $n \in \mathbb{N}$, we write $a_n = \bigO(b_n)$ ($|a_n| \lesssim |b_n|$) if there exist a constant $\mathcal{C} \in (0,\infty)$ and $N \in \mathbb{N}$ such that $|a_n| \leq \mathcal{C} |b_n|$ for all $n \geq N$. Similarly, $a_n = o(b_n)$ ($|a_n| \ll |b_n|$) holds if for every constant $\delta\in(0,\infty)$, there exists $N_\delta \in \mathbb{N}$ such that $|a_n| \leq \delta |b_n|$ for all $n \geq N_\delta$. We use $a_n \asymp b_n$ to indicate that $a_n = \bigO(b_n)$ and $b_n = \bigO(a_n)$.
	
	For any sequences $X_n, Y_n \in \mathbb{R}$ for $n\in\mathbb{N}$, where at least one of them is random, we say $X_n = \bigO_{\mathbb{P}}(Y_n)$ ($|X_n| \lesssim |Y_n|$) if for all $\delta \in (0,\infty)$, there exist a constant $\mathcal{C}_\delta \in (0,\infty)$ and $N_\delta \in \mathbb{N}$ such that $\mathbb{P}(|X_n| \leq \mathcal{C}_\delta |Y_n|) \geq 1-\delta$ for all $n \geq N_\delta$. Furthermore, $X_n = o_{\mathbb{P}}(Y_n)$ ($|X_n| \ll |Y_n|$) if $\mathbb{P}(|X_n| \leq \delta |Y_n|) \to 1$ as $n \to \infty$ for any constant $\delta \in (0,\infty)$. 
	
	Let \(\dgoto\), \(\Pgoto\), and \(\asgoto\) denote convergence in distribution, convergence in probability, and almost sure convergence, respectively, as \(n \to \infty\).
	
	The true function \(\func^\ast = \sum_{S \in \mathbb{S}} \func^\ast_S\) is fixed and does not depend on \(n\), except for the case of power analysis with a distinguishable rate for each \( \mathcal{H}_S\) with \(S \in \mathbb{S} \setminus \{\emptyset\}\) in Theorem \ref{thm:global_power}, where we test \( \mathrm{H}_{0,S}: \func^\ast_S = 0 \). In this case, only \(\func^\ast_S = \func^\ast_{S(n)}\) is allowed to depend on \(n\) (triangular array setting for $Y=Y_{(n)}$ and $\{Y_i\}_{i=1}^n=\{Y_{i(n)}\}_{i=1}^n$), while all other \(\func^\ast_{S'}\) for \(S' \in \mathbb{S} \setminus \{S\}\) remain fixed and does not depend of \(n\).
	
	For any arbitrary continuous linear functional \(\mathscr{A}: \mathcal{H} \rightarrow \mathbb{R}\), the Riesz representation theorem guarantees the existence of \( g_{\mathscr{A},\lambda} \in \mathcal{H} \) such that $\mathscr{A}f = \langle f, g_{\mathscr{A},\lambda} \rangle_\lambda$ for all $f\in\mathcal{H}$. For simplicity, we abuse notation and write \( g_{\mathscr{A},\lambda} \) as \(\mathscr{A}\), so that $\mathscr{A}f = \langle f, \mathscr{A} \rangle_\lambda.$ For $S \in \mathbb{S}$, the same abuse of notation with respect to \(\langle \cdot, \cdot \rangle_{S,\lambda}\) is used for any arbitrary continuous linear functional \(\mathscr{A}_S:\mathcal{H}_S\rightarrow\mathbb{R}\). 
	
	Let vector spaces $\mathcal{U}_1$ and $\mathcal{U}_2$ be subspaces of a vector space $\mathcal{U}$. We say that $\mathcal{U} = \mathcal{U}_1 \oplus \mathcal{U}_2$, 
	or equivalently, $\mathcal{U}_1 = \mathcal{U} \ominus \mathcal{U}_2$ or $\mathcal{U}_2 = \mathcal{U} \ominus \mathcal{U}_1$,
	if the following conditions hold:
	\begin{enumerate}[label=(\roman*)]
		\item $\mathcal{U}_1 \cap \mathcal{U}_2 = \{0\}$,
		\item For every $f \in \mathcal{U}$, there exist unique $f_1 \in \mathcal{U}_1$ and $f_2 \in \mathcal{U}_2$ such that $f = f_1 + f_2$.
	\end{enumerate}
	
	For each $j=1,\ldots,d$, assume that Hilbert space $\mathcal{V}_j$ has basis $\{ \varphi_{j,v} \}_{v \in \mathcal{I}_j}$, where the index set $\mathcal{I}_j$ may satisfy $|\mathcal{I}_j| \in \mathbb{N}$ (when $\mathcal{V}_j$ is finite-dimensional) or $\mathcal{I}_j = \mathbb{N}$ (when $\mathcal{V}_j$ is infinite-dimensional). Their tensor product space $\otimes_{j=1}^d \mathcal{V}_j$ is defined as the completion of $\operatorname{Span}\left\{ \prod_{j=1}^d\varphi_{j,v_j}  \right\}_{v_j \in \mathcal{I}_j, j=1,\ldots,d}$ with respect to some inner product $\langle \cdot,\cdot \rangle_{\otimes_{j=1}^d \mathcal{V}_j}$ satisfying $\langle \prod_{j=1}^d f_j, \prod_{j=1}^d g_j \rangle_{\otimes_{j=1}^d \mathcal{V}_j}=\prod_{j=1}^d\langle f_j, g_j \rangle_{\mathcal{V}_j}$ for any $f_j,g_j \in \mathcal{V}_j$, $j=1,\ldots,d$. Note that $\langle\cdot,\cdot\rangle_{\mathcal{V}_j}$ is some inner product on $\mathcal{V}_j$ for each $j=1,\ldots,d$.
	
	%\section{Technical Details}\label{sec:Tech_detail}
	\section{Proofs of Lemmas and Corollaries}\label{sec:Lemma_cor_proof}
	\begin{proof}[Proof of Lemma \ref{lem:ortho}.] Note that \( S \neq S' \) and both are in \( \mathcal{P}_d \). First, referring to \cite{gu2013smoothing}, it is straightforward that \( J(f_S, g_{S'}) = 0 \). Hence, we focus on orthogonality in \( V \). There exist \( f, g \in \otimes_{j=1}^d\mathcal{H}_{[j]} \) such that  
		\[
		f_S = \left\{\prod_{j \in S} (id - \mathcal{A}_{[j]}) \prod_{j \in \{1, \ldots, d\} \setminus S} \mathcal{A}_{[j]} \right\} f \quad\text{and}\quad g_{S'} = \left\{\prod_{j \in S'} (id - \mathcal{A}_{[j]}) \prod_{j \in \{1, \ldots, d\} \setminus S'} \mathcal{A}_{[j]} \right\} g.
		\] Consider the case that \( d = 2 \).  For \( S = \emptyset \) and \( S' = \{1\} \),   
		\[
		V(f_\emptyset, g_{\{1\}}) = f_\emptyset \int_{\mathcal{X}_{[1]}} g_{\{1\}}(x_{[1]}) \, dx_{[1]} = f_\emptyset \mathcal{A}_{[1]} (id - \mathcal{A}_{[1]}) \mathcal{A}_{[2]} g = 0.
		\]  For \( S = \emptyset \) and \( S' = \{1,2\} \),  
		\[
		V(f_\emptyset, g_{\{1,2\}}) = f_\emptyset \int_{\mathcal{X}_{[2]}} \int_{\mathcal{X}_{[1]}} g_{\{1,2\}}(x_{[1]}, x_{[2]}) \, dx_{[1]} \, dx_{[2]} = f_\emptyset \mathcal{A}_{[1]} (id - \mathcal{A}_{[1]}) \mathcal{A}_{[2]} (id - \mathcal{A}_{[2]}) g = 0.
		\]  For \( S = \{1\} \) and \( S' = \{2\} \),  \begin{align}
			V(f_{\{1\}}, g_{\{2\}}) &= \int_{\mathcal{X}_{[1]}} f_{\{1\}}(x_{[1]}) \, dx_{[1]} \int_{\mathcal{X}_{[2]}} g_{\{2\}}(x_{[2]}) \, dx_{[2]}\\
			&= \mathcal{A}_{[1]} (id - \mathcal{A}_{[1]}) \mathcal{A}_{[2]} f \cdot \mathcal{A}_{[2]} (id - \mathcal{A}_{[2]}) \mathcal{A}_{[1]} g = 0.\end{align}
		For \( S = \{1\} \) and \( S' = \{1,2\} \),  
		\[
		V(f_{\{1\}}, g_{\{1,2\}}) = \int_{\mathcal{X}_{[1]}} f_{\{1\}}(x_{[1]}) \int_{\mathcal{X}_{[2]}} g_{\{1,2\}}(x_{[1]}, x_{[2]}) \, dx_{[2]} \, dx_{[1]} = 0,
		\]  where $\int_{\mathcal{X}_{[2]}} g_{\{1,2\}}(\cdot, x_{[2]}) \, dx_{[2]} =  (id - \mathcal{A}_{[1]}) \mathcal{A}_{[2]} (id - \mathcal{A}_{[2]}) g = 0.$ 
		
		For $d=3$, all cases are straightforward except when $S=\{1,2\}$ and $S'=\{2,3\}$:\begin{align}
			V (f_{\{1,2\}},g_{\{2,3\}}) = \int_{\mathcal{X}_{[2]}} \int_{\mathcal{X}_{[1]}} f_{\{1,2\}}(x_{[1]},x_{[2]})dx_{[1]} \int_{\mathcal{X}_{[3]}} g_{\{2,3\}}(x_{[2]},x_{[3]})dx_{[3]} dx_{[2]}=0,
		\end{align} where $\int_{\mathcal{X}_{[1]}} f_{\{1,2\}}(x_{[1]},\cdot)dx_{[1]}= \mathcal{A}_{[1]} (id - \mathcal{A}_{[1]}) (id - \mathcal{A}_{[2]}) \mathcal{A}_{[3]}  f = 0$ and $\int_{\mathcal{X}_{[3]}} g_{\{2,3\}}(\cdot,x_{[3]})dx_{[3]}=  (id - \mathcal{A}_{[2]}) \mathcal{A}_{[3]} (id - \mathcal{A}_{[3]}) \mathcal{A}_{[1]} g  = 0.$
		
		To generalize, for \( d \in \mathbb{N} \), every case can be proven using one of the following examples. When $S=\emptyset$, $V(f_S,g_{S'})=f_S \int_{\mathcal{X}_{S'}} g_{S'}(x_{S'})dx_{S'}=0$. When $S'=\emptyset$, $V(f_S,g_{S'})=g_{S'} \int_{\mathcal{X}_{S}} f_S(x_S)dx_S=0$. For all remaining examples, we have \( S \neq \emptyset \) and \( S' \neq \emptyset \). When $S \subseteq S'$, $V(f_S,g_{S'})=\int_{\mathcal{X}_S} f_S(x_S) \int_{\mathcal{X}_{S'\setminus S}} g_{S'}(x_{S'}) dx_{S'\setminus S} dx_S = 0$. When $S' \subseteq S$, $V(f_S,g_{S'})=\int_{\mathcal{X}_{S'}} g_{S'}(x_{S'}) \int_{\mathcal{X}_{S\setminus S'}} f_{S}(x_S) dx_{S\setminus S'} dx_{S'} = 0$. When $S$ and $S'$ are disjoint, $$V(f_S,g_{S'})=\int_{\mathcal{X}_S} f_S(x_S) dx_S \int_{\mathcal{X}_{S'}} g_{S'}(x_{S'}) dx_{S'} = 0.$$ When $S \not\subseteq S'$, $S' \not\subseteq S$, and $S \cap S' = S''$ is non-empty set, $$V(f_S,g_{S'})=\int_{\mathcal{X}_{S''}} \int_{\mathcal{X}_{S\setminus S''}} f_S(x_S ) dx_{S\setminus S''} \int_{\mathcal{X}_{S'\setminus S''}} g_{S'}(x_{S'} ) dx_{S'\setminus S''} dx_{S''}=0.$$ \end{proof} 
	
	\begin{proof}[Proof of Corollary \ref{cor:Inner_prod_RK_ortho}.]
		The first part is straightforward from orthogonality in Lemma \ref{lem:ortho}. 
		
		For all $x \in \mathcal{X}$, as $\mathcal{R}_{S,\lambda} (x_S,) \in \mathcal{H}_S$, we have $\sum_{S \in \mathbb{S}} \mathcal{R}_{S,\lambda}(x_S,) \in \mathcal{H}$ and
		\begin{align}
			\left\langle f,  \sum_{S \in \mathbb{S}} \mathcal{R}_{S,\lambda}(x_S,) \right\rangle_{\lambda} &= \left\langle \sum_{S \in \mathbb{S}} f_S,  \sum_{S \in \mathbb{S}} \mathcal{R}_{S,\lambda}(x_S,) \right\rangle_{\lambda}\\&= \sum_{S \in \mathbb{S}} \langle f_S, \mathcal{R}_{S,\lambda}(x_S,) \rangle_{S,\lambda}=\sum_{S \in \mathbb{S}} f_S(x_S) = f(x),
		\end{align} where the second equality follows from the first part. Thus, we conclude that \( \mathcal{R}_\lambda(x,) = \sum_{S \in \mathbb{S}} \mathcal{R}_{S,\lambda} (x_S,) \in \mathcal{H}\) for all $x\in \mathcal{X}$. For all $x,x' \in \mathcal{X}$, \begin{align}
			\mathcal{R}_{\lambda}(x,x')=\langle \mathcal{R}_{\lambda}(x,) , \mathcal{R}_{\lambda}(x',)\rangle_{\lambda}&=\left\langle \sum_{S \in \mathbb{S}} \mathcal{R}_{S,\lambda}(x_S,) , \sum_{S \in \mathbb{S}} \mathcal{R}_{S,\lambda}(x'_S,) \right\rangle_{\lambda}\\
			&=\sum_{S \in \mathbb{S}}\langle \mathcal{R}_{S,\lambda}(x_S,) , \mathcal{R}_{S,\lambda}(x'_S,)\rangle_{S,\lambda}=\sum_{S \in \mathbb{S}}\mathcal{R}_{S,\lambda}(x_S,x'_S),
		\end{align}where the third equality follows from the first part. Hence, we have \( \mathcal{R}_\lambda = \sum_{S \in \mathbb{S}} \mathcal{R}_{S,\lambda} \).
		
		Corollary \ref{cor:Inner_prod_RK_ortho} can be also proved by applying Lemma \ref{lem:operator_eigen}. Given \( f = \sum_{S \in \mathbb{S}} f_S \) and \( g = \sum_{S \in \mathbb{S}} g_S \) in \(\mathcal{H}\), we have  \begin{align}
			\langle f,g \rangle_\lambda &= \sum_{v\in\mathbb{N}}V(f,\psi_{v})V(g,\psi_{v})(1+\lambda/\mu_{v}) \\
			&=V(f,\psi_{\emptyset,0})V(g,\psi_{\emptyset,0})(1+\lambda/\mu_{\emptyset,0}) + \sum_{S \in\mathbb{S}\setminus\{\emptyset\}} \sum_{v\in\mathbb{N}} V(f,\psi_{S,v})V(g,\psi_{S,v})(1+\lambda/\mu_{S,v})\\
			&=V_\emptyset(f_\emptyset,\psi_{\emptyset,0})V_\emptyset(g_\emptyset,\psi_{\emptyset,0})(1+\lambda/\mu_{\emptyset,0}) + \sum_{S \in\mathbb{S}\setminus\{\emptyset\}} \sum_{v\in\mathbb{N}} V_S(f_S,\psi_{S,v})V_S(g_S,\psi_{S,v})(1+\lambda/\mu_{S,v})\\
			&=\sum_{S \in \mathbb{S}} \langle f_S,g_S \rangle_{S,\lambda}, \quad \text{where} \quad f_S,g_S \in \mathcal{H}_S.
		\end{align} For all $x,x' \in \mathcal{X}$, 
		\begin{align}
			\mathcal{R}_{\lambda}(x,x')=\sum_{v \in \mathbb{N}}\frac{\psi_{v}(x)\psi_{v}(x')}{1+\lambda/\mu_{v}} = \frac{\psi_{\emptyset,0}\psi_{\emptyset,0}}{1+\lambda/\mu_{\emptyset,0}}+ \sum_{S\in\mathbb{S}\setminus\{\emptyset\}}\sum_{v \in \mathbb{N}}\frac{\psi_{S,v}(x_S)\psi_{S,v}(x'_S)}{1+\lambda/\mu_{S,v}}=\sum_{S \in \mathbb{S}}\mathcal{R}_{S,\lambda}(x_S,x'_S).
	\end{align}\end{proof}
	\begin{proof}[Proof of Lemma \ref{lem:eigen}] For each $S \in \mathbb{S}$, the existence of such eigensystem of $\mathcal{H}_S$ is standard \citep{gu2013smoothing}. Thus, we focus on proving that the union of these eigensystems forms an eigensystem of $\mathcal{H}$.
		
		For each \(S\in\mathbb{S}\), define \(\mathbb{N}_S\) as \(\{0\}\) when \(S=\emptyset\) and \(\mathbb{N}\) otherwise. For any \(v \in \mathbb{N}\), there exist \(S \in \mathbb{S}\) and \(u \in \mathbb{N}_S\) such that \(\{\mu_v,\psi_v\}=\{\mu_{S,u},\psi_{S,u}\}\). So, we have $\mu_v=\mu_{S,u}\geq 0$ and $\psi_v=\psi_{S,u}\in\mathcal{H}_S \subseteq \mathcal{H}$. Similarly, for \(v' \in \mathbb{N}\), there exist \(S' \in \mathbb{S}\) and \(u' \in \mathbb{N}_{S'}\) such that \(\{\mu_{v'},\psi_{v'}\}=\{\mu_{S',u'},\psi_{S',u'}\}\). So, we have $\mu_{v'}=\mu_{S',u'}\geq 0$ and $\psi_{v'}=\psi_{S',u'}\in\mathcal{H}_{S'} \subseteq \mathcal{H}$. We have
		\begin{align}
			&V(\psi_v,\psi_{v'})=V(\psi_{S,u},\psi_{S',u'})=\delta_{S,S'}V_S(\psi_{S,u},\psi_{S',u'})=\delta_{S,S'}\delta_{u,u'}=\delta_{v,v'}\\
			&J(\psi_v,\psi_{v'})=J(\psi_{S,u},\psi_{S',u'})=\delta_{S,S'}J_S(\psi_{S,u},\psi_{S',u'})=\mu_{S,u}^{-1}\delta_{S,S'}\delta_{u,u'}=\mu_v^{-1}\delta_{v,v'},
		\end{align} where in both equations above, the second equality follows from orthogonality in Lemma \ref{lem:ortho}, and the third equality follows from the eigensystem of $\mathcal{H}_S$ established in the preceding part of Lemma~\ref{lem:eigen}.
		
		For any \(f \in \mathcal{H}\), we have the decomposition \(f=\sum_{S \in \mathbb{S}}f_S\), where \(f_S \in \mathcal{H}_S\). We have $f_\emptyset=V_\emptyset(f_\emptyset,\psi_{\emptyset,0})\psi_{\emptyset,0}=V(f_\emptyset,\psi_{\emptyset,0})\psi_{\emptyset,0}=V(f,\psi_{\emptyset,0})\psi_{\emptyset,0}$ and for \(S \in \mathbb{S}\setminus\{\emptyset\}\), $$f_S = \sum_{v \in \mathbb{N}} V_S(f_S,\psi_{S,v})\psi_{S,v}= \sum_{v \in \mathbb{N}} V(f_S,\psi_{S,v})\psi_{S,v}= \sum_{v \in \mathbb{N}} V(f,\psi_{S,v})\psi_{S,v},$$ where in both cases, the last equality follows from orthogonality in Lemma \ref{lem:ortho}. We obtain $$f=\sum_{S \in \mathbb{S}}f_S = V(f,\psi_{\emptyset,0})\psi_{\emptyset,0} + \sum_{S \in \mathbb{S}\setminus\{\emptyset\}}\sum_{v \in \mathbb{N}}V(f,\psi_{S,v})\psi_{S,v} =\sum_{v\in\mathbb{N}}V(f,\psi_v)\psi_v.$$
	\end{proof}
	
	\begin{proof}[Proof of Lemma \ref{lemma:eigen_order}.]
		We first prove results for $\mathcal{H}$. By \cite{lin_tensor}, for $b \in [0, 2-1/(2m))$, with {\color{black}$\beta=2m/(2mb+1) \in (1/2,2m]$,} we have \begin{align}
			\sum_{v \in \mathbb{N}} \frac{ (1+1/\mu_v)^b }{(1+\lambda/\mu_v)^2} &\asymp \lambda^{-1/\beta} \int_{\lambda^{1/\beta}}^\infty (1+ t^\beta)^{-2} \left( \log t + \frac{1}{\beta} \log \frac{1}{\lambda} \right)^{|S_{\sup}|-1} dt\\
			&\asymp\lambda^{-1/\beta} \left(\frac{1}{\beta}\log\frac{1}{\lambda}\right)^{|S_{\sup}|-1} \int_{\lambda^{1/\beta}}^\infty (1+ t^\beta)^{-2} dt \asymp\lambda^{-b-1/(2m)} (-\log\lambda)^{|S_{\sup}|-1},
		\end{align} where by $\beta >1/2$, \begin{align}
			\int_{\lambda^{1/\beta}}^\infty (1+ t^\beta)^{-2} dt \rightarrow \int_0^\infty (1+ t^\beta)^{-2} dt \in (0,\infty).
		\end{align} Since {\color{black}$2m\geq 4 >1$,} we choose $b=0,1$: \begin{align}\label{eqn:eigen_order_1}
			&\sum_{v \in \mathbb{N}} \frac{ 1 }{(1+\lambda/\mu_v)^2} \asymp \lambda^{-1/(2m)} (-\log\lambda)^{|S_{\sup}|-1},\\
			&\sum_{v \in \mathbb{N}} \frac{ 1+1/\mu_v }{(1+\lambda/\mu_v)^2} \asymp \lambda^{-1-1/(2m)} (-\log\lambda)^{|S_{\sup}|-1}.
		\end{align} Considering $\sum_{v \in \mathbb{N}} \lambda/(1+\lambda/\mu_v)^2 \asymp \lambda^{1-1/(2m)} (-\log\lambda)^{|S_{\sup}|-1} \ll \lambda^{-1/(2m)} (-\log\lambda)^{|S_{\sup}|-1}$, \begin{align}
			\sum_{v \in \mathbb{N}} \frac{ \lambda}{(1+\lambda/\mu_v)^2}+\sum_{v \in \mathbb{N}} \frac{ \lambda/\mu_v }{(1+\lambda/\mu_v)^2}=\sum_{v \in \mathbb{N}} \frac{ \lambda+\lambda/\mu_v }{(1+\lambda/\mu_v)^2} \asymp \lambda^{-1/(2m)} (-\log\lambda)^{|S_{\sup}|-1}
		\end{align} implies \begin{align}\label{eqn:eigen_order_2}
			&\sum_{v \in \mathbb{N}} \frac{ \lambda/\mu_v }{(1+\lambda/\mu_v)^2} \asymp \lambda^{-1/(2m)} (-\log\lambda)^{|S_{\sup}|-1}.
		\end{align} Using \eqref{eqn:eigen_order_1} and \eqref{eqn:eigen_order_2}, we have \begin{align}\label{eqn:eigen_order_3}
			&\sum_{v \in \mathbb{N}} \frac{ 1 }{1+\lambda/\mu_v}=\sum_{v \in \mathbb{N}} \frac{ 1 }{(1+\lambda/\mu_v)^2} + \sum_{v \in \mathbb{N}} \frac{ \lambda/\mu_v }{(1+\lambda/\mu_v)^2} \asymp \lambda^{-1/(2m)} (-\log\lambda)^{|S_{\sup}|-1},
		\end{align} which completes the proof of results for $\mathcal{H}$.

		Now we prove results for $\mathcal{H_S}$ for each $S \in \mathbb{S}\setminus\{\emptyset\}$. For any $j \in \{1,\ldots,d\}$, the eigensystem of $\mathcal{H}_{[j]}= \mathcal{H}_{\emptyset[j]} \oplus \mathcal{H}_{\{j\}}= \mathcal{H}_{\emptyset} \oplus \mathcal{H}_{\{j\}}$ with respect to $V$ and {\color{black}$J_\emptyset+J_{\{j\}}$, obtained by applying Lemma~\ref{lem:eigen} to the case $d=1$ and $\mathbb{S} = \mathcal{P}_1$,} is given by $$\{\mu_{\emptyset,0}, \psi_{\emptyset,0}\} \cup\{\mu_{\{j\},v}, \psi_{\{j\},v}\}_{v \in \mathbb{N}},$$ which is aligned in {\color{black}nonincreasing} order of the eigenvalues. By \eqref{eqn:eigen_order_3}, we have \begin{align}
			1+\sum_{v \in \mathbb{N}} \frac{ 1 }{1+\lambda/\mu_{\{j\},v}} = \frac{ 1 }{1+\lambda/\mu_{\emptyset,0}}+\sum_{v \in \mathbb{N}} \frac{ 1 }{1+\lambda/\mu_{\{j\},v}} \asymp\lambda^{-1/(2m)} (-\log\lambda)^{|\{j\}|-1}=\lambda^{-1/(2m)},
		\end{align} which leads to $\sum_{v \in \mathbb{N}} (1+\lambda/\mu_{\{j\},v})^{-1} \asymp\lambda^{-1/(2m)}$ by $1\ll \lambda^{-1/(2m)}$.
		
		Similarly, for any $j \neq j'$ in $\{1,\ldots,d\}$, the eigensystem of $\mathcal{H}_{[j]} \otimes \mathcal{H}_{[j']}=\mathcal{H}_{\emptyset} \oplus \mathcal{H}_{\{j\}} \oplus \mathcal{H}_{\{j'\}} \oplus \mathcal{H}_{\{j,j'\}}$ with respect to $V$ and {\color{black}$J_\emptyset+J_{\{j\}}+J_{\{j'\}}+J_{\{j,j'\}}$, obtained by applying Lemma~\ref{lem:eigen} to the case $d=2$ and $\mathbb{S} = \mathcal{P}_2$,} is given by $$\{\mu_{\emptyset,0}, \psi_{\emptyset,0}\} \cup\{\mu_{\{j\},v}, \psi_{\{j\},v}\}_{v \in \mathbb{N}}\cup\{\mu_{\{j'\},v}, \psi_{\{j'\},v}\}_{v \in \mathbb{N}}\cup\{\mu_{\{j,j'\},v}, \psi_{\{j,j'\},v}\}_{v \in \mathbb{N}},$$ which is aligned in {\color{black}nonincreasing} order of the eigenvalues. By \eqref{eqn:eigen_order_3}, we have \begin{align}
			\frac{ 1 }{1+\lambda/\mu_{\emptyset,0}}&+\sum_{v \in \mathbb{N}} \frac{ 1 }{1+\lambda/\mu_{\{j\},v}}+\sum_{v \in \mathbb{N}} \frac{ 1 }{1+\lambda/\mu_{\{j'\},v}}+\sum_{v \in \mathbb{N}} \frac{ 1 }{1+\lambda/\mu_{\{j,j'\},v}}\\ &\asymp\lambda^{-1/(2m)} (-\log\lambda)^{|\{j,j'\}|-1}=\lambda^{-1/(2m)}(-\log\lambda),
		\end{align} where by $1 \ll \lambda^{-1/(2m)}(-\log\lambda)$ and $\sum_{v \in \mathbb{N}} (1+\lambda/\mu_{\{j\},v})^{-1}, \sum_{v \in \mathbb{N}} (1+\lambda/\mu_{\{j'\},v})^{-1} \asymp\lambda^{-1/(2m)} \ll \lambda^{-1/(2m)}(-\log\lambda)$, we have $\sum_{v \in \mathbb{N}} (1+\lambda/\mu_{\{j,j'\},v})^{-1} \asymp \lambda^{-1/(2m)}(-\log\lambda)$.
		
		Using a similar argument for other values of $|S|$ with $S \in \mathbb{S}\setminus \{\emptyset\}$, one can show that \begin{align}\label{eqn:eigen_order_4}
			\sum_{v\in\mathbb{N}}\frac{1}{1+ \lambda/\mu_{S,v}} \asymp\lambda^{-1/(2m)} (-\log\lambda)^{|S|-1}.
		\end{align}
		
		Note that \eqref{eqn:eigen_order_3} can be verified again by applying \eqref{eqn:eigen_order_4}. We have \begin{align}
			\sum_{v \in \mathbb{N}} \frac{ 1 }{1+\lambda/\mu_v} = \frac{ 1 }{1+\lambda/\mu_{\emptyset,0}} + \sum_{S \in \mathbb{S}\setminus \{\emptyset\}} \sum_{v \in \mathbb{N}} \frac{ 1 }{1+\lambda/\mu_{S,v}} \asymp\lambda^{-1/(2m)} (-\log\lambda)^{|S_{\sup}|-1},
		\end{align} where the first equality holds by Lemma \ref{lem:eigen} and the asymptotic order follows from \eqref{eqn:eigen_order_4}. The results for the other summations can be derived in a similar manner.
	\end{proof} 
	
	\begin{proof}[Proof of Corollary \ref{coro:Wdecomp}.]
		\begin{align}
			\mathcal{W}_{\lambda}f&=\sum_{v\in\mathbb{N}}V(f,\psi_{v})\frac{\lambda/\mu_{v}}{1+\lambda/\mu_{v}}\psi_{v}\\
			&= V(f,\psi_{\emptyset,0})\frac{\lambda/\mu_{\emptyset,0}}{1+\lambda/\mu_{\emptyset,0}}\psi_{\emptyset,0}+ \sum_{S\in\mathbb{S}\setminus\{\emptyset\}}\sum_{v\in\mathbb{N}}V(f,\psi_{S,v})\frac{\lambda/\mu_{S,v}}{1+\lambda/\mu_{S,v}}\psi_{S,v}\\
			&= V_\emptyset(f_\emptyset,\psi_{\emptyset,0})\frac{\lambda/\mu_{\emptyset,0}}{1+\lambda/\mu_{\emptyset,0}}\psi_{\emptyset,0}+ \sum_{S\in\mathbb{S}\setminus\{\emptyset\}}\sum_{v\in\mathbb{N}}V_S(f_S,\psi_{S,v})\frac{\lambda/\mu_{S,v}}{1+\lambda/\mu_{S,v}}\psi_{S,v}= \sum_{S \in \mathbb{S}} \mathcal{W}_{S,\lambda}f_S.\end{align}\end{proof}
	
	\section{Supporting Theoretical Results}\label{sec:support_theory}
	
	The following Lemma \ref{lem:support} presents properties used in proving main results.
	
	\begin{lemma}\label{lem:support}
		The following statements hold:
		\begin{enumerate}
			\item $\mathbb{E}(Y)= \func^\ast_\emptyset + \sum_{S \in \mathbb{S}\setminus\{\emptyset\}} \mathbb{E}(\func^\ast_S(X_S))+\mathbb{E}(\epsilon)= \func^\ast_\emptyset + \sum_{S \in \mathbb{S}\setminus\{\emptyset\}} \int_{\mathcal{X}_S}\func^\ast_S(x_S) d x_S=\func^\ast_\emptyset$.
			\item $\mathbb{E}\left(\mathcal{R}_{\lambda}(X,X)\right) = \sum_{v\in \mathbb{N}}(1+\lambda/\mu_v)^{-1}=\sum_{S\in\mathbb{S}}\mathbb{E}\left(\mathcal{R}_{S,\lambda}(X_S,X_S)\right)$, where for $S\in\mathbb{S}\setminus\{\emptyset\}$, $\mathbb{E}\left(\mathcal{R}_{S,\lambda}(X_S,X_S)\right) = \sum_{v\in \mathbb{N}}(1+\lambda/\mu_{S,v})^{-1}$ and $\mathcal{R}_{\emptyset,\lambda}= (1+\lambda/\mu_{\emptyset,0})^{-1}=1$.
			\item $\mathbb{E}_X\left( \mathcal{R}_{\lambda}(X,\cdot)\right)=\mathcal{R}_{\emptyset,\lambda}=1$.
			\item For any $x \in \mathcal{X}$, $ \mathcal{R}_\lambda(x,x)  \leq \mathcal{C}_{\psi}^2 \mathcal{C}_{\mathcal{R}} \lambda^{-1/(2m)} (-\log\lambda)^{|S_{\sup}|-1}$ for a constant $\mathcal{C}_{\mathcal{R}} \in (0,\infty)$.
			\item For each $S \in \mathbb{S}\setminus\{\emptyset\}$, for any $x_S \in \mathcal{X}_S$, $ \mathcal{R}_{S,\lambda}(x_S,x_S)  \leq \mathcal{C}_{\psi}^2 \mathcal{C}_{\mathcal{R},S} \lambda^{-1/(2m)} (-\log\lambda)^{|S|-1}$ for a constant $\mathcal{C}_{\mathcal{R},S} \in (0,\infty)$.
			\item For all $f \in \mathcal{H}$, $\|f\|_{\sup} \leq \mathcal{C}_{\psi} \mathcal{C}_{\mathcal{R}}^{1/2} \lambda^{-1/(4m)} (-\log\lambda)^{(|S_{\sup}|-1)/2} \|f\|_\lambda$.
			\item For each $S \in \mathbb{S}\setminus\{\emptyset\}$, for all $f_S \in \mathcal{H}_S$, $\|f_S\|_{\sup} \leq \mathcal{C}_{\psi} \mathcal{C}_{\mathcal{R},S}^{1/2} \lambda^{-1/(4m)} (-\log\lambda)^{(|S|-1)/2} \|f_S\|_{S,\lambda}$.
			\item For all $f \in \mathcal{H}$, $\|\mathcal{W}_\lambda f\|_\lambda \leq \lambda^{1/2}\sqrt{J(f)}$.
			\item For each $S \in \mathbb{S}\setminus\{\emptyset\}$, for all $f_S \in \mathcal{H}_S$, $\|\mathcal{W}_{S,\lambda} f_S\|_{S,\lambda} \leq \lambda^{1/2}\sqrt{J_S(f_S)}$.
			\item Suppose that Assumption \ref{asp:smooth} holds. {\color{black}If \iffalse$\lambda$ satisfies $\sqrt{n} \lambda =o(1)$, $n^{-1} \lambda^{-1/m} (-\log \lambda)^{2|S_{\sup}|-2}=o(1)$, and \fi$\sqrt{n}\alpha_n=o(1)$,} then $\beta_n$ from Theorem \ref{thm:fbr} is $o(1)$.
	\end{enumerate}\end{lemma}
	
	\begin{proof}[Proof of Lemma \ref{lem:support}.]
		We omit the proofs of (i) and (ii) since these results are straightforward. We now proceed to prove (iii): \begin{align}
			\mathbb{E}_X\left(\mathcal{R}_\lambda(X,\cdot)\right) &= \sum_{v \in \mathbb{N}}\frac{\mathbb{E}_X\left(\psi_{v}(X)\right)}{1+\lambda/\mu_{v}} \psi_{v} = \frac{\psi_{\emptyset,0}\psi_{\emptyset,0}}{1+\lambda/\mu_{\emptyset,0}}+ \sum_{S\in\mathbb{S}\setminus\{\emptyset\}}\sum_{v \in \mathbb{N}}\frac{\int_{\mathcal{X}_S}\psi_{S,v}(x_S)dx_S}{1+\lambda/\mu_{S,v}}\psi_{S,v}=\mathcal{R}_{\emptyset,\lambda}=1,\end{align}which leads to $\mathbb{E}_X\left(f(X)\right) = \mathbb{E}_X \left\langle f, \mathcal{R}_\lambda(X,\cdot) \right\rangle_\lambda=  \left\langle f, \mathbb{E}_X\left(\mathcal{R}_\lambda(X,\cdot)\right) \right\rangle_\lambda=\langle f_\emptyset, \mathcal{R}_{\emptyset,\lambda} \rangle_{\emptyset,\lambda}=f_\emptyset$ for all $f=\sum_{S\in\mathbb{S}}f_S \in \mathcal{H}$ with $f_S \in \mathcal{H}_S$, which corresponds to (i).
		
		We proceed to prove (iv), while the proof of (v) is similar and thus omitted. We have \begin{align}
			\mathcal{R}_{\lambda}(x,x)=\sum_{v \in \mathbb{N}}\frac{\psi_{v}^2(x)}{1+\lambda/\mu_{v}} \leq \mathcal{C}_\psi^2 \sum_{v \in \mathbb{N}}\frac{1}{1+\lambda/\mu_{v}} \leq \mathcal{C}_{\psi}^2 \mathcal{C}_{\mathcal{R}} \lambda^{-1/(2m)} (-\log\lambda)^{|S_{\sup}|-1},
		\end{align} where the first inequality is by Lemma \ref{lem:eigen} and the second inequality is by Lemma \ref{lemma:eigen_order}.
		
		We prove (vi), as the proof of (vii) is similar and thus omitted. \begin{align}
			\|f\|_{\sup}&=\sup_{x\in\mathcal{X}}|f(x)| = \sup_{x\in\mathcal{X}}|\langle f, \mathcal{R}_\lambda(x,\cdot)\rangle_\lambda| \leq \sup_{x\in\mathcal{X}}\|f\|_\lambda \|\mathcal{R}_\lambda(x,\cdot)\|_\lambda\\
			&=\sup_{x\in\mathcal{X}}\|f\|_\lambda \sqrt{\mathcal{R}_\lambda(x,x)}\leq \mathcal{C}_{\psi} \mathcal{C}_{\mathcal{R}}^{1/2} \lambda^{-1/(4m)} (-\log\lambda)^{(|S_{\sup}|-1)/2} \|f\|_\lambda.
		\end{align}
		
		We now prove (viii); the proof of (ix) is analogous and therefore omitted.\begin{align}
			\|\mathcal{W}_\lambda f\|_\lambda &= \underset{g \in \mathcal{H}, \|g\|_\lambda \leq 1}{\sup} \langle \mathcal{W}_\lambda f,g \rangle_\lambda = \underset{g \in \mathcal{H}, \|g\|_\lambda \leq 1}{\sup} \lambda J(f,g)\\
			&\leq \lambda^{1/2}\sqrt{J(f)} \underset{g \in \mathcal{H}, \|g\|_\lambda \leq 1}{\sup} \sqrt{\lambda J(g)}\leq \lambda^{1/2}\sqrt{J(f)} \underset{g \in \mathcal{H}, \|g\|_\lambda \leq 1}{\sup} \|g\|_\lambda \leq \lambda^{1/2}\sqrt{J(f)}.
		\end{align}
		
		We proceed to prove (x). We have \begin{align}
			\alpha_n=\beta_n \gamma_n =o(n^{-1/2}), 
		\end{align} where $ n^{-1/2} \ll  \gamma_n= n^{-1/2} \lambda^{-1/(4m)} (-\log\lambda)^{(|S_{\sup}|-1)/2} + \lambda^{1/2}J^{1/2}(\func^\ast)$. Thus, we have \begin{align}
			\beta_n = \gamma_n^{-1} o( n^{-1/2}) = o( n^{1/2})o( n^{-1/2}) = o(1).
		\end{align}
	\end{proof}
	
	We present the Fréchet derivative used in our model. For all $f \in \mathcal{H}$, define $$\mathscr{L}(f)=\frac{1}{2}\mathbb{E}_Z\left(Y-f(X)\right)^2 $$ and $\mathscr{L}_\lambda(f)=\mathscr{L}(f)+ \lambda J(f)/2$. We introduce the Fréchet derivative notations for $\mathscr{L}$, $\mathscr{L}_\lambda$, $\mathscr{L}_n$, and $\mathscr{L}_{n,\lambda}$, in the following with Fréchet derivative operator $\mathcal{D}$. For all $f,g,h \in \mathcal{H}$, \begin{align}
		\mathscr{S}(f)g &= \mathcal{D}\mathscr{L}(f)g = - \mathbb{E}_Z\left(\left(Y-f(X)\right)g(X)\right) = \langle \mathscr{S}(f),g \rangle_\lambda, \\
		\mathscr{S}_{\lambda}(f)g &= \mathcal{D}\mathscr{L}_{\lambda}(f)g = \mathscr{S}(f)g + \lambda J(f,g) = \langle \mathscr{S}_{\lambda}(f),g \rangle_\lambda, \\
		\mathcal{D}\mathscr{S}(f)gh &= \mathcal{D}^2\mathscr{L}(f)gh =  \mathbb{E}_X \left(g(X)h(X)\right) = \langle \mathcal{D}\mathscr{S}(f)g,h \rangle_\lambda,\\
		\mathcal{D}\mathscr{S}_{\lambda}(f)gh &= \mathcal{D}^2\mathscr{L}_{\lambda}(f)gh = \mathcal{D}\mathscr{S}(f)gh +\lambda J(g,h)  = \langle \mathcal{D}\mathscr{S}_{\lambda}(f)g,h \rangle_\lambda,
	\end{align} where \begin{align}
		\mathscr{S}(f) &= -\mathbb{E}_Z\left(\left(Y-f(X)\right)\mathcal{R}_{\lambda}(X,\cdot)\right) \in \mathcal{H},\quad \mathscr{S}_{\lambda}(f) = \mathscr{S}(f) + \mathcal{W}_{\lambda}f \in \mathcal{H},\\
		\mathcal{D}\mathscr{S}(f)g &= \mathbb{E}_X\left(g(X)\mathcal{R}_{\lambda}(X,\cdot)\right)  \in \mathcal{H},\quad \mathcal{D}\mathscr{S}_{\lambda}(f)g = \mathcal{D}\mathscr{S}(f)g + \mathcal{W}_{\lambda}g \in \mathcal{H}.
	\end{align} Using the eigensystem representation, we can represent $\mathscr{S}(f)$ and $\mathcal{D}\mathscr{S}(f)g$:
	\begin{align}
		\mathscr{S}(f) &= -\mathbb{E}_Z\left(Y\mathcal{R}_{\lambda}(X,\cdot)\right) + \mathbb{E}_X\left(f(X)\mathcal{R}_{\lambda}(X,\cdot)\right)  \\&=-\sum_{v\in\mathbb{N}} \frac{ \mathbb{E}_Z\left( (Y-f(X))\psi_v(X) \right) }{1+\lambda/\mu_{v}} \psi_{v} \in \mathcal{H},\\
		\mathcal{D}\mathscr{S}(f)g &= \mathbb{E}_X\left(g(X)\mathcal{R}_{\lambda}(X,\cdot)\right) =  \mathbb{E}_X \left( \sum_{v\in\mathbb{N}}V(g,\psi_v) \psi_v(X)   \right)\left( \sum_{v\in\mathbb{N}}\frac{ \psi_v(X)}{1+\lambda/\mu_v} \psi_v   \right)\\
		&=\sum_{v\in\mathbb{N}}\sum_{v'\in\mathbb{N}} \frac{V(g,\psi_v) }{1+\lambda/\mu_{v'}} \psi_{v'} V(\psi_v,\psi_{v'})=\sum_{v\in\mathbb{N}} \frac{V(g,\psi_v) }{1+\lambda/\mu_{v}} \psi_{v} \in \mathcal{H}.
	\end{align}

	Similarly, for all $f,g,h \in \mathcal{H}$,
	\begin{align}
		\mathscr{S}_n(f)g &= \mathcal{D}\mathscr{L}_n(f)g = -\frac{1}{n} \sum_{i=1}^n\left(Y_i-f(X_i)\right)g(X_i) = \langle \mathscr{S}_n(f),g \rangle_\lambda, \\
		\mathscr{S}_{n,\lambda}(f)g &= \mathcal{D}\mathscr{L}_{n,\lambda}(f)g = \mathscr{S}_n(f)g + \lambda J(f,g) = \langle \mathscr{S}_{n,\lambda}(f),g \rangle_\lambda, \\
		\mathcal{D}\mathscr{S}_n(f)gh &= \mathcal{D}^2\mathscr{L}_n(f)gh = \frac{1}{n} \sum_{i=1}^n g(X_i)h(X_i) = \langle \mathcal{D}\mathscr{S}_n(f)g,h \rangle_\lambda,\\
		\mathcal{D}\mathscr{S}_{n,\lambda}(f)gh &= \mathcal{D}^2\mathscr{L}_{n,\lambda}(f)gh = \mathcal{D}\mathscr{S}_n(f)gh +\lambda J(g,h)  = \langle \mathcal{D}\mathscr{S}_{n,\lambda}(f)g,h \rangle_\lambda,
	\end{align} where \begin{align}
		\mathscr{S}_n(f) &= -\frac{1}{n} \sum_{i=1}^n\left(Y_i-f(X_i)\right)\mathcal{R}_\lambda(X_i,\cdot) \in \mathcal{H}, \quad 
		\mathscr{S}_{n,\lambda}(f) = \mathscr{S}_n(f) + \mathcal{W}_{\lambda}f \in \mathcal{H},\\
		\mathcal{D}\mathscr{S}_n(f)g &= \frac{1}{n}\sum_{i=1}^n g(X_i)\mathcal{R}_{\lambda}(X_i,\cdot) \in \mathcal{H},\quad
		\mathcal{D}\mathscr{S}_{n,\lambda}(f)g = \mathcal{D}\mathscr{S}_n(f)g + \mathcal{W}_{\lambda}g \in \mathcal{H}.\end{align} Note that $$\mathscr{S}_n(f^\ast)=-\frac{1}{n} \sum_{i=1}^n\left(Y_i-f^\ast(X_i)\right)\mathcal{R}_\lambda(X_i,\cdot)=-\frac{1}{n} \sum_{i=1}^n\epsilon_i\mathcal{R}_\lambda(X_i,\cdot).$$
	
	The following Lemma \ref{lem:frechet} presents properties related to Fréchet derivatives.
	
	\begin{lemma}\label{lem:frechet}
		The following statements hold:
		\begin{enumerate}
			\item $\mathscr{S}(f^\ast)=0$.
			\item For any \( f \in \mathcal{H} \), \( \mathcal{D} \mathscr{S}_\lambda(f) \) is \( id \) on \( \mathcal{H} \), that is, for all $g \in \mathcal{H}$, $\mathcal{D} \mathscr{S}_\lambda(f) g = g\in\mathcal{H}$.
	\end{enumerate}\end{lemma} Note that as $\mathscr{S}(f^\ast)=0$, we have $\mathscr{S}(f^\ast) g = \langle \mathscr{S}(f^\ast), g \rangle_\lambda =0 $ for all $g \in \mathcal{H}$, which implies \begin{align}\label{eqn:opt_true}f^\ast = \underset{f \in \mathcal{H}}{\argmin} \, \mathscr{L}(f).\end{align} Similarly, by the optimization of $\hat{f}$ in \eqref{eqn:opt}, $\mathscr{S}_{n,\lambda}(\hat{f})=0$, so that for all $g \in \mathcal{H}$, we have $\mathscr{S}_{n,\lambda}(\hat{f})g  = \langle \mathscr{S}_{n,\lambda}(\hat{f}),g \rangle_\lambda =0.$
	
	\begin{proof}[Proof of Lemma \ref{lem:frechet}.]
		We prove (i) by the following:
		\begin{align}
			\mathscr{S}(f^\ast)&=-\sum_{v\in\mathbb{N}} \frac{ \mathbb{E}_Z\left( (Y-f^\ast(X))\psi_v(X) \right) }{1+\lambda/\mu_{v}} \psi_{v} \\&=-\sum_{v\in\mathbb{N}} \frac{ \mathbb{E}\left( \epsilon\psi_v(X) \right) }{1+\lambda/\mu_{v}} \psi_{v}=-\sum_{v\in\mathbb{N}} \frac{ \mathbb{E}\left( \epsilon)\mathbb{E}(\psi_v(X) \right) }{1+\lambda/\mu_{v}} \psi_{v}=0.
		\end{align}
		
		We proceed to prove (ii). For any $f \in \mathcal{H}$, $\langle\mathcal{D}\mathscr{S}_\lambda(f)g,h\rangle_\lambda = \langle g,h \rangle_\lambda$ for all $g,h \in \mathcal{H}$, implying that $\mathcal{D}\mathscr{S}_\lambda(f)g=g$ for all $g\in\mathcal{H}$. To prove this in another way, for all $g\in \mathcal{H}$, \begin{align}
			\mathcal{D}\mathscr{S}_\lambda(f)g=\sum_{v\in\mathbb{N}} \frac{V(g,\psi_v) }{1+\lambda/\mu_{v}} \psi_{v} + \sum_{v\in\mathbb{N}}V(g,\psi_{v})\frac{\lambda/\mu_{v}}{1+\lambda/\mu_{v}}\psi_{v} = \sum_{v\in\mathbb{N}} V(g,\psi_v) \psi_v = g \in \mathcal{H}.\end{align}\end{proof}

	\section{Proofs of Theorems}\label{sec:Tech_detail}
	\subsection{Proof of Theorem \ref{thm:fbr}} 
	
	The proof is based on derivations from \cite{liu_2020}.
	
	\subsubsection*{Part 1}
	
	Define $\Tilde{\func}=\mathbb{E}(\hat{\func}|\boldsymbol{X}) = \sum_{v\in\mathbb{N}}\mathbb{E}(V(\hat{\func},\psi_v)|\boldsymbol{X}) \psi_v \in \mathcal{H}$, where $\boldsymbol{X}=(X_1,\ldots,X_n)^{\top}$ and for all $x \in \mathcal{X}$, $\Tilde{\func}(x)=\mathbb{E}(\hat{\func}(x)|\boldsymbol{X})=\langle \mathbb{E}(\hat{\func}|\boldsymbol{X}), \mathcal{R}_\lambda(x,\cdot) \rangle_\lambda$. One can see that \begin{align}
		&0=\mathscr{S}_{n,\lambda}(\hat{\func})=-\frac{1}{n} \sum_{i=1}^n\left(Y_i-\hat{\func}(X_i)\right)\mathcal{R}_\lambda(X_i,\cdot) + \mathcal{W}_{\lambda}\hat{\func},\\
		&0=\frac{1}{n}\sum_{i=1}^n(\Tilde{\func}(X_i)-\func^\ast(X_i))\mathcal{R}_\lambda(X_i,\cdot) + \mathcal{W}_{\lambda}\tilde{\func},
	\end{align} where the following line is obtained by applying $\mathbb{E}(\cdot|\boldsymbol{X})$ to the above line. Subtracting them gives \begin{align}
		-\frac{1}{n} \sum_{i=1}^n\left(\epsilon_i-(\hat{\func}-\tilde{\func})(X_i)\right)\mathcal{R}_\lambda(X_i,\cdot) + \mathcal{W}_{\lambda}(\hat{\func}-\tilde{\func})=0,
	\end{align} implying that \begin{align}
		\func^\star\equiv\hat{\func}-\tilde{\func}=\underset{f \in \mathcal{H}}{\argmin} \, \mathscr{L}_{n,\lambda}^\star(f) \equiv \underset{f \in \mathcal{H}}{\argmin} \, \frac{1}{2 n} \sum_{i=1}^n \left( \epsilon_i - f(X_i) \right)^2 + \frac{\lambda}{2} J(f).
	\end{align} 
	
	Let $\func^\epsilon=\sum_{i=1}^n \epsilon_i \mathcal{R}_\lambda(X_i,\cdot)/n$. By the Taylor series expansion of $\mathscr{L}_{n,\lambda}^\star(\func^\epsilon)$ at $\func^\star$, we have \begin{align}\label{eqn:thm1_taylor_1}
		\mathscr{L}_{n,\lambda}^\star(\func^\epsilon)-\mathscr{L}_{n,\lambda}^\star(\func^\star)&=\mathcal{D}\mathscr{L}_{n,\lambda}^\star(\func^\star)(\func^\epsilon-\func^\star)+\frac{1}{2}\mathcal{D}^2\mathscr{L}_{n,\lambda}^\star(\func^\star)(\func^\epsilon-\func^\star)(\func^\epsilon-\func^\star)\\
		&=\frac{1}{2}\mathcal{D}^2\mathscr{L}_{n,\lambda}^\star(\func^\star)(\func^\epsilon-\func^\star)(\func^\epsilon-\func^\star)=\frac{1}{2}\mathsf{P}_n(\func^\epsilon-\func^\star)^2 + \frac{\lambda}{2} J(\func^\epsilon-\func^\star).
	\end{align} Similarly, by the Taylor series expansion of $\mathscr{L}_{n,\lambda}^\star(\func^\star)$ at $\func^\epsilon$, we have \begin{align}\label{eqn:thm1_taylor_2}
		\mathscr{L}_{n,\lambda}^\star(\func^\star)-\mathscr{L}_{n,\lambda}^\star(\func^\epsilon)&=\mathcal{D}\mathscr{L}_{n,\lambda}^\star(\func^\epsilon)(\func^\star-\func^\epsilon)+\frac{1}{2}\mathcal{D}^2\mathscr{L}_{n,\lambda}^\star(\func^\epsilon)(\func^\star-\func^\epsilon)(\func^\star-\func^\epsilon)\\
		&= (\mathsf{P}_n-\mathsf{P})\func^\epsilon(\func^\star-\func^\epsilon)+\frac{1}{2}\mathsf{P}_n(\func^\star-\func^\epsilon)^2 + \frac{\lambda}{2} J(\func^\star-\func^\epsilon),\end{align} where \begin{align}
		\mathcal{D}\mathscr{L}_{n,\lambda}^\star(\func^\epsilon)(\func^\star-\func^\epsilon)&=-\frac{1}{n} \sum_{i=1}^n\left(\epsilon_i-\func^\epsilon(X_i)\right)(\func^\star-\func^\epsilon)(X_i)+\lambda J(\func^\epsilon,\func^\star-\func^\epsilon)\\
		&=-\frac{1}{n} \sum_{i=1}^n \epsilon_i \langle \func^\star-\func^\epsilon,\mathcal{R}_\lambda(X_i,\cdot) \rangle_\lambda + \mathsf{P}_n\func^\epsilon(\func^\star-\func^\epsilon)+\lambda J(\func^\epsilon,\func^\star-\func^\epsilon)\\
		&=- \langle \func^\epsilon, \func^\star-\func^\epsilon \rangle_\lambda + \mathsf{P}_n\func^\epsilon(\func^\star-\func^\epsilon)+\lambda J(\func^\epsilon,\func^\star-\func^\epsilon)= (\mathsf{P}_n-\mathsf{P})\func^\epsilon(\func^\star-\func^\epsilon).
	\end{align} 
	
	Adding \eqref{eqn:thm1_taylor_1} and \eqref{eqn:thm1_taylor_2}, we get \begin{align}\label{eqn:thm1_taylor_3}
		\mathsf{P}_n(\func^\star-\func^\epsilon)^2 + \lambda J(\func^\star-\func^\epsilon)=-(\mathsf{P}_n-\mathsf{P})\func^\epsilon(\func^\star-\func^\epsilon).
	\end{align} Define $\xi= \underset{f,g\in \mathcal{H}, \|f\|_\lambda=\|g\|_\lambda=1}{\sup}|\mathsf{P}_n fg-\mathsf{P}fg|.$ Then we have
	\begin{align}\label{eqn:thm1_xi_ineq_1}
		|(\mathsf{P}_n-\mathsf{P})\func^\epsilon(\func^\star-\func^\epsilon)| &=\|\func^\epsilon\|_\lambda\|\func^\star-\func^\epsilon\|_\lambda\left|(\mathsf{P}_n-\mathsf{P})\frac{\func^\epsilon}{\|\func^\epsilon\|_\lambda}\frac{\func^\star-\func^\epsilon}{\|\func^\star-\func^\epsilon\|_\lambda} \right|\\&\leq \|\func^\epsilon\|_\lambda\|\func^\star-\func^\epsilon\|_\lambda \xi,
	\end{align} \begin{align}\label{eqn:thm1_xi_ineq_2}
		\mathsf{P}_n(\func^\star-\func^\epsilon)^2 + \lambda J(\func^\star-\func^\epsilon)&=(\mathsf{P}_n-\mathsf{P})(\func^\star-\func^\epsilon)^2 + \|\func^\star-\func^\epsilon\|_\lambda^2 \\&\in \left[ (1-\xi)\|\func^\star-\func^\epsilon\|_\lambda^2,(1+\xi)\|\func^\star-\func^\epsilon\|_\lambda^2 \right],
	\end{align} where \begin{align}
		|(\mathsf{P}_n-\mathsf{P})(\func^\star-\func^\epsilon)^2| =\|\func^\star-\func^\epsilon\|_\lambda^2\left|(\mathsf{P}_n-\mathsf{P})\left(\frac{\func^\star-\func^\epsilon}{\|\func^\star-\func^\epsilon\|_\lambda}\right)^2 \right| \leq \|\func^\star-\func^\epsilon\|_\lambda^2 \xi.
	\end{align} Using \eqref{eqn:thm1_taylor_3}, \eqref{eqn:thm1_xi_ineq_1}, and \eqref{eqn:thm1_xi_ineq_2}, we get \begin{align}
		(1-\xi)\|\func^\star-\func^\epsilon\|_\lambda \leq  \xi\|\func^\epsilon\|_\lambda.
	\end{align} We have \begin{align}
		\mathbb{E}\|\func^\epsilon\|_\lambda^2 &= \frac{1}{n^2}\sum_{i=1}^n\sum_{i'=1}^n \mathbb{E}(\epsilon_i\epsilon_{i'})\mathbb{E}(\mathcal{R}_\lambda(X_i,X_{i'}))\\
		&= \frac{\sigma^2}{n^2}\sum_{i=1}^n \mathbb{E}(\mathcal{R}_\lambda(X_i,X_i)) = \frac{\sigma^2}{n} \sum_{v\in \mathbb{N}}\frac{1}{1+\lambda/\mu_v} \\&\asymp n^{-1}\lambda^{-1/(2m)} (-\log\lambda)^{|S_{\sup}|-1}=\bigO\left(n^{-1}\lambda^{-1/(2m)} (-\log\lambda)^{|S_{\sup}|-1}\right),
	\end{align} which implies $\|\func^\epsilon\|_\lambda=\bigO_\mathbb{P}(n^{-1/2}\lambda^{-1/(4m)} (-\log\lambda)^{(|S_{\sup}|-1)/2})$. As $\xi = \bigO_{\mathbb{P}}(\beta_n)=o_{\mathbb{P}}(1)$ (as shown later), we have $\xi \leq 1/2$ with probability approaching one. So, we have $\|\func^\star-\func^\epsilon\|_\lambda \leq 2 \xi\|\func^\epsilon\|_\lambda$ with probability approaching one, resulting in \begin{align}\label{eqn:thm1_main_res1}
		\left\|\hat{\func}-\tilde{\func}-\frac{1}{n}\sum_{i=1}^n \epsilon_i \mathcal{R}_\lambda(X_i,\cdot)\right\|_\lambda= \bigO_\mathbb{P}\left( \beta_n n^{-1/2}\lambda^{-1/(4m)} (-\log\lambda)^{(|S_{\sup}|-1)/2}\right).    
	\end{align} \iffalse Considering that $\xi$ is $\sigma(\boldsymbol{X})$-measurable, taking $\mathbb{E}(\cdot|\boldsymbol{X})$ to above gives \begin{align}
		(1-\xi)\mathbb{E}(\|\func^\star-\func^\epsilon\|_\lambda^2|\boldsymbol{X}) &\leq \xi\mathbb{E}(\|\func^\epsilon\|_\lambda\|\func^\star-\func^\epsilon\|_\lambda|\boldsymbol{X})\\
		&\leq\xi\sqrt{\mathbb{E}(\|\func^\epsilon\|_\lambda^2|\boldsymbol{X})\mathbb{E}(\|\func^\star-\func^\epsilon\|_\lambda^2|\boldsymbol{X})}.
	\end{align} As $\xi = \bigO_{\mathbb{P}}(\beta_n)=o_{\mathbb{P}}(1)$ (as shown later), we have $\xi \leq 1/2$ with probability approaching one, implying that the following holds with probability approaching one: \begin{align}
		\mathbb{E}\left(\left\|\hat{\func}-\tilde{\func}-\frac{1}{n}\sum_{i=1}^n \epsilon_i \mathcal{R}_\lambda(X_i,\cdot)\right\|_\lambda^2\Bigg|\boldsymbol{X}\right)=\mathbb{E}(\|\func^\star-\func^\epsilon\|_\lambda^2|\boldsymbol{X}) &\leq \frac{1}{(1-\xi)^2} \xi^2 \mathbb{E}(\|\func^\epsilon\|_\lambda^2|\boldsymbol{X})\\
		&\leq 4 \xi^2 \mathbb{E}(\|\func^\epsilon\|_\lambda^2|\boldsymbol{X})\\
		&\leq \frac{4 \xi^2}{n^2}\sum_{i=1}^n\sum_{i'=1}^n \mathbb{E}(\epsilon_i\epsilon_{i'}|\boldsymbol{X})\mathcal{R}_\lambda(X_i,X_{i'})\\
		&= \frac{4 \xi^2 \sigma^2}{n^2}\sum_{i=1}^n \mathcal{R}_\lambda(X_i,X_i)\\
		& \leq 4  \sigma^2 \mathcal{C}_{\psi}^2 \mathcal{C}_{\mathcal{R}} \xi^2 n^{-1}\lambda^{-1/(2m)} (-\log\lambda)^{|S_{\sup}|-1}.
	\end{align}\fi
	
	\subsubsection*{Part 2}
	
	Define $\func^\ast_\mathcal{W}=\func^\ast-\mathcal{W}_\lambda\func^\ast$. One can see that \begin{align}
		\tilde{\func}=\underset{f \in \mathcal{H}}{\argmin} \, \tilde{\mathscr{L}}_{n,\lambda}(f) \equiv \underset{f \in \mathcal{H}}{\argmin} \, \frac{1}{2 n} \sum_{i=1}^n \left( \func^\ast(X_i) - f(X_i) \right)^2 + \frac{\lambda}{2} J(f).
	\end{align} Using the Taylor series expansion of $\tilde{\mathscr{L}}_{n,\lambda}(\func^\ast_\mathcal{W})$ at $\tilde{\func}$, we have \begin{align}\label{eqn:thm1_taylor_4}
		\tilde{\mathscr{L}}_{n,\lambda}(\func^\ast_\mathcal{W})-\tilde{\mathscr{L}}_{n,\lambda}(\tilde{\func})&=\mathcal{D}\tilde{\mathscr{L}}_{n,\lambda}(\tilde{\func})(\func^\ast_\mathcal{W}-\tilde{\func})+\frac{1}{2}\mathcal{D}^2\tilde{\mathscr{L}}_{n,\lambda}(\tilde{\func})(\func^\ast_\mathcal{W}-\tilde{\func})(\func^\ast_\mathcal{W}-\tilde{\func})\\
		&=\frac{1}{2}\mathcal{D}^2\tilde{\mathscr{L}}_{n,\lambda}(\tilde{\func})(\func^\ast_\mathcal{W}-\tilde{\func})(\func^\ast_\mathcal{W}-\tilde{\func})=\frac{1}{2}\mathsf{P}_n(\func^\ast_\mathcal{W}-\tilde{\func})^2 + \frac{\lambda}{2} J(\func^\ast_\mathcal{W}-\tilde{\func}).
	\end{align} Similarly, sing the Taylor series expansion of $\tilde{\mathscr{L}}_{n,\lambda}(\tilde{\func})$ at $\func^\ast_\mathcal{W}$, we have \begin{align}\label{eqn:thm1_taylor_5}
		\tilde{\mathscr{L}}_{n,\lambda}(\tilde{\func})-\tilde{\mathscr{L}}_{n,\lambda}(\func^\ast_\mathcal{W})&=\mathcal{D}\tilde{\mathscr{L}}_{n,\lambda}(\func^\ast_\mathcal{W})(\tilde{\func}-\func^\ast_\mathcal{W})+\frac{1}{2}\mathcal{D}^2\tilde{\mathscr{L}}_{n,\lambda}(\func^\ast_\mathcal{W})(\tilde{\func}-\func^\ast_\mathcal{W})(\tilde{\func}-\func^\ast_\mathcal{W})\\
		&=\mathsf{P}_n(\func^\ast_\mathcal{W}-\func^\ast)(\tilde{\func}-\func^\ast_\mathcal{W})+\lambda J(\func^\ast_\mathcal{W},\tilde{\func}-\func^\ast_\mathcal{W}) +\frac{1}{2}\mathsf{P}_n(\tilde{\func}-\func^\ast_\mathcal{W})^2 + \frac{\lambda}{2} J(\tilde{\func}-\func^\ast_\mathcal{W}).
	\end{align} 
	
	Adding \eqref{eqn:thm1_taylor_4} and \eqref{eqn:thm1_taylor_5}, we get \begin{align}\label{eqn:thm1_taylor_6}
		\mathsf{P}_n(\tilde{\func}-\func^\ast_\mathcal{W})^2 + \lambda J(\tilde{\func}-\func^\ast_\mathcal{W}) &= \mathsf{P}_n(\func^\ast-\func^\ast_\mathcal{W})(\tilde{\func}-\func^\ast_\mathcal{W})-\lambda J(\func^\ast_\mathcal{W},\tilde{\func}-\func^\ast_\mathcal{W})\\
		&= (\mathsf{P}_n-\mathsf{P})(\func^\ast-\func^\ast_\mathcal{W})(\tilde{\func}-\func^\ast_\mathcal{W})\\&\quad+\mathsf{P}(\func^\ast-\func^\ast_\mathcal{W})(\tilde{\func}-\func^\ast_\mathcal{W})-\lambda J(\func^\ast_\mathcal{W},\tilde{\func}-\func^\ast_\mathcal{W})\\&= (\mathsf{P}_n-\mathsf{P})(\func^\ast-\func^\ast_\mathcal{W})(\tilde{\func}-\func^\ast_\mathcal{W})\\& \leq \left|(\mathsf{P}_n-\mathsf{P})(\func^\ast-\func^\ast_\mathcal{W})(\tilde{\func}-\func^\ast_\mathcal{W})\right|\\
		&= \|\func^\ast-\func^\ast_\mathcal{W}\|_\lambda \|\tilde{\func}-\func^\ast_\mathcal{W}\|_\lambda \left|(\mathsf{P}_n-\mathsf{P})\frac{\func^\ast-\func^\ast_\mathcal{W}}{\|\func^\ast-\func^\ast_\mathcal{W}\|_\lambda}\frac{\tilde{\func}-\func^\ast_\mathcal{W}}{\|\tilde{\func}-\func^\ast_\mathcal{W}\|_\lambda}\right|\\
		&\leq \|\mathcal{W}_\lambda\func^\ast\|_\lambda \|\tilde{\func}-\func^\ast_\mathcal{W}\|_\lambda \xi \leq \lambda^{1/2} \sqrt{J(\func^\ast)}\|\tilde{\func}-\func^\ast_\mathcal{W}\|_\lambda \xi, 
	\end{align} where \begin{align}
		&\mathsf{P}(\func^\ast-\func^\ast_\mathcal{W})(\tilde{\func}-\func^\ast_\mathcal{W})-\lambda J(\func^\ast_\mathcal{W},\tilde{\func}-\func^\ast_\mathcal{W})\\&\quad=\langle \func^\ast-\func^\ast_\mathcal{W},\tilde{\func}-\func^\ast_\mathcal{W}\rangle_\lambda-\lambda J(\func^\ast-\func^\ast_\mathcal{W},\tilde{\func}-\func^\ast_\mathcal{W})-\lambda J(\func^\ast_\mathcal{W},\tilde{\func}-\func^\ast_\mathcal{W})\\&\quad=\langle \func^\ast-\func^\ast_\mathcal{W},\tilde{\func}-\func^\ast_\mathcal{W}\rangle_\lambda-\langle\mathcal{W}_\lambda\func^\ast,\tilde{\func}-\func^\ast_\mathcal{W}\rangle_\lambda=\langle \func^\ast-\mathcal{W}_\lambda\func^\ast-\func^\ast_\mathcal{W},\tilde{\func}-\func^\ast_\mathcal{W}\rangle_\lambda=\langle 0,\tilde{\func}-\func^\ast_\mathcal{W}\rangle_\lambda=0.
	\end{align} Note that 
	\begin{align}\label{eqn:thm1_xi_ineq_3}
		\mathsf{P}_n(\tilde{\func}-\func^\ast_\mathcal{W})^2 + \lambda J(\tilde{\func}-\func^\ast_\mathcal{W})&=(\mathsf{P}_n-\mathsf{P})(\tilde{\func}-\func^\ast_\mathcal{W})^2 + \|\tilde{\func}-\func^\ast_\mathcal{W}\|_\lambda^2 \\&\in \left[ (1-\xi)\|\tilde{\func}-\func^\ast_\mathcal{W}\|_\lambda^2,(1+\xi)\|\tilde{\func}-\func^\ast_\mathcal{W}\|_\lambda^2 \right],
	\end{align} where \begin{align}
		|(\mathsf{P}_n-\mathsf{P})(\tilde{\func}-\func^\ast_\mathcal{W})^2| =\|\tilde{\func}-\func^\ast_\mathcal{W}\|_\lambda^2\left|(\mathsf{P}_n-\mathsf{P})\left(\frac{\tilde{\func}-\func^\ast_\mathcal{W}}{\|\tilde{\func}-\func^\ast_\mathcal{W}\|_\lambda}\right)^2 \right| \leq \|\tilde{\func}-\func^\ast_\mathcal{W}\|_\lambda^2 \xi.
	\end{align} Combining \eqref{eqn:thm1_taylor_6} and \eqref{eqn:thm1_xi_ineq_3}, we get \begin{align}
		(1-\xi)\|\tilde{\func}-\func^\ast + \mathcal{W}_\lambda \func^\ast \|_\lambda \leq  \xi \lambda^{1/2} \sqrt{J(\func^\ast)}.
	\end{align} With probability approaching one, we have $    \|\tilde{\func}-\func^\ast + \mathcal{W}_\lambda \func^\ast \|_\lambda \leq 2 \xi \lambda^{1/2} \sqrt{J(\func^\ast)}$, which leads to \begin{align}\label{eqn:thm1_main_res2}
		\|\tilde{\func}-\func^\ast + \mathcal{W}_\lambda \func^\ast \|_\lambda =\bigO_\mathbb{P} \left( \beta_n \lambda^{1/2} J^{1/2}(\func^\ast)\right).
	\end{align}
	
	\subsubsection*{Part 3}
	We discover order of $\xi$. {\color{black}For $\delta \in (0,\infty)$, let $\mathrm{N}\left(\delta, \mathcal{G}, \|\cdot\|_{\mathcal{V}}\right)$ and $\mathrm{D}\left(\delta, \mathcal{G}, \|\cdot\|_{\mathcal{V}}\right)$ denote the $\delta$-covering number and $\delta$-packing number, respectively, of some set $\mathcal{G}$ in some vector space $\mathcal{V}$ with respect to some norm $\|\cdot \|_{\mathcal{V}}$ on $\mathcal{V}$. The entropy bound result used in the proof is provided in Section~\ref{sec:covering}.} For $p_\mathcal{H}, \delta > 0$, define $\mathcal{G}(p_\mathcal{H})=\{f \in \mathcal{H}: \|f\|_{\sup} \leq 1, J(f) \leq p_\mathcal{H}^2 \}$ and the entropy integral \begin{align}
		\mathcal{J}(p_\mathcal{H},\delta)=\int_0^\delta \sqrt{\log\left(\mathrm{D}(\delta',\mathcal{G}(p_\mathcal{H}),\|\cdot\|_{\sup})+1\right)} d \delta' + \delta \sqrt{\log\left(\mathrm{D}^2(\delta,\mathcal{G}(p_\mathcal{H}),\|\cdot\|_{\sup})+1\right)}.
	\end{align} We choose $p_\mathcal{H}=\mathcal{C}_\psi^{-1} \mathcal{C}_\mathcal{R}^{-1/2}\lambda^{1/(4m)-1/2} (-\log\lambda)^{(1-|S_{\sup}|)/2} \rightarrow \infty$. For $f \in \mathcal{H}$, define $\Psi(f,X)= \mathcal{C}_\psi^{-1} \mathcal{C}_\mathcal{R}^{-1/2}\lambda^{1/(4m)} (-\log\lambda)^{(1-|S_{\sup}|)/2}f(X)$. One can see that for any $f,g \in \mathcal{G}(p_\mathcal{H})$, we have \begin{align}
		\| (\Psi(f,X)- \Psi(g,X)) \mathcal{R}_\lambda(X,\cdot) \|_\lambda &= |\Psi(f,X)- \Psi(g,X)| \sqrt{\mathcal{R}_\lambda(X,X)} \\&\leq |f(X)-g(X)| \leq \|f-g\|_{\sup}.
	\end{align} By \cite{liu_2020} and Lemma 6.1 of \cite{pmlr-v125-yang20a}, with a constant $\mathcal{C}_\mathcal{H} \in (0,\infty)$, for any $t \geq 0$, we have \begin{align}\label{eqn:thm_1_etp_1}
		\mathbb{P}\left( \sup_{f \in \mathcal{G}(p_\mathcal{H})} \left\| \frac{1}{\sqrt{n}}\sum_{i=1}^n \left[ \Psi(f,X_i) \mathcal{R}_\lambda(X_i,\cdot) -\mathbb{E}_X\left(\Psi(f,X) \mathcal{R}_\lambda(X,\cdot)\right) \right]    \right\|_\lambda \geq t \right) \leq 2 \exp\left(-\frac{t^2}{ \mathcal{C}_{\mathcal{H}}^2\mathcal{J}^2(p_\mathcal{H},1)} \right).
	\end{align}
	
	We have \begin{align}
		&\mathcal{C}_\psi^{-2}\mathcal{C}_\mathcal{R}^{-1}\lambda^{1/(2m)} (-\log\lambda)^{1-|S_{\sup}|} \sqrt{n} \xi \\&= \sup_{f\in \mathcal{H},\|f\|_\lambda=1}\frac{\mathcal{C}_\psi^{-2}\mathcal{C}_\mathcal{R}^{-1}\lambda^{1/(2m)} (-\log\lambda)^{1-|S_{\sup}|}}{\sqrt{n}} \sup_{g\in \mathcal{H},\|g\|_\lambda=1}\left|\left\langle \sum_{i=1}^n\left[ f(X_i) \mathcal{R}_\lambda(X_i,\cdot)-\mathbb{E}_X(f(X)\mathcal{R}_\lambda(X,\cdot)) \right] , g \right\rangle_\lambda\right|\\
		\\&\leq \sup_{f\in \mathcal{H},\|f\|_\lambda=1}\frac{\mathcal{C}_\psi^{-2}\mathcal{C}_\mathcal{R}^{-1}\lambda^{1/(2m)} (-\log\lambda)^{1-|S_{\sup}|}}{\sqrt{n}} \sup_{g\in \mathcal{H},\|g\|_\lambda=1}  \left\|\sum_{i=1}^n\left[ f(X_i) \mathcal{R}_\lambda(X_i,\cdot)-\mathbb{E}_X(f(X)\mathcal{R}_\lambda(X,\cdot)) \right] \right\|_\lambda \| g \|_\lambda\\
		&=\sup_{f\in \mathcal{H},\|f\|_\lambda=1}\frac{\mathcal{C}_\psi^{-2}\mathcal{C}_\mathcal{R}^{-1}\lambda^{1/(2m)} (-\log\lambda)^{1-|S_{\sup}|}}{\sqrt{n}} \left\| \sum_{i=1}^n\left[ f(X_i) \mathcal{R}_\lambda(X_i,\cdot)-\mathbb{E}_X(f(X)\mathcal{R}_\lambda(X,\cdot)) \right]\right\|_\lambda\\
		&=\sup_{f\in \mathcal{H},\|f\|_\lambda=\mathcal{C}_\psi^{-1} \mathcal{C}_\mathcal{R}^{-1/2}\lambda^{1/(4m)} (-\log\lambda)^{(1-|S_{\sup}|)/2}} \left\| \frac{1}{\sqrt{n}}\sum_{i=1}^n \left[\Psi(f,X_i) \mathcal{R}_\lambda(X_i,\cdot) -\mathbb{E}_X\left(\Psi(f,X) \mathcal{R}_\lambda(X,\cdot)\right) \right]   \right\|_\lambda\\& \leq \sup_{f \in \mathcal{G}(p_\mathcal{H})} \left\| \frac{1}{\sqrt{n}}\sum_{i=1}^n \left[\Psi(f,X_i) \mathcal{R}_\lambda(X_i,\cdot) -\mathbb{E}_X\left(\Psi(f,X) \mathcal{R}_\lambda(X,\cdot)\right) \right]   \right\|_\lambda,\label{eqn:thm_1_etp_2}
	\end{align} where when $f \in \mathcal{H}$ satisfies $\|f\|_\lambda=\mathcal{C}_\psi^{-1} \mathcal{C}_\mathcal{R}^{-1/2}\lambda^{1/(4m)} (-\log\lambda)^{(1-|S_{\sup}|)/2}$, we have $\|f\|_{\sup} \leq 1$ and $\lambda J(f) \leq \|f\|_\lambda^2 = \lambda p_\mathcal{H}^2$, leading to $f \in \mathcal{G}(p_\mathcal{H})$.
	
	Combining \eqref{eqn:thm_1_etp_1} and \eqref{eqn:thm_1_etp_2}, for any $t \geq 0$, we have \begin{align}
		&\mathbb{P}\left( \mathcal{C}_\psi^{-2}\mathcal{C}_\mathcal{R}^{-1}\lambda^{1/(2m)} (-\log\lambda)^{1-|S_{\sup}|} \sqrt{n} \xi \geq t \right) \leq  2 \exp\left(-\frac{t^2}{ \mathcal{C}_{\mathcal{H}}^2\mathcal{J}^2(p_\mathcal{H},1)} \right),\\
		&\mathbb{P}\left( \mathcal{C}_\psi^{-2}\mathcal{C}_\mathcal{R}^{-1}\lambda^{1/(2m)} (-\log\lambda)^{1-|S_{\sup}|} \sqrt{n} \xi \geq \mathcal{C}_\mathcal{H} \sqrt{\log n} \mathcal{J}(p_\mathcal{H},1) \right) \leq  2 \exp(-\log n)=\frac{2}{n} \rightarrow 0,
	\end{align} implying that \begin{align}\label{eq:xi_order}
		\xi=\bigO_\mathbb{P}\left( n^{-1/2} \lambda^{-1/(2m)} (-\log\lambda)^{|S_{\sup}|-1} \mathcal{J}(p_\mathcal{H},1) (\log n)^{1/2} \right).
	\end{align}
	
	Now we get upper bound for $\mathcal{J}(p_\mathcal{H},1)$. Considering $1<p_\mathcal{H} \rightarrow \infty$, We have {\color{black}\begin{align}
			\mathcal{G}(p_\mathcal{H})&=p_\mathcal{H} \{ f \in \mathcal{H}: \|f\|_{\sup} \leq p_\mathcal{H}^{-1}, J(f) \leq 1\}\\
			& \subseteq p_\mathcal{H} \{ f \in \mathcal{H}: \|f\|_{\sup} \leq 1, J(f) \leq 1\},
		\end{align} which leads to \begin{align}
			\log \mathrm{D}\left(\delta, \mathcal{G}(p_\mathcal{H}),\|\cdot\|_{\sup} \right) &\leq \log \mathrm{N}\left(\delta/2, \mathcal{G}(p_\mathcal{H}),\|\cdot\|_{\sup} \right) \\&\leq \log \mathrm{N}\left(\delta/2,  p_\mathcal{H} \{ f \in \mathcal{H}: \|f\|_{\sup} \leq 1, J(f) \leq 1\},\|\cdot\|_{\sup} \right)\\
			&= \log \mathrm{N}\left(p_\mathcal{H}^{-1}\delta/2, \{ f \in \mathcal{H}: \|f\|_{\sup} \leq 1, J(f) \leq 1\},\|\cdot\|_{\sup} \right)\\
			&\leq \mathcal{C}_{\mathcal{H},\tau} \left(\frac{2 p_\mathcal{H}}{\delta}\right)^{2/(2m-\tau-1)}\log\left(\frac{2 p_\mathcal{H}}{\delta}\right)
		\end{align} as $n \rightarrow \infty$ for any $\delta \in (0,\infty)$ by Lemma \ref{lem:entropy} since $0<p_\mathcal{H}^{-1}\delta/2 \rightarrow 0$ as $n \rightarrow \infty$.} Thus, we have \begin{align}
		\mathcal{J}(p_\mathcal{H},1)&=\int_0^1 \sqrt{\log\left(\mathrm{D}(\delta,\mathcal{G}(p_\mathcal{H}),\|\cdot\|_{\sup})+1\right)} d \delta +  \sqrt{\log\left(\mathrm{D}^2(1,\mathcal{G}(p_\mathcal{H}),\|\cdot\|_{\sup})+1\right)} \notag \\
		&\asymp \int_0^1 \sqrt{\log\mathrm{D}(\delta,\mathcal{G}(p_\mathcal{H}),\|\cdot\|_{\sup})} d \delta +  \sqrt{\log\mathrm{D}(1,\mathcal{G}(p_\mathcal{H}),\|\cdot\|_{\sup})}  \\
		&{\color{black} \lesssim p_\mathcal{H}^{1/(2m-\tau-1)} \int_0^1 \delta^{-1/(2m-\tau-1)} \sqrt{\log p_\mathcal{H}+\log(2/\delta)} d\delta + p_\mathcal{H}^{1/(2m-\tau-1)}\sqrt{\log p_\mathcal{H}+\log 2}}\\
		&{\color{black} \leq p_\mathcal{H}^{1/(2m-\tau-1)} \int_0^1 \delta^{-1/(2m-\tau-1)} \left(\sqrt{\log p_\mathcal{H}}+\sqrt{\log(2/\delta)}\right) d\delta + p_\mathcal{H}^{1/(2m-\tau-1)}\left(\sqrt{\log p_\mathcal{H}}+\sqrt{\log 2}\right)}\\&{\color{black} \asymp p_\mathcal{H}^{1/(2m-\tau-1)}\sqrt{\log p_\mathcal{H}}\label{eq:J_order}}\end{align} {\color{black}since $\int_0^1 \delta^{-1/(2m-\tau-1)} d\delta$, $\int_0^1 \delta^{-1/(2m-\tau-1)} \sqrt{\log(2/\delta)} d\delta \in (0,\infty)$ due to $2m-\tau-1>1$.} By \eqref{eq:xi_order} and \eqref{eq:J_order}, we get {\color{black}\begin{align}
			&\xi=\bigO_\mathbb{P}\left( n^{-1/2} \lambda^{-1/(2m)} (-\log\lambda)^{|S_{\sup}|-1}  p_\mathcal{H}^{1/(2m-\tau-1)}(\log p_\mathcal{H})^{1/2} (\log n)^{1/2} \right)=\bigO_\mathbb{P}(\beta_n).\\\label{eq:xi_order2}
	\end{align}}
	
	By \eqref{eqn:thm1_main_res1}, \eqref{eqn:thm1_main_res2}, and \eqref{eq:xi_order2}, we have \begin{align}
		& \left\|\hat{\func}-\func^\ast-\frac{1}{n}\sum_{i=1}^n \epsilon_i \mathcal{R}_\lambda(X_i,\cdot) + \mathcal{W}_\lambda \func^\ast\right\|_\lambda \\& \quad \leq \left\|\hat{\func}-\tilde{\func}-\frac{1}{n}\sum_{i=1}^n \epsilon_i \mathcal{R}_\lambda(X_i,\cdot)\right\|_\lambda + \|\tilde{\func}-\func^\ast + \mathcal{W}_\lambda \func^\ast \|_\lambda = \bigO_\mathbb{P}(\alpha_n).
	\end{align} For each $S \in \mathbb{S}$, we have
	\begin{align}
		&\left\| \hat{\func}_S-\func^\ast_S - \frac{1}{n}\sum_{i=1}^n \epsilon_i \mathcal{R}_{S,\lambda}(X_{iS},\cdot) + \mathcal{W}_{S,\lambda}\func^\ast_S  \right\|_{S,\lambda}^2 \\
		&\quad \leq \sum_{S'\in\mathbb{S}} \left\| \hat{\func}_{S'}-\func^\ast_{S'} - \frac{1}{n}\sum_{i=1}^n \epsilon_i \mathcal{R}_{S',\lambda}(X_{iS'},\cdot) + \mathcal{W}_{S',\lambda}\func^\ast_{S'}  \right\|_{S',\lambda}^2 \\
		& \quad =\left\| \hat{\func}-\func^\ast - \frac{1}{n}\sum_{i=1}^n \epsilon_i \mathcal{R}_{\lambda}(X_i,\cdot) + \mathcal{W}_{\lambda}\func^\ast  \right\|_{\lambda}^2 = \bigO_{\mathbb{P}}(\alpha_n^2),
	\end{align} which completes the proof.\qed
	
	\subsection{Proof of Theorem \ref{thm:rate_of_conv}}
	
	For $S \in \mathbb{S}\setminus\{\emptyset\}$, \begin{align}
		\mathbb{E}\left\| \frac{1}{n}\sum_{i=1}^n \epsilon_i \mathcal{R}_{S,\lambda}(X_{iS},\cdot) \right\|_{S,\lambda}^2 &= \frac{1}{n^2}\sum_{i=1}^n\sum_{i'=1}^n \mathbb{E}(\epsilon_i\epsilon_{i'})\mathbb{E}(\mathcal{R}_{S,\lambda}(X_{iS},X_{i'S}))\\
		&= \frac{\sigma^2}{n^2}\sum_{i=1}^n \mathbb{E}(\mathcal{R}_{S,\lambda}(X_{iS},X_{iS})) = \frac{\sigma^2}{n} \sum_{v\in \mathbb{N}}\frac{1}{1+\lambda/\mu_{S,v}} \\&\asymp n^{-1}\lambda^{-1/(2m)} (-\log\lambda)^{|S|-1}=\bigO\left(n^{-1}\lambda^{-1/(2m)} (-\log\lambda)^{|S|-1}\right)
	\end{align} implies \begin{align}\label{eqn:score_S_order}
		\left\| \frac{1}{n}\sum_{i=1}^n \epsilon_i \mathcal{R}_{S,\lambda}(X_{iS},\cdot) \right\|_{S,\lambda}=\bigO_\mathbb{P}\left(n^{-1/2}\lambda^{-1/(4m)} (-\log\lambda)^{(|S|-1)/2}\right).
	\end{align} We have
	\begin{align}\label{eqn:rate_general_pf}
		\| \hat{\func}_S-\func^\ast_S\|_{S,\lambda} &\leq 
		\left\| \hat{\func}_S-\func^\ast_S - \frac{1}{n}\sum_{i=1}^n \epsilon_i \mathcal{R}_{S,\lambda}(X_{iS},\cdot) + \mathcal{W}_{S,\lambda}\func^\ast_S  \right\|_{S,\lambda}\\
		& +\left\| \frac{1}{n}\sum_{i=1}^n \epsilon_i \mathcal{R}_{S,\lambda}(X_{iS},\cdot) \right\|_{S,\lambda} + \| \mathcal{W}_{S,\lambda}\func^\ast_S\|_{S,\lambda} \\ 
		& = o_\mathbb{P}(n^{-1/2}) + \bigO_\mathbb{P}\left(n^{-1/2}\lambda^{-1/(4m)} (-\log\lambda)^{(|S|-1)/2}\right) + \bigO\left( \lambda^{1/2}J_S^{1/2}(\func^\ast_S)\right)\\&=\bigO_\mathbb{P}\left(n^{-1/2}\lambda^{-1/(4m)} (-\log\lambda)^{(|S|-1)/2}+\lambda^{1/2}J_S^{1/2}(\func^\ast_S)\right).
	\end{align}
	
	{\color{black}If condition \eqref{eqn:supersmooth_in_rate_of_conv} is satisfied, we have\begin{align}
			\|\mathcal{W}_{S,\lambda} f^\ast_S \|_{S,\lambda}^2 &= \sum_{v \in \mathbb{N}} V_S^2(f^\ast_S,\psi_{S,v})\frac{\lambda^2/\mu_{S,v}^2}{(1+\lambda/\mu_{S,v})^2}(1+\lambda/\mu_{S,v})\\
			&= \lambda^2 \sum_{v \in \mathbb{N}} \frac{\mu_{S,v}^{-2}V_S^2(f^\ast_S,\psi_{S,v})}{1+\lambda/\mu_{S,v}} \leq \lambda^2 \sum_{v \in \mathbb{N}}\mu_{S,v}^{-2}V_S^2(f^\ast_S,\psi_{S,v}) \leq  \lambda^2 \mathcal{C}_S^\ast,
		\end{align} resulting in an improvement of \eqref{eqn:rate_general_pf} to $$\|\hat{\func}_S - \func^\ast_S\|_{S,\lambda} = \bigO_{\mathbb{P}}\left( n^{-1/2} \lambda^{-1/(4m)} (-\log\lambda)^{(|S|-1)/2} + \lambda \right).$$}
	
	{\color{black}For the intercept,} we have \begin{align}\label{eqn:fbr_intercept}
		\left| \hat{\func}_\emptyset-\func^\ast_\emptyset - \frac{1}{n}\sum_{i=1}^n \epsilon_i  \right|&=\left\| \hat{\func}_\emptyset-\func^\ast_\emptyset - \frac{1}{n}\sum_{i=1}^n \epsilon_i \mathcal{R}_{\emptyset,\lambda} + \mathcal{W}_{\emptyset,\lambda}\func^\ast_\emptyset  \right\|_{\emptyset,\lambda}\\&=\bigO_\mathbb{P}(\alpha_n)= o_\mathbb{P}(n^{-1/2}),
	\end{align} and \begin{align}
		\mathbb{E}\left| \frac{1}{n} \sum_{i=1}^n \epsilon_i \right|^2 = \frac{\sigma^2}{n}=\bigO(n^{-1}),
	\end{align} implying that $|\sum_{i=1}^n \epsilon_i/n|=\bigO_\mathbb{P}(n^{-1/2})$. We get \begin{align}
		\|\hat{\func}_\emptyset-\func^\ast_\emptyset\|_{\emptyset,\lambda} = |\hat{\func}_\emptyset-\func^\ast_\emptyset| \leq \left| \hat{\func}_\emptyset-\func^\ast_\emptyset - \frac{1}{n}\sum_{i=1}^n \epsilon_i  \right| + \left| \frac{1}{n} \sum_{i=1}^n \epsilon_i \right| = o_\mathbb{P}(n^{-1/2})+\bigO_\mathbb{P}(n^{-1/2})=\bigO_\mathbb{P}(n^{-1/2}).\end{align}\qed
	
	\subsection{Proof of Theorem \ref{thm:local}}
	The proof proceeds similarly to that in \cite{shang_2013, shang_2015}. Let $S \in \mathbb{S}\setminus\{\emptyset\}$. We have \begin{align}\label{eqn:bound1}
		& \left| \hat{\func}_S(x_S)-\func^\ast_S(x_S) - \frac{1}{n}\sum_{i=1}^n \epsilon_i \mathcal{R}_{S,\lambda}(X_{iS},x_S) + \left(\mathcal{W}_{S,\lambda}\func^\ast_S\right)(x_S) \right|\\
		&\quad=\left|\left\langle \hat{\func}_S-\func^\ast_S - \frac{1}{n}\sum_{i=1}^n \epsilon_i \mathcal{R}_{S,\lambda}(X_{iS},\cdot) + \mathcal{W}_{S,\lambda}\func^\ast_S, \enspace\mathcal{R}_{S,\lambda}(x_{S},\cdot) \right\rangle_{S,\lambda}\right| \\ &\quad \leq \left\| \hat{\func}_S-\func^\ast_S - \frac{1}{n}\sum_{i=1}^n \epsilon_i \mathcal{R}_{S,\lambda}(X_{iS},\cdot) + \mathcal{W}_{S,\lambda}\func^\ast_S  \right\|_{S,\lambda} \|\mathcal{R}_{S,\lambda}(x_S,\cdot)\|_{S,\lambda}\\
		& \quad= \bigO_{\mathbb{P}}(\alpha_n) \sqrt{ \mathcal{R}_{S,\lambda}(x_S,x_S) }=\bigO_{\mathbb{P}}(\alpha_n) \sqrt{\sum_{v\in\mathbb{N}}\frac{\psi_{S,v}^2(x_S)}{1+\lambda/\mu_{S,v}}} = \bigO_{\mathbb{P}}(\alpha_n) \bigO\left( \lambda^{-1/(4m)} (-\log\lambda)^{(|S|-1)/2} \right) ,
	\end{align} leading to 
	\begin{align}\label{eqn:vanish1}
		\left|\frac{\sqrt{n}\left( \hat{\func}_S(x_S)-\func^\ast_S(x_S) - \sum_{i=1}^n \epsilon_i \mathcal{R}_{S,\lambda}(X_{iS},x_S)/n + \left(\mathcal{W}_{S,\lambda}\func^\ast_S\right)(x_S)\right) }{\sqrt{ \sigma^2\sum_{v\in\mathbb{N}} \psi_{S,v}^2(x_S)/(1+\lambda/\mu_{S,v})^2}}\right|&=\sqrt{n}\bigO_{\mathbb{P}}(\alpha_n)=o_{\mathbb{P}}(1).
	\end{align} 
	
	We have \begin{align}
		\left|\left(\mathcal{W}_{S,\lambda}\func^\ast_S\right)(x_S)\right|&=\left|\sum_{v\in\mathbb{N}}V_S(\func^\ast_S,\psi_{S,v})\frac{\lambda/\mu_{S,v}}{1+\lambda/\mu_{S,v}}\psi_{S,v}(x_S)\right| \leq \lambda \sqrt{\sum_{v \in \mathbb{N}} \mu_{S,v}^{-2} V_S^2(\func^\ast_S,\psi_{S,v})} \sqrt{\sum_{v\in\mathbb{N}}\frac{\psi_{S,v}^2(x_S)}{(1+\lambda/\mu_{S,v})^2}}\\
		&\leq \lambda \sqrt{\mathcal{C}_{S}^\ast}\sqrt{\sum_{v\in\mathbb{N}}\frac{\psi_{S,v}^2(x_S)}{(1+\lambda/\mu_{S,v})^2}},
	\end{align} leading to \begin{align}\label{eqn:vanish2}
		\left|\frac{\sqrt{n}\left(\mathcal{W}_{S,\lambda}\func^\ast_S\right)(x_S) }{\sqrt{ \sigma^2\sum_{v\in\mathbb{N}} \psi_{S,v}^2(x_S)/(1+\lambda/\mu_{S,v})^2}}\right| \leq \frac{\sqrt{n}\lambda\sqrt{\mathcal{C}_{S}^\ast}}{\sigma} = o(1).
	\end{align}
	
	Define $$W_{S}(x_S)\equiv  \lambda^{1/(4m)} (-\log\lambda)^{(1-|S|)/2} \epsilon \mathcal{R}_{S,\lambda}(X_S,x_S),$$ where $\mathbb{E}\left(W_{S}(x_S)\right)=0$ and $$\mathrm{Var}\left(W_{S}(x_S)\right)=  \lambda^{1/(2m)} (-\log\lambda)^{1-|S|} \sigma^2 \sum_{v\in\mathbb{N}}\frac{\psi_{S,v}^2(x_S)}{(1+\lambda/\mu_{S,v})^2} \asymp 1$$ by \begin{align}
		\mathrm{Var}\left( \epsilon \mathcal{R}_{S,\lambda}(X_S,x_S) \right) &= \mathrm{Var}\left( \mathbb{E}\left( \epsilon \mathcal{R}_{S,\lambda}(X_S,x_S) | X_S \right) \right) + \mathbb{E}\left( \mathrm{Var}\left( \epsilon \mathcal{R}_{S,\lambda}(X_S,x_S) | X_S \right) \right) \\
		&=\sigma^2\mathbb{E}\left(\mathcal{R}_{S,\lambda}^2(X_S,x_S)\right)= \sigma^2 \mathbb{E}\left( \sum_{v\in\mathbb{N}}\sum_{v'\in\mathbb{N}} \frac{\psi_{S,v}(x_S)\psi_{S,v}(X_S)}{1+\lambda/\mu_{S,v}} \frac{\psi_{S,v'}(x_S)\psi_{S,v'}(X_S)}{1+\lambda/\mu_{S,v'}}  \right)\\
		&=\sigma^2  \sum_{v\in\mathbb{N}}\sum_{v'\in\mathbb{N}} \frac{\psi_{S,v}(x_S)}{1+\lambda/\mu_{S,v}} \frac{\psi_{S,v'}(x_S)}{1+\lambda/\mu_{S,v'}} V_S(\psi_{S,v},\psi_{S,v'})  = \sigma^2 \sum_{v\in\mathbb{N}}\frac{\psi_{S,v}^2(x_S)}{(1+\lambda/\mu_{S,v})^2}.
	\end{align} For \( i = 1, \ldots, n \), define \( W_{iS}(x_S) \) using \( \epsilon_i \) and \( X_{iS} \) in the same manner as \( W_S(x_S) \) is defined. Now, we verify the Lindeberg condition for \(\left\{W_{iS}(x_S)\right\}_{i=1}^n\), which are $n$ independent and identically distributed copies of $W_S(x_S)$. Each \( W_{iS}(x_S) \) depends on \( n \) because \( \lambda \) is \( n \)-dependent, thereby forming a triangular array structure. For any constant $\delta \in (0,\infty)$,
	\begin{align}
		&\frac{n \mathbb{E}\left[ W_S^2(x_S) \mathds{1}\left\{ W_S^2(x_S) > \delta^2 n \mathrm{Var}\left(W_{S}(x_S)\right) \right\} \right]}{n \mathrm{Var}\left(W_{S}(x_S)\right) } =\frac{\mathbb{E}\left[ W_S^2(x_S) \mathds{1}\left\{ W_S^2(x_S) > \delta^2 n \mathrm{Var}\left(W_{S}(x_S)\right) \right\} \right]}{\mathrm{Var}\left(W_{S}(x_S)\right) }\\
		& \qquad \leq \frac{ \sqrt{\mathbb{E}\left(W_S^4(x_S)\right)} \sqrt{\mathbb{E}\left(\mathds{1}^2\left\{ W_S^2(x_S) > \delta^2 n \mathrm{Var}\left(W_{S}(x_S)\right) \right\} \right)}   } {\mathrm{Var}\left(W_{S}(x_S)\right) } \\& \qquad= \bigO\left(n^{-1/2} \lambda^{-1/(2m)} (-\log\lambda)^{|S|-1} \right)=  \bigO\left(n^{-1/2} \lambda^{-1/(2m)} (-\log\lambda)^{|S_{\sup}|-1} \right) = o(1),
	\end{align} where \begin{align}
		\mathcal{R}_{S,\lambda}^4(X_S,x_S) = \left| \sum_{v\in\mathbb{N}} \frac{\psi_{S,v}(X_S)\psi_{S,v}(x_S)}{1+\lambda/\mu_{S,v}} \right|^4 \leq \left( \sum_{v\in\mathbb{N}} \frac{\left|\psi_{S,v}(X_S)\right|\left|\psi_{S,v}(x_S)\right|}{1+\lambda/\mu_{S,v}} \right)^4\leq \mathcal{C}_\psi^8\left(\sum_{v\in\mathbb{N}}\frac{1}{1+\lambda/\mu_{S,v}}\right)^4,
	\end{align} \begin{align}
		\mathbb{E}\left(W_S^4(x_S)\right) &= \lambda^{1/m} (-\log\lambda)^{2-2|S|} \mathbb{E}( \epsilon^4 ) \mathbb{E}\left(\mathcal{R}_{S,\lambda}^4(X_S,x_S)\right) \\&\leq \mathcal{C}_\psi^8 \lambda^{1/m} (-\log\lambda)^{2-2|S|} \mathbb{E}( \epsilon^4 )\left(\sum_{v\in\mathbb{N}}\frac{1}{1+\lambda/\mu_{S,v}}\right)^4 = \bigO\left(\lambda^{-1/m} (-\log\lambda)^{2|S|-2}\right),
	\end{align} and \begin{align}
		\mathbb{P}\left( W_S^2(x_S) > \delta^2 n \mathrm{Var}\left(W_{S}(x_S)\right)  \right) \leq \frac{ \mathbb{E}\left( W_S^2(x_S)\right) }{\delta^2 n \mathrm{Var}\left(W_{S}(x_S)\right) } = \delta^{-2} n^{-1}.
	\end{align} By the Lindeberg central limit theorem, we have
	\begin{align}\label{eqn:clt1}
		&\frac{\sqrt{n} \sum_{i=1}^n \epsilon_i \mathcal{R}_{S,\lambda}(X_{iS},x_S)/n  }{\sqrt{ \sigma^2\sum_{v\in\mathbb{N}} \psi_{S,v}^2(x_S)/(1+\lambda/\mu_{S,v})^2}}\\&\quad=
		\frac{ \lambda^{1/(4m)} (-\log\lambda)^{(1-|S|)/2}\sum_{i=1}^n \epsilon_i \mathcal{R}_{S,\lambda}(X_{iS},x_S)  }{\sqrt{ n \lambda^{1/(2m)} (-\log\lambda)^{1-|S|}\sigma^2\sum_{v\in\mathbb{N}} \psi_{S,v}^2(x_S)/(1+\lambda/\mu_{S,v})^2}}\dgoto \mathrm{N}(0,1) \text{ as } n\rightarrow\infty.
	\end{align} 
	
	Combining \eqref{eqn:vanish1}, \eqref{eqn:vanish2}, and \eqref{eqn:clt1}, by the Slutsky's theorem, we have
	\begin{align}\label{eqn:bias_vanishing}
		&\frac{\sqrt{n}\left(\hat{\func}_S(x_S)-\func^\ast_S(x_S)\right)}{\sqrt{ \sigma^2\sum_{v\in\mathbb{N}} \psi_{S,v}^2(x_S)/(1+\lambda/\mu_{S,v})^2}} \\
		&\quad =\frac{\sqrt{n} \sum_{i=1}^n \epsilon_i \mathcal{R}_{S,\lambda}(X_{iS},x_S)/n  }{\sqrt{ \sigma^2\sum_{v\in\mathbb{N}} \psi_{S,v}^2(x_S)/(1+\lambda/\mu_{S,v})^2}} +o_{\mathbb{P}}(1) +o(1)  \dgoto \mathrm{N}(0,1) \text{ as } n\rightarrow\infty.
	\end{align}
	
	We now derive the limiting distribution of the intercept. By \eqref{eqn:fbr_intercept}, we have
	\begin{align}
		\frac{\sqrt{n}}{\sigma}\left| \hat{\func}_\emptyset-\func^\ast_\emptyset - \frac{1}{n}\sum_{i=1}^n \epsilon_i  \right|  = \sqrt{n}o_{\mathbb{P}}(n^{-1/2})=o_{\mathbb{P}}(1),
	\end{align} which leads to \begin{align}
		\frac{\sqrt{n}\left(\hat{\func}_\emptyset-\func^\ast_\emptyset\right)}{\sigma}=\frac{\sqrt{n}\sum_{i=1}^n\epsilon_i/n}{\sigma} +o_{\mathbb{P}}(1)  \dgoto \mathrm{N}(0,1) \text{ as } n\rightarrow\infty
	\end{align} by the central limit theorem and the Slutsky's theorem.\qed
	
	\subsection{Proof of Theorem \ref{thm:global_null}}
	
	The proof is based on derivations from \cite{liu_2020}. Let $S \in \mathbb{S}\setminus \{\emptyset\}$.
	
	\subsubsection*{Part 1}
	
	Define $W_S=\sum_{1 \leq i_1 < i_2 \leq n}W_S(i_1,i_2)$, where $W_S(i_1,i_2) = 2 \epsilon_{i_1}\epsilon_{i_2}\mathcal{R}_{S,\lambda}(X_{i_1 S},X_{i_2 S})$. Note that $\mathbb{E}(W_S) = \sum_{1 \leq i_1 < i_2 \leq n} 2 \mathbb{E}(\epsilon_{i_1})\mathbb{E}(\epsilon_{i_2})\mathbb{E}\left(\mathcal{R}_{S,\lambda}(X_{i_1 S},X_{i_2 S})\right) =0 $ and $\mathrm{Var}(W_S)$ is \begin{align}
		\mathbb{E}(W_S^2)&=\sum_{1 \leq i_1 < i_2 \leq n} \sum_{1 \leq i_1' < i_2' \leq n} \mathbb{E}\left(W_S(i_1,i_2)W_S(i_1',i_2')\right)=4 \sum_{1 \leq i_1 < i_2 \leq n} \mathbb{E}(\epsilon_{i_1}^2)\mathbb{E}(\epsilon_{i_2}^2) \mathbb{E}\left(\mathcal{R}_{S,\lambda}^2(X_{i_1 S},X_{i_2 S})\right) \\&\quad+ 4\sum_{1 \leq i_1 < i_2 \leq n} \sum_{1 \leq i_1' < i_2' \leq n} \mathds{1}(i_1 \neq i_1'\text{ or }i_2 \neq i_2') \mathbb{E}( \epsilon_{i_1}\epsilon_{i_2}\epsilon_{i_1'}\epsilon_{i_2'})\mathbb{E}\left(\mathcal{R}_{S,\lambda}(X_{i_1 S},X_{i_2 S})\mathcal{R}_{S,\lambda}(X_{i_1' S},X_{i_2' S})\right) \\
		&=4 \sigma^4 \sum_{1 \leq i_1 < i_2 \leq n} \sum_{v \in \mathbb{N}} \sum_{v' \in \mathbb{N}} \frac{\mathbb{E}(\psi_{S,v}(X_{i_1S})\psi_{S,v'}(X_{i_1S}))\mathbb{E}(\psi_{S,v}(X_{i_2S})\psi_{S,v'}(X_{i_2S}))}{(1+\lambda/\mu_{S,v})(1+\lambda/\mu_{S,v'})} \\
		&=2 \sigma^4 n(n-1)  \sum_{v\in\mathbb{N}}\frac{1}{(1+\lambda/\mu_{S,v})^2} \asymp n^2 \lambda^{-1/(2m)} (-\log\lambda)^{|S|-1}, \end{align}where all possible combinations of active indices in the second line are included in at least one of the following cases: \( i_1 < i_1' < i_2' \), \( i_1' < i_1 < i_2 \), \( i_1 < i_2 < i_2' \), or \( i_1' < i_2' < i_2 \). In each case, the expectation involves at least one term \( \mathbb{E}(\epsilon_\ell) = 0 \), which causes the entire term to vanish. 
	
	Define \begin{align}
		&\mathrm{G}_{I}^S = \sum_{1\leq i_1 < i_2 \leq n} \mathbb{E}(W_S^4(i_1,i_2)),\\
		&\mathrm{G}_{II}^S = \sum_{1\leq i_1 < i_2 < i_3 \leq n}\left[ \mathbb{E}(W_S^2(i_1,i_2)W_S^2(i_1,i_3)) + \mathbb{E}(W_S^2(i_2,i_1)W_S^2(i_2,i_3)) + \mathbb{E}(W_S^2(i_3,i_1)W_S^2(i_3,i_2)) \right],\\
		&\mathrm{G}_{IV}^S = \sum_{1\leq i_1 < i_2 < i_3 < i_4 \leq n}\big[ \mathbb{E}(W_S(i_1,i_2)W_S(i_1,i_3)W_S(i_4,i_2)W_S(i_4,i_3))\\
		& \qquad\qquad\qquad\qquad\quad +\mathbb{E}(W_S(i_1,i_2)W_S(i_1,i_4)W_S(i_3,i_2)W_S(i_3,i_4))\\
		& \qquad\qquad\qquad\qquad\quad +\mathbb{E}(W_S(i_1,i_3)W_S(i_1,i_4)W_S(i_2,i_3)W_S(i_2,i_4))\big].
	\end{align} For $\mathrm{G}_{I}^S$, as \begin{align}
		\mathbb{E}(W_S^4(i_1,i_2)) &= 16 \mathbb{E}(\epsilon_{i_1}^4)\mathbb{E}(\epsilon_{i_2}^4) \mathbb{E} \left|\sum_{v \in \mathbb{N}} \frac{\psi_{S,v}(X_{i_1S})\psi_{S,v}(X_{i_2
				S})}{1+\lambda/\mu_{S,v}} \right|^4 \\
		& \leq 16 \mathbb{E}(\epsilon_{i_1}^4)\mathbb{E}(\epsilon_{i_2}^4) \mathbb{E} \left(\sum_{v \in \mathbb{N}} \frac{|\psi_{S,v}(X_{i_1S})||\psi_{S,v}(X_{i_2
				S})|}{1+\lambda/\mu_{S,v}} \right)^4 \\
		& \leq 16 \mathcal{C}_\psi^8 \left(\mathbb{E}(\epsilon^4)\right)^2 \left(\sum_{v \in \mathbb{N}}\frac{1}{1+ \lambda/\mu_{S,v}}\right)^4 \asymp \lambda^{-2/m} (-\log\lambda)^{4|S|-4},
	\end{align} we have $\mathrm{G}_{I}^S = \bigO(n^2 \lambda^{-2/m} (-\log\lambda)^{4|S|-4} )$. For $\mathrm{G}_{II}^S$, we have \begin{align}
		\mathbb{E}(W_S^2(i_1,i_2)W_S^2(i_1,i_3)) &=16 \mathbb{E}(\epsilon_{i_1}^4)\mathbb{E}(\epsilon_{i_2}^2)\mathbb{E}(\epsilon_{i_3}^2) 
		\mathbb{E}\left|\sum_{v \in \mathbb{N}} \frac{\psi_{S,v}(X_{i_1S})\psi_{S,v}(X_{i_2
				S})}{1+\lambda/\mu_{S,v}} \right|^2  \left|\sum_{v \in \mathbb{N}} \frac{\psi_{S,v}(X_{i_1S})\psi_{S,v}(X_{i_3
				S})}{1+\lambda/\mu_{S,v}} \right|^2\\
		&\leq 16 \sigma^4 \mathbb{E}(\epsilon^4) \mathbb{E} \left(\sum_{v \in \mathbb{N}} \frac{|\psi_{S,v}(X_{i_1S})||\psi_{S,v}(X_{i_2
				S})|}{1+\lambda/\mu_{S,v}} \right)^2 \left(\sum_{v \in \mathbb{N}} \frac{|\psi_{S,v}(X_{i_1S})||\psi_{S,v}(X_{i_3
				S})|}{1+\lambda/\mu_{S,v}} \right)^2 \\
		& \leq 16 \mathcal{C}_\psi^8 \sigma^4 \mathbb{E}(\epsilon^4) \left(\sum_{v \in \mathbb{N}}\frac{1}{1+ \lambda/\mu_{S,v}}\right)^4 \asymp \lambda^{-2/m} (-\log\lambda)^{4|S|-4},
	\end{align} where the other terms in the summation share the same structure and thus the same order, leading to $\mathrm{G}_{II}^S = \bigO(n^3 \lambda^{-2/m} (-\log\lambda)^{4|S|-4} )$. For $\mathrm{G}_{IV}^S$, we have \begin{align}
		&\mathbb{E}(W_S(i_1,i_2)W_S(i_1,i_3)W_S(i_4,i_2)W_S(i_4,i_3))=16 \mathbb{E}(\epsilon_{i_1}^2)\mathbb{E}(\epsilon_{i_2}^2)\mathbb{E}(\epsilon_{i_3}^2)\mathbb{E}(\epsilon_{i_4}^2) \sum_{v\in\mathbb{N}}\sum_{v'\in\mathbb{N}}\sum_{v''\in\mathbb{N}}\sum_{v'''\in\mathbb{N}} \\ & \quad\frac{ \mathbb{E}(\psi_{S,v}(X_{i_1S})\psi_{S,v'}(X_{i_1S})) \mathbb{E}(\psi_{S,v}(X_{i_2S})\psi_{S,v''}(X_{i_2S})) \mathbb{E}(\psi_{S,v'}(X_{i_3S})\psi_{S,v'''}(X_{i_3S})) \mathbb{E}(\psi_{S,v''}(X_{i_4S})\psi_{S,v'''}(X_{i_4S}))  }{(1+\lambda/\mu_{S,v})(1+\lambda/\mu_{S,v'})(1+\lambda/\mu_{S,v''})(1+\lambda/\mu_{S,v'''})}\\
		& \quad= 16 \sigma^8  \sum_{v \in \mathbb{N}}\frac{1}{(1+ \lambda/\mu_{S,v})^4} \leq 16 \sigma^8 \sum_{v \in \mathbb{N}}\frac{1}{1+ \lambda/\mu_{S,v}} \asymp \lambda^{-1/(2m)} (-\log\lambda)^{|S|-1},
	\end{align} where the remaining terms in the summation exhibit the same structural form and therefore contribute at the same order, resulting in $\mathrm{G}_{IV}^S = \bigO(n^4 \lambda^{-1/(2m)} (-\log\lambda)^{|S|-1} )$. 
	
	We get \begin{align}
		&\frac{\mathrm{G}_{II}^S}{\left(\mathrm{Var}(W_S)\right)^2} = \bigO(n^{-1} \lambda^{-1/m} (-\log\lambda)^{2|S|-2} ) = \bigO(n^{-1} \lambda^{-1/m} (-\log \lambda)^{2|S_{\sup}|-2})= o(1) ,\\
		&\frac{\mathrm{G}_{I}^S}{\left(\mathrm{Var}(W_S)\right)^2} = \bigO(n^{-2} \lambda^{-1/m} (-\log\lambda)^{2|S|-2} ) = o(1),\\
		&\frac{\mathrm{G}_{IV}^S}{\left(\mathrm{Var}(W_S)\right)^2}=\bigO(\lambda^{1/(2m)} (-\log\lambda)^{1-|S|}) = o(1),
	\end{align} and thus by the Proposition 3.2 from \cite{dejong_1987}, we have \begin{align}\label{eqn:quad_clt}
		\frac{W_S}{\sqrt{\mathrm{Var}(W_S)}} \dgoto \mathrm{N}(0,1) \text{ as }n \rightarrow \infty,
	\end{align} which leads to \begin{align}\label{eqn:W_S_order}
		W_S=\bigO_\mathbb{P}\left(\sqrt{\mathrm{Var}(W_S)}\right)=\bigO_\mathbb{P}(n \lambda^{-1/(4m)} (-\log\lambda)^{(|S|-1)/2}).
	\end{align}
	
	\subsubsection*{Part 2}
	
	We now demonstrate that the proposed test statistic converges asymptotically to \( W_S \). We have \begin{align}\label{eqn:wd_null_pt2_res1}
		&\|\hat{\func}_S\|_{S,\lambda}^2 - \left\|\frac{1}{n}\sum_{i=1}^n \epsilon_i \mathcal{R}_{S,\lambda}(X_{iS},\cdot)\right\|_{S,\lambda}^2 \\&\quad= \left\|\hat{\func}_S - \frac{1}{n}\sum_{i=1}^n \epsilon_i \mathcal{R}_{S,\lambda}(X_{iS},\cdot)  \right\|_{S,\lambda}^2 +  2 \left\langle \frac{1}{n}\sum_{i=1}^n \epsilon_i \mathcal{R}_{S,\lambda}(X_{iS},\cdot), \hat{\func}_S - \frac{1}{n}\sum_{i=1}^n \epsilon_i \mathcal{R}_{S,\lambda}(X_{iS},\cdot) \right\rangle_{S,\lambda} \\ 
		& \quad= \bigO_\mathbb{P}(\alpha_n^2) + \bigO_\mathbb{P}\left(\alpha_n n^{-1/2}\lambda^{-1/(4m)} (-\log\lambda)^{(|S|-1)/2}\right) = \bigO_\mathbb{P}\left(\alpha_n n^{-1/2}\lambda^{-1/(4m)} (-\log\lambda)^{(|S|-1)/2}\right),
	\end{align} where \begin{align}
		&\left| \left\langle \frac{1}{n}\sum_{i=1}^n \epsilon_i \mathcal{R}_{S,\lambda}(X_{iS},\cdot), \hat{\func}_S - \frac{1}{n}\sum_{i=1}^n \epsilon_i \mathcal{R}_{S,\lambda}(X_{iS},\cdot) \right\rangle_{S,\lambda} \right|\\ 
		& \quad \leq \left\|\frac{1}{n}\sum_{i=1}^n \epsilon_i \mathcal{R}_{S,\lambda}(X_{iS},\cdot)\right\|_{S,\lambda}  \left\|\hat{\func}_S - \frac{1}{n}\sum_{i=1}^n \epsilon_i \mathcal{R}_{S,\lambda}(X_{iS},\cdot)  \right\|_{S,\lambda}= \bigO_\mathbb{P}\left( n^{-1/2}\lambda^{-1/(4m)} (-\log\lambda)^{(|S|-1)/2}\alpha_n\right)
	\end{align} by \eqref{eqn:score_S_order} and Theorem \ref{thm:fbr} due to \( \func^\ast_S = 0 \). 
	
	We get \begin{align}\label{eqn:wd_null_pt2_res2}
		\left\|\frac{1}{n}\sum_{i=1}^n \epsilon_i \mathcal{R}_{S,\lambda}(X_{iS},\cdot)\right\|_{S,\lambda}^2 &= \frac{1}{n^2}W_S + \frac{1}{n^2}\sum_{i=1}^n \epsilon_i^2 \mathcal{R}_{S,\lambda}(X_{iS},X_{iS})\\
		&= \frac{1}{n^2}W_S + \frac{\sigma^2}{n}  \sum_{v\in \mathbb{N}}\frac{1}{1+\lambda/\mu_{S,v}} + \bigO_\mathbb{P}\left( n^{-3/2} \lambda^{-1/(2m)} (-\log\lambda)^{|S|-1} \right),
	\end{align} where we used \begin{align}
		\mathbb{E}\left( \frac{1}{n^2}\sum_{i=1}^n \epsilon_i^2 \mathcal{R}_{S,\lambda}(X_{iS},X_{iS}) \right) = \frac{\sigma^2}{n}  \sum_{v\in \mathbb{N}}\frac{1}{1+\lambda/\mu_{S,v}}
	\end{align} and \begin{align}
		&\mathbb{E}\left( \frac{1}{n^2}\sum_{i=1}^n \epsilon_i^2 \mathcal{R}_{S,\lambda}(X_{iS},X_{iS}) -  \frac{\sigma^2}{n}  \sum_{v\in \mathbb{N}}\frac{1}{1+\lambda/\mu_{S,v}} \right)^2 = \mathrm{Var}\left(\frac{1}{n^2}\sum_{i=1}^n \epsilon_i^2 \mathcal{R}_{S,\lambda}(X_{iS},X_{iS})\right) \\
		& \quad = \frac{1}{n^4} n \mathrm{Var}\left(\epsilon^2 \mathcal{R}_{S,\lambda}(X_S,X_S) \right) \leq \frac{1}{n^3} \mathbb{E} \left( \epsilon^4 \mathcal{R}_{S,\lambda}^2(X_S,X_S) \right) = \frac{1}{n^3} \mathbb{E} (\epsilon^4)\mathbb{E}\left( \mathcal{R}_{S,\lambda}^2(X_S,X_S) \right)\\
		& \quad \lesssim n^{-3} \lambda^{-1/m} (-\log\lambda)^{2|S|-2}.
	\end{align} 
	
	So, using \eqref{eqn:wd_null_pt2_res1} and \eqref{eqn:wd_null_pt2_res2}, we have \begin{align}
		&\frac{n^2\left( \|\hat{\func}_S\|_{S,\lambda}^2 -\sigma^2\sum_{v\in\mathbb{N}}(1+\lambda/\mu_{S,v})^{-1}/n  \right)}{\sqrt{2 \sigma^4 n(n-1)  \sum_{v\in\mathbb{N}}(1+\lambda/\mu_{S,v})^{-2}}}\\ &\quad= \frac{W_S}{\sqrt{\mathrm{Var}(W_S)}} + \frac{\bigO_\mathbb{P}\left(\alpha_n n^{3/2}\lambda^{-1/(4m)} (-\log\lambda)^{(|S|-1)/2}\right)}{{\sqrt{2 \sigma^4 n(n-1)  \sum_{v\in\mathbb{N}}(1+\lambda/\mu_{S,v})^{-2}}}} + \frac{n^{1/2} \lambda^{-1/(2m)} (-\log\lambda)^{|S|-1}}{{\sqrt{2 \sigma^4 n(n-1)  \sum_{v\in\mathbb{N}}(1+\lambda/\mu_{S,v})^{-2}}}}\\ &\quad= \frac{W_S}{\sqrt{\mathrm{Var}(W_S)}} + \bigO_\mathbb{P}(\alpha_n n^{1/2}) + \bigO_\mathbb{P}\left(n^{-1/2} \lambda^{-1/(4m)} (-\log\lambda)^{(|S|-1)/2}\right) \\&\quad= \frac{W_S}{\sqrt{\mathrm{Var}(W_S)}} + o_\mathbb{P}(1) \dgoto \mathrm{N}(0,1) \text{ as } n\rightarrow\infty,
	\end{align}where the last line follows from \eqref{eqn:quad_clt} and the Slutsky's theorem.\qed
	
	\subsection{Proof of Theorem \ref{thm:global_power}}
	
	The structure of the proof aligns with the derivation in \cite{liu_2020}. Let $S \in \mathbb{S}\setminus\{\emptyset\}$. Note that although $\func^\ast_{S(n)}$ is allowed to depend on $n$, since $J_S(\func^\ast_{S(n)}) \leq \mathcal{C}_{J,S}$ by Assumption \ref{asp:smooth}, $\sqrt{n}\alpha_n = o(1)$ and Theorem \ref{thm:fbr} is still valid here. Denote $\func_{S,\lambda} =\hat{\func}_S-\func^\ast_{S(n)} - \sum_{i=1}^n \epsilon_i \mathcal{R}_{S,\lambda}(X_{iS},\cdot)/n + \mathcal{W}_{S,\lambda}\func^\ast_{S(n)}$ and $$\zeta_{S,\lambda} =\frac{n^2}{\sqrt{2 \sigma^4 n(n-1)  \sum_{v\in\mathbb{N}}(1+\lambda/\mu_{S,v})^{-2}}} \asymp n \lambda^{1/(4m)} (-\log \lambda)^{(1-|S|)/2}.$$ Considering that $\mathcal{D}_{S,\lambda}^2 \asymp  n^{-1} \lambda^{-1/(4m)} (-\log \lambda)^{(|S|-1)/2} + \lambda J_S(\func^\ast_{S(n)})$ and $\zeta_{S,\lambda}^{-1}=\bigO(\mathcal{D}_{S,\lambda}^2)$, there exist a constant $\mathcal{C}_{S,\mu} \in (0,\infty)$ and $N_{S,\mu} \in \mathbb{N}$ such that \begin{align}\label{eqn:power_bound_0}\zeta_{S,\lambda}\mathcal{D}_{S,\lambda}^2 \geq \mathcal{C}_{S,\mu}^{-1} \text{ for all } n \geq N_{S,\mu}.\end{align} 
	
	We have \begin{align}
		\mathcal{T}_{S,\lambda} &= \zeta_{S,\lambda}\left( \left\| \func_{S,\lambda} + \frac{1}{n}\sum_{i=1}^n \epsilon_i \mathcal{R}_{S,\lambda}(X_{iS},\cdot)  +\func^\ast_{S(n)} - \mathcal{W}_{S,\lambda}\func^\ast_{S(n)} \right\|_{S,\lambda}^2 -\frac{\sigma^2}{n}\sum_{v\in\mathbb{N}}\frac{1}{1+\lambda/\mu_{S,v}}\right) \\
		&  = \zeta_{S,\lambda} \left(\left\| \frac{1}{n}\sum_{i=1}^n \epsilon_i \mathcal{R}_{S,\lambda}(X_{iS},\cdot) \right\|_{S,\lambda}^2 -\frac{\sigma^2}{n}\sum_{v\in\mathbb{N}}\frac{1}{1+\lambda/\mu_{S,v}} \right) + \zeta_{S,\lambda}\| \func^\ast_{S(n)} - \mathcal{W}_{S,\lambda}\func^\ast_{S(n)} \|_{S,\lambda}^2\\
		& \quad + 2\zeta_{S,\lambda} \frac{1}{n} \sum_{i=1}^n \epsilon_i (\func^\ast_{S(n)} - \mathcal{W}_{S,\lambda}\func^\ast_{S(n)})(X_{iS}) + 2\zeta_{S,\lambda} \frac{1}{n} \sum_{i=1}^n \epsilon_i \func_{S,\lambda}(X_{iS}) \\
		& \quad + 2\zeta_{S,\lambda} \langle \func^\ast_{S(n)} - \mathcal{W}_{S,\lambda}\func^\ast_{S(n)}, \func_{S,\lambda}\rangle_{S,\lambda} + \zeta_{S,\lambda}\|\func_{S,\lambda}\|_{S,\lambda}^2  \\
		& \quad=T_{S,1}+T_{S,2}+T_{S,3}+T_{S,4}+T_{S,5}+\zeta_{S,\lambda}\|\func_{S,\lambda}\|_{S,\lambda}^2.
	\end{align} For $T_{S,1}$, we have \begin{align}
		T_{S,1}&= \zeta_{S,\lambda}\left( \frac{1}{n^2}W_S + \bigO_\mathbb{P}\left( n^{-3/2} \lambda^{-1/(2m)} (-\log\lambda)^{|S|-1} \right) \right) \\
		&= \bigO_\mathbb{P}(1) + \bigO_\mathbb{P}(n^{-1/2} \lambda^{-1/(4m)} (-\log\lambda)^{(|S|-1)/2}) =  \bigO_\mathbb{P}(1) +  o_\mathbb{P}(1) =  \bigO_\mathbb{P}(1),\end{align} where we used \eqref{eqn:quad_clt}, \eqref{eqn:W_S_order}, and \eqref{eqn:wd_null_pt2_res2}, which are not affected by the hypothesis or the values of $\func^\ast_{S(n)}$. So, there exist a constant $\mathcal{C}_{S,1}^\delta \in (0,\infty)$ and $N_{S,1}^\delta \in \mathbb{N}$ such that \begin{align}\label{eqn:power_bound_1}\mathbb{P}(|T_{S,1}| \leq \mathcal{C}_{S,1}^\delta) \geq 1-\delta/4 \text{ for all } n \geq N_{S,1}^\delta.\end{align} For $T_{S,3}$, with any $t>0$, we have \begin{align}
		&\mathbb{P}\left( \left| T_{S,3} \right| \geq \frac{\zeta_{S,\lambda}\| \func^\ast_{S(n)} - \mathcal{W}_{S,\lambda}\func^\ast_{S(n)} \|_{S,\lambda}}{ t \sqrt{n} } \right)\\
		& \quad \leq t^2 \frac{4}{n} \frac{  \mathbb{E}\left| \sum_{i=1}^n \epsilon_i (\func^\ast_{S(n)} - \mathcal{W}_{S,\lambda}\func^\ast_{S(n)})(X_{iS}) \right|^2  }{\| \func^\ast_{S(n)} - \mathcal{W}_{S,\lambda}\func^\ast_{S(n)}\|_{S,\lambda}^2} = 4 \sigma^2 t^2 \frac{ V_S(\func^\ast_{S(n)} - \mathcal{W}_{S,\lambda}\func^\ast_{S(n)})}{\| \func^\ast_{S(n)} - \mathcal{W}_{S,\lambda}\func^\ast_{S(n)}\|_{S,\lambda}^2} \leq 4 \sigma^2 t^2,
	\end{align} where plugging in $t=\lambda^{1/(8m)} (-\log\lambda)^{(1-|S|)/4}=o(1)$ leads to $$T_{S,3}=\bigO_\mathbb{P}\left(\zeta_{S,\lambda}\| \func^\ast_{S(n)} - \mathcal{W}_{S,\lambda}\func^\ast_{S(n)} \|_{S,\lambda} \mathcal{D}_{S,\lambda} \right).$$ So, there exist a constant $\mathcal{C}_{S,3}^\delta \in (0,\infty)$ and $N_{S,3}^\delta \in \mathbb{N}$ such that \begin{align}\label{eqn:power_bound_3}\mathbb{P}\left(|T_{S,3}| \leq \mathcal{C}_{S,3}^\delta \zeta_{S,\lambda}\| \func^\ast_{S(n)} - \mathcal{W}_{S,\lambda}\func^\ast_{S(n)} \|_{S,\lambda} \mathcal{D}_{S,\lambda} \right) \geq 1-\delta/4 \text{ for all } n \geq N_{S,3}^\delta.\end{align} For $T_{S,4}$, we have \begin{align}
		\left|T_{S,4}\right| &= 2\zeta_{S,\lambda} \left| \left\langle \func_{S,\lambda}, \frac{1}{n} \sum_{i=1}^n \epsilon_i \mathcal{R}_{S,\lambda}(X_{iS},\cdot)  \right\rangle_{S,\lambda}  \right| \leq2\zeta_{S,\lambda} \|\func_{S,\lambda}\|_{S,\lambda} \left\| \frac{1}{n}\sum_{i=1}^n \epsilon_i \mathcal{R}_{S,\lambda}(X_{iS},\cdot) \right\|_{S,\lambda} \\
		& = \bigO(n \lambda^{1/(4m)} (-\log \lambda)^{(1-|S|)/2}) o_\mathbb{P}(n^{-1/2}) \bigO_\mathbb{P}\left(n^{-1/2}\lambda^{-1/(4m)} (-\log\lambda)^{(|S|-1)/2}\right)= o_\mathbb{P}(1),
	\end{align} where for any constant $\mathcal{C}_{S,4}^\delta \in (0, \infty)$, chosen sufficiently small, there exists $N_{S,4}^\delta \in \mathbb{N}$ such that \begin{align}\label{eqn:power_bound_4}\mathbb{P}(|T_{S,4}| \leq \mathcal{C}_{S,4}^\delta) \geq 1-\delta/4 \text{ for all } n \geq N_{S,4}^\delta.\end{align} For $T_{S,5}$, we have \begin{align}
		\left|T_{S,5}\right| &= 2\zeta_{S,\lambda}\left| \langle \func^\ast_{S(n)} - \mathcal{W}_{S,\lambda}\func^\ast_{S(n)}, \func_{S,\lambda}\rangle_{S,\lambda}\right| \leq 2\zeta_{S,\lambda}\|\func^\ast_{S(n)} - \mathcal{W}_{S,\lambda}\func^\ast_{S(n)}\|_{S,\lambda}\|\func_{S,\lambda}\|_{S,\lambda}\\
		&= o_\mathbb{P}\left( \zeta_{S,\lambda}\|\func^\ast_{S(n)} - \mathcal{W}_{S,\lambda}\func^\ast_{S(n)}\|_{S,\lambda} n^{-1/2}\right) = o_\mathbb{P}\left(\zeta_{S,\lambda}\| \func^\ast_{S(n)} - \mathcal{W}_{S,\lambda}\func^\ast_{S(n)} \|_{S,\lambda} \mathcal{D}_{S,\lambda} \right),
	\end{align} where for any constant $\mathcal{C}_{S,5}^\delta \in (0,\infty)$, chosen sufficiently small, there exists $N_{S,5}^\delta \in \mathbb{N}$ such that \begin{align}\label{eqn:power_bound_5}\mathbb{P}\left(|T_{S,5}| \leq \mathcal{C}_{S,5}^\delta \zeta_{S,\lambda}\| \func^\ast_{S(n)} - \mathcal{W}_{S,\lambda}\func^\ast_{S(n)} \|_{S,\lambda} \mathcal{D}_{S,\lambda} \right) \geq 1-\delta/4 \text{ for all } n \geq N_{S,5}^\delta.\end{align} 
	
	We choose $N_{S,\delta} = \max\{ N_{S,\mu}, N_{S,1}^\delta, N_{S,3}^\delta, N_{S,4}^\delta, N_{S,5}^\delta \}$. Assume \begin{align}\label{eqn:power_bound_n}\|\func^\ast_{S(n)} \|_{S,\lambda} \geq \mathcal{C}_{S,\delta} \mathcal{D}_{S,\lambda} \text{ for all } n \geq N_{S,\delta},\end{align} where the specific choices of $\mathcal{C}_{S,\delta}$ is given at the end of the proof. For $n \geq N_{S,\delta}$, as $\lambda J_S(\func^\ast_{S(n)}) \leq \mathcal{D}_{S,\lambda}^2 \leq \|\func^\ast_{S(n)} \|_{S,\lambda}^2/\mathcal{C}_{S,\delta}^2$, we get \begin{align}\label{eqn:power_bound_nn}
		\| \func^\ast_{S(n)} - \mathcal{W}_{S,\lambda}\func^\ast_{S(n)} \|_{S,\lambda}^2 &= \| \func^\ast_{S(n)} \|_{S,\lambda}^2 + \|  \mathcal{W}_{S,\lambda}\func^\ast_{S(n)} \|_{S,\lambda}^2 -2 \lambda J_S(\func^\ast_{S(n)})\\&\geq \| \func^\ast_{S(n)} \|_{S,\lambda}^2 -2 \lambda J_S(\func^\ast_{S(n)}) \geq (1-2/\mathcal{C}_{S,\delta}^2)\| \func^\ast_{S(n)} \|_{S,\lambda}^2\geq(\mathcal{C}_{S,\delta}^2-2) \mathcal{D}_{S,\lambda}^2.
	\end{align} 
	
	By \eqref{eqn:power_bound_0}, \eqref{eqn:power_bound_1}, \eqref{eqn:power_bound_3}, \eqref{eqn:power_bound_4}, \eqref{eqn:power_bound_5}, \eqref{eqn:power_bound_n}, \eqref{eqn:power_bound_nn}, when $n \geq N_{S,\delta}$, with probability greater than or equal to $1-\delta$, the following holds: \begin{align}
		\left| \mathcal{T}_{S,\lambda} \right| &\geq  \mathcal{T}_{S,\lambda} \geq T_{S,1}+T_{S,2}+T_{S,3}+T_{S,4}+T_{S,5} \\& \geq  \zeta_{S,\lambda}\| \func^\ast_{S(n)} - \mathcal{W}_{S,\lambda}\func^\ast_{S(n)} \|_{S,\lambda}^2 -(\mathcal{C}_{S,3}^\delta+\mathcal{C}_{S,5}^\delta) \zeta_{S,\lambda}\| \func^\ast_{S(n)} - \mathcal{W}_{S,\lambda}\func^\ast_{S(n)} \|_{S,\lambda} \mathcal{D}_{S,\lambda} -(\mathcal{C}_{S,1}^\delta+\mathcal{C}_{S,4}^\delta)\\
		&  \geq   \zeta_{S,\lambda}\| \func^\ast_{S(n)} - \mathcal{W}_{S,\lambda}\func^\ast_{S(n)} \|_{S,\lambda}^2 -\frac{\mathcal{C}_{S,3}^\delta+\mathcal{C}_{S,5}^\delta}{\mathcal{C}_{S,\delta}} \zeta_{S,\lambda}\| \func^\ast_{S(n)} - \mathcal{W}_{S,\lambda}\func^\ast_{S(n)} \|_{S,\lambda} \|\func^\ast_{S(n)}\|_{S,\lambda} -(\mathcal{C}_{S,1}^\delta+\mathcal{C}_{S,4}^\delta)\\
		& \geq   \zeta_{S,\lambda}\| \func^\ast_{S(n)} - \mathcal{W}_{S,\lambda}\func^\ast_{S(n)} \|_{S,\lambda}^2 -\frac{\mathcal{C}_{S,3}^\delta+\mathcal{C}_{S,5}^\delta}{\mathcal{C}_{S,\delta}\sqrt{1-2/\mathcal{C}_{S,\delta}^2}} \zeta_{S,\lambda}\| \func^\ast_{S(n)} - \mathcal{W}_{S,\lambda}\func^\ast_{S(n)} \|_{S,\lambda}^2 -(\mathcal{C}_{S,1}^\delta+\mathcal{C}_{S,4}^\delta) \\
		&  \geq   \left(1 -\frac{\mathcal{C}_{S,3}^\delta+\mathcal{C}_{S,5}^\delta}{\mathcal{C}_{S,\delta}\sqrt{1-2/\mathcal{C}_{S,\delta}^2}}\right) \zeta_{S,\lambda}(\mathcal{C}_{S,\delta}^2-2) \mathcal{D}_{S,\lambda}^2 -(\mathcal{C}_{S,1}^\delta+\mathcal{C}_{S,4}^\delta) \\
		&  \geq   \left(1 -\frac{\mathcal{C}_{S,3}^\delta+\mathcal{C}_{S,5}^\delta}{\sqrt{\mathcal{C}_{S,\delta}^2-2}}\right) \frac{\mathcal{C}_{S,\delta}^2-2}{\mathcal{C}_{S,\mu}} -(\mathcal{C}_{S,1}^\delta+\mathcal{C}_{S,4}^\delta) \geq z_{1-\alpha/2},
	\end{align} where $\mathcal{C}_{S,\delta}$ is chosen large enough to make no contradiction in all derivations presented here. \qed
	\newpage
	\section{Details in Reproducing Kernel Hilbert Space}\label{sec:detail_RKHS}
	{\color{black}\subsection{Connection Between $\mathcal{H}$ and the Sobolev Spaces on $[0,1]^d$}\label{sec:sobolev_spaces} Here we provide relationship between $\mathcal{H}$ and the $m$th-order Sobolev space on $\mathcal{X}$.} For \( f:\mathcal{X} \rightarrow \mathbb{R} \), and a multi-index \( \boldsymbol{\alpha}=\{\alpha_{j}\}_{j=1}^d \) with each \( \alpha_{j}  \in \mathbb{N}_0 \equiv \{0\} \cup \mathbb{N} \), define \( f^{(\boldsymbol{\alpha})} \) as  
	\[
	f^{(\boldsymbol{\alpha})}(x)=\frac{\partial^{\sum_{j=1}^d \alpha_{j}}}{\partial x_{[1]}^{\alpha_{1}} \ldots \partial x_{[d]}^{\alpha_{d}} }f(x).
	\]  Additionally, define  \[
	|\boldsymbol{\alpha}|_1=\sum_{j=1}^d \alpha_{j} \text{ and } |\boldsymbol{\alpha}|_\infty=\sup_{j \in \{1,\ldots,d\}} \alpha_{j}.
	\]  The \(\mathcal{L}_2\) space on \(\mathcal{X}\) is given by  
	\[
	\mathcal{L}_2(\mathcal{X})=\left\{f : \mathcal{X} \rightarrow \mathbb{R}, \int_{\mathcal{X}} |f(x)|^2 dx < \infty \right\}.
	\]  
	Define the \( m \)th-order Sobolev space on \( \mathcal{X} \) as \[
	\mathcal{F}=\left\{f \in \mathcal{L}_2(\mathcal{X}) : f^{(\boldsymbol{\alpha})} \in \mathcal{L}_2(\mathcal{X}) \text{ for all } \boldsymbol{\alpha} \in \mathbb{N}_0^d\text{ with } |\boldsymbol{\alpha}|_1 \leq m\right\}.
	\] Similarly, define \[
	\mathscr{F}=\left\{f \in \mathcal{L}_2(\mathcal{X}) : f^{(\boldsymbol{\alpha})} \in \mathcal{L}_2(\mathcal{X}) \text{ for all } \boldsymbol{\alpha} \in \mathbb{N}_0^d\text{ with } |\boldsymbol{\alpha}|_\infty \leq m \right\}.
	\] Note that \( \otimes_{j=1}^d\mathcal{H}_{[j]} \) is a dense subspace of \( \mathscr{F} \), leading to \[
	\mathcal{H} \subseteq \otimes_{j=1}^d\mathcal{H}_{[j]} \subseteq \mathscr{F} \subseteq \mathcal{F} \subseteq \mathcal{L}_2(\mathcal{X}),\] where $\mathcal{H} = \otimes_{j=1}^d\mathcal{H}_{[j]} = \mathscr{F} = \mathcal{F} \subseteq \mathcal{L}_2(\mathcal{X})$ when $d=1$.
	
	{\color{black}\subsection{Entropy Bound for $\mathcal{H}$}\label{sec:covering} Here we provide the entropy bound for the unit ball of $\mathcal{H}$, which is used in the proof of Theorem \ref{thm:fbr}. The eigenvalue order in the following Remark~\ref{rmk:eigen_H} \citep{gu2013smoothing,du_2010,Ma_2015} is used in the derivation. \begin{remark}\label{rmk:eigen_H}
			Let $\tau = 0$ when $|S_{\sup}| = 1$, and $\tau$ can take any value in $(0, 2m-2)$ when $|S_{\sup}| > 1$. There exist a constant $\mathcal{C}_{\tau} \in (0,\infty)$ and $2<v_{\tau} \in \mathbb{N}$ such that $\mu_v^{-1} \geq \mathcal{C}_{\tau} v^{2m-\tau}$ for all $v \geq v_{\tau}$. 
		\end{remark} In particular, Remark~\ref{rmk:eigen_H} corresponds to the eigenvalue–based summations in Lemma~\ref{lemma:eigen_order}, 
		whose order increases only by a logarithmic factor as $|S_{\sup}|$ grows.
		
		Note that $\mu_1^{-1}=\mu_{\emptyset,0}^{-1}=0$ but $\mu_v^{-1} > 0$ for all $v \geq 2$. The following Lemma~\ref{lem:ellipsoid} provides a crude entropy bound for a Euclidean ellipsoid generated by a finite subset of eigenvalues, but it is sufficient for our purpose. \begin{lemma}\label{lem:ellipsoid}
			Let $\mathcal{E}= \left\{ \boldsymbol{c}=(c_2,\ldots,c_\kappa)^\top \in \mathbb{R}^{\kappa-1}: \sum_{v=2}^\kappa \mu_v^{-1} c_v^2 \leq 1 \right\}$ with $\kappa > 2$. For any $\delta \in (0,\infty)$, we have $$\mathrm{N}\left( \delta, \mathcal{E} ,\|\cdot\|_2 \right) \leq \prod_{v=2}^\kappa \left( 1 + \frac{2\sqrt{\kappa-1}}{\delta\mu_v^{-1/2}} \right).$$
		\end{lemma} \begin{proof}[Proof of Lemma \ref{lem:ellipsoid}.] %First note that for any $\boldsymbol{c} \in \mathcal{E}_\kappa$, for each $v\in\{2,3,\ldots,\kappa\}$, we have $\rho_{v}c_{v}^2 \leq  \sum_{v'=2}^\kappa \rho_{v'} c_{v'}^2 \leq 1$, which implies $|c_{v}| \leq \rho_v^{-1/2}$. Hence $$\mathcal{E}_\kappa \subseteq \prod_{v=2}^\kappa [-\rho_v^{-1/2},\rho_v^{-1/2}].$$ 
			Let $$\mathcal{G}_{\delta}=\prod_{v=2}^\kappa \left( \frac{\delta}{\sqrt{\kappa-1}}\mathbb{Z} \cap [-\mu_v^{1/2},\mu_v^{1/2}]  \right) \subseteq \mathbb{R}^{\kappa-1},$$ where $$ |\mathcal{G}_\delta| \leq \prod_{v=2}^\kappa \left( 1+ 2\frac{ \mu_v^{1/2}}{\delta/\sqrt{\kappa-1}} \right).$$
			
			For all $\boldsymbol{c}\in \mathcal{E}$, for each $v\in\{2,\ldots,\kappa\}$, we have $\mu_{v}^{-1}c_{v}^2 \leq  \sum_{v'=2}^\kappa \mu_{v'}^{-1} c_{v'}^2 \leq 1$, which implies $|c_{v}| \leq \mu_v^{1/2}$. So there exists $$g_{\delta,v} \in  \frac{\delta}{\sqrt{\kappa-1}}\mathbb{Z} \cap [-\mu_v^{1/2},\mu_v^{1/2}]$$ such that $|c_v-g_{\delta,v}| \leq \delta/\sqrt{\kappa-1}$. Then $\boldsymbol{g}_\delta = (g_{\delta,2},\ldots,g_{\delta,\kappa})^\top \in \mathcal{G}_\delta$ and $$\|\boldsymbol{c}-\boldsymbol{g}_{\delta}\|_2 = \sqrt{\sum_{v=2}^\kappa |c_v-g_{\delta,v}|^2  } \leq \sqrt{\kappa-1}\frac{\delta}{\sqrt{\kappa-1}} = \delta,$$ which completes the proof.\end{proof}
		
		The following Lemma \ref{lem:entropy} gives the entropy bound for the unit ball of $\mathcal{H}$, which is derived in a manner similar to Proposition~C.8. of \cite{sieve_2023}. \begin{lemma}\label{lem:entropy}
			Let $\tau = 0$ when $|S_{\sup}| = 1$, and $\tau$ can take any value in $(0, 2m-2)$ when $|S_{\sup}| > 1$. Denote by $\mathcal{C}_{\mathcal{H},\tau} \in (0,\infty)$ a constant. When $0<\eta \rightarrow 0$ as $n \rightarrow \infty$, we have $$\log\mathrm{N}\left( \eta, \{ f \in \mathcal{H}: \|f\|_{\sup} \leq 1, J(f) \leq 1\}, \|\cdot\|_{\sup} \right) \leq \mathcal{C}_{\mathcal{H},\tau} \left(\frac{1}{\eta}\right)^{2/(2m-\tau-1)}\log\left(\frac{1}{\eta}\right)$$ as $n \rightarrow \infty$.\end{lemma} Since $2m>1$, we have $\|f\|_{\sup} = \|f\|_\infty$ for all $f \in \otimes_{j=1}^d \mathcal{H}_{[j]}$, 
		and thus $\|\cdot\|_{\sup}$ in Lemma~\ref{lem:entropy} can be just written as $\|\cdot\|_\infty$. \begin{proof}[Proof of Lemma \ref{lem:entropy}.]Note that \begin{align}
				\mathds{B} &\equiv \{ f \in \mathcal{H}: \|f\|_{\sup} \leq 1, J(f) \leq 1\} \\&= \left\{ f = \sum_{v\in\mathbb{N}} V(f,\psi_v)\psi_v \in \mathcal{H} : \|f\|_{\sup} \leq 1, \sum_{v\in\mathbb{N}} \mu_v^{-1} V^2(f,\psi_v) \leq 1\right\}.
			\end{align} 
			
			First, define \begin{align}
				\mathds{B}_1=\left\{ c \psi_1: c \in [-1,1] \right\}.
			\end{align} Note that \begin{align}\label{eqn:cover_1}
				\mathrm{N}_{\eta,1}\equiv \mathrm{N}\left(\eta/3,[-1,1],|\cdot|\right) \leq 1+6/\eta
			\end{align} and let $\mathrm{C}_{\eta,1}\equiv\{c_{\eta,1},\ldots,c_{\eta,\mathrm{N}_{\eta,1}}\} \subseteq \mathbb{R}$ be $\eta/3$-cover of $[-1,1]$ with respect to $|\cdot|$. For any $c\psi_1 \in \mathds{B}_1$, there exists $c_{\eta,\ell} \in \mathrm{C}_{\eta,1}$ such that \begin{align}\label{eqn:net_B1}
				\|c\psi_1 - c_{\eta,\ell} \psi_1 \|_{\sup} = |c - c_{\eta,\ell}|\cdot\|\psi_{\emptyset,0}\|_{\sup}=|c - c_{\eta,\ell}| \leq \eta/3,
			\end{align} which leads to $\left\{c_{\eta,\ell}\psi_1\right\}_{\ell=1}^{\mathrm{N}_{\eta,1}} \subseteq \mathcal{H} \subseteq  \mathcal{L}_2(\mathcal{X})$ forming a $\eta/3$-cover of $\mathds{B}_1$ with respect to $\|\cdot\|_{\sup}$. %Thus, by \eqref{eqn:cover_1}, we have \begin{align}\label{eqn:cover_B1}
				%    \mathrm{N}_\eta(\mathds{B}_1) \equiv \mathrm{N}\left(\eta/3, \mathds{B}_1, \|\cdot\|_{\sup} \right) \leq 1+6/\eta.
				%\end{align} 
				
				Second, define \begin{align}
					\mathds{B}_{\eta,2}=\left\{ \sum_{v=2}^{\kappa(\eta)}c_v \psi_v: \sum_{v=2}^{\kappa(\eta)} \mu_v^{-1} c_v^2 \leq 1\right\},
				\end{align}where with the ceiling operator $\lceil\cdot\rceil$, \begin{align}
					\kappa(\eta) = \left\lceil \left(\frac{1}{\eta}\right)^{2/(2m-\tau-1)} \left(\frac{9\mathcal{C}_\psi^2}{\mathcal{C}_\tau(2m-\tau-1)}\right)^{1/(2m-\tau-1)}\right\rceil\in\mathbb{N}
				\end{align} and $\kappa(\eta) \rightarrow \infty$ as $n \rightarrow \infty$. By Lemma \ref{lem:ellipsoid}, we get \begin{align}
					\mathrm{N}_{\eta,2} &\equiv \mathrm{N}\left(\frac{\eta}{3\mathcal{C}_\psi\sqrt{\kappa(\eta)-1}}, \left\{ \boldsymbol{c}=(c_2,\ldots,c_{\kappa(\eta)})^\top \in \mathbb{R}^{\kappa(\eta)-1}: \sum_{v=2}^{\kappa(\eta)} \mu_v^{-1} c_v^2 \leq 1 \right\}, \|\cdot\|_2\right)\\\label{eqn:cover_2}& \leq \prod_{v=2}^{\kappa(\eta)} \left( 1 + \frac{6\mathcal{C}_\psi(\kappa(\eta)-1)}{\eta \mu_v^{-1/2}} \right). 
				\end{align} Let $\mathrm{C}_{\eta,2}\equiv\{\boldsymbol{c}_{\eta,1},\ldots,\boldsymbol{c}_{\eta,\mathrm{N}_{\eta,2}}\} \subseteq \mathbb{R}^{\kappa(\eta)-1}$ be $$\frac{\eta}{3\mathcal{C}_\psi\sqrt{\kappa(\eta)-1}}\text{-cover of } \left\{ \boldsymbol{c}=(c_2,\ldots,c_{\kappa(\eta)})^\top \in \mathbb{R}^{\kappa(\eta)-1}: \sum_{v=2}^{\kappa(\eta)} \mu_v^{-1} c_v^2 \leq 1 \right\}$$ with respect to $\|\cdot\|_2$. For any  $\sum_{v=2}^{\kappa(\eta)}c_v \psi_v \in \mathds{B}_{\eta,2}$, we have $(c_2,\ldots,c_{\kappa(\eta)})^\top \in \mathbb{R}^{\kappa(\eta)-1}$ by $\sum_{v=2}^{\kappa(\eta)} \mu_v^{-1} c_v^2 \leq 1$. So, there exists $\boldsymbol{c}_{\eta,\ell}=(c_{\eta,\ell,2},\ldots,c_{\eta,\ell,\kappa(\eta)})^\top \in \mathrm{C}_{\eta,2}$ such that \begin{align}\label{eqn:net_B2}
					\left\| \sum_{v=2}^{\kappa(\eta)}c_v \psi_v -\sum_{v=2}^{\kappa(\eta)}c_{\eta,\ell,v} \psi_v  \right\|_{\sup} &\leq \sum_{v=2}^{\kappa(\eta)} |c_v-c_{\eta,\ell,v}| \cdot \|\psi_v\|_{\sup} \leq \mathcal{C}_\psi \sum_{v=2}^{\kappa(\eta)} |c_v-c_{\eta,\ell,v}|\\&\leq \mathcal{C}_\psi \sqrt{\kappa(\eta)-1}\sqrt{\sum_{v=2}^{\kappa(\eta)} |c_v-c_{\eta,\ell,v}|^2} \leq \eta/3,
				\end{align} which leads to $$\left\{\sum_{v=2}^{\kappa(\eta)}c_{\eta,\ell,v} \psi_v\right\}_{\ell=1}^{\mathrm{N}_{\eta,2}} \subseteq \mathcal{H} \subseteq  \mathcal{L}_2(\mathcal{X})$$ forming a $\eta/3$-cover of $\mathds{B}_{\eta,2}$ with respect to $\|\cdot\|_{\sup}$. %Thus, by \eqref{eqn:cover_2}, we have \begin{align}\label{eqn:cover_B2}
					%    \mathrm{N}_\eta(\mathds{B}_{\eta,2}) \equiv \mathrm{N}\left(\eta/3, \mathds{B}_{\eta,2}, \|\cdot\|_{\sup} \right) \leq \prod_{v=2}^{\kappa(\eta)} \left( 1 + \frac{6\mathcal{C}_\psi(\kappa(\eta)-1)}{\eta \rho_v^{1/2}} \right).
					%\end{align}
					
					For any $f \in \mathds{B}$, we have $V(f,\psi_v) \in \mathbb{R}$ for all $v \in \mathbb{N}$ by $\sum_{v\in\mathbb{N}} \mu_v^{-1} V^2(f,\psi_v) \leq 1$ and \begin{align}
						f&=V(f,\psi_1)\psi_1 + \sum_{v=2}^{\kappa(\eta)}V(f,\psi_v)\psi_v + \sum_{v=\kappa(\eta)+1}^{\infty}V(f,\psi_v)\psi_v \\&=f_1 + f_{\eta,2} + f_{\eta,3}.   
					\end{align} For $f_{\eta,3}$, we have \begin{align}
						\left\|\sum_{v=\kappa(\eta)+1}^{\infty}V(f,\psi_v)\psi_v\right\|_{\sup} & \leq \sum_{v=\kappa(\eta)+1}^{\infty}|V(f,\psi_v)| \cdot \|\psi_v\|_{\sup} \\
						& \leq \mathcal{C}_\psi \sum_{v=\kappa(\eta)+1}^{\infty}|V(f,\psi_v)| \mu_v^{-1/2} \mu_v^{1/2}\\
						\label{eqn:f3_bound}& \leq \mathcal{C}_\psi \sqrt{\sum_{v=\kappa(\eta)+1}^{\infty}\mu_v^{-1} V^2(f,\psi_v)} \sqrt{\sum_{v=\kappa(\eta)+1}^{\infty} {\mu_v}}.\\
						& \leq \mathcal{C}_\psi \sqrt{\sum_{v\in\mathbb{N}}\mu_v^{-1} V^2(f,\psi_v)} \sqrt{\sum_{v=\kappa(\eta)+1}^{\infty} {\mu_v}} \leq \mathcal{C}_\psi \sqrt{\sum_{v=\kappa(\eta)+1}^{\infty} {\mu_v}} \leq \eta/3,
					\end{align}where the last inequality follows from $v_\tau \leq \kappa(\eta)\rightarrow \infty$, $2m-\tau>2>1$, and \begin{align}
						\sum_{v=\kappa(\eta)+1}^{\infty} {\mu_v} &\leq \frac{1}{\mathcal{C}_\tau}\sum_{v=\kappa(\eta)+1}^{\infty}\frac{1}{v^{2m-\tau}}\\
						&\leq \frac{1}{\mathcal{C}_\tau}\int_{\kappa(\eta)}^\infty \frac{1}{v^{2m-\tau}} dv \\
						& =\frac{1}{\mathcal{C}_\tau} \frac{1}{2m-\tau-1} \kappa(\eta)^{-(2m-\tau-1)}\\
						& \leq\frac{1}{\mathcal{C}_\tau} \frac{1}{2m-\tau-1} \left(\left(\frac{1}{\eta}\right)^{2/(2m-\tau-1)} \left(\frac{9\mathcal{C}_\psi^2}{\mathcal{C}_\tau(2m-\tau-1)}\right)^{1/(2m-\tau-1)}\right)^{-(2m-\tau-1)}=\frac{\eta^2}{9\mathcal{C}_\psi^2}.
					\end{align}
					
					For $f_1$, as \begin{align}
						|V(f,\psi_1)| \leq \sqrt{V(f,f)}\sqrt{V(\psi_{\emptyset,0},\psi_{\emptyset,0})}=\sqrt{\int_{\mathcal{X}}|f(x)|^2dx} \leq \|f\|_{\sup}\sqrt{\int_{\mathcal{X}}1 dx} \leq 1,
					\end{align} we have $f_1 \in \mathds{B}_1$. So there exists $g_{\eta,1} \in \left\{c_{\eta,\ell}\psi_1\right\}_{\ell=1}^{\mathrm{N}_{\eta,1}}$ such that \begin{align}\label{eqn:f1_net}
						\|f_1 - g_{\eta,1}\|_{\sup} \leq \eta/3.
					\end{align}
					
					For $f_{\eta,2}$, as $\sum_{v=2}^{\kappa(\eta)} \mu_v^{-1} V^2(f,\psi_v) \leq \sum_{v\in\mathbb{N}} \mu_v^{-1} V^2(f,\psi_v) \leq 1$, we have $f_{\eta,2} \in \mathds{B}_{\eta,2}$. So there exists $g_{\eta,2} \in \left\{\sum_{v=2}^{\kappa(\eta)}c_{\eta,\ell,v} \psi_v\right\}_{\ell=1}^{\mathrm{N}_{\eta,2}}$ such that \begin{align}\label{eqn:f2_net}
						\|f_{\eta,2} - g_{\eta,2}\|_{\sup} \leq \eta/3.
					\end{align}
					
					Combining \eqref{eqn:f3_bound}, \eqref{eqn:f1_net}, and \eqref{eqn:f2_net}, with $g_\eta=g_{\eta,1}+g_{\eta,2} \in \mathcal{H} \subseteq \mathcal{L}_2(\mathcal{X})$, we have \begin{align}
						\|f-g_\eta\|_{\sup} \leq  \|f_1 - g_{\eta,1}\|_{\sup} + \|f_{\eta,2} - g_{\eta,2}\|_{\sup} + \|f_{\eta,3}\|_{\sup} \leq \eta,
					\end{align} which, by \eqref{eqn:cover_1} and \eqref{eqn:cover_2}, leads to \begin{align}
						&\log \mathrm{N}\left( \eta, \{ f \in \mathcal{H}: \|f\|_{\sup} \leq 1, J(f) \leq 1\}, \|\cdot\|_{\sup} \right) \\&\quad\leq \log \left( \mathrm{N}_{\eta,1} \cdot \mathrm{N}_{\eta,2} \right) \\ &\quad\leq \log(1+6/\eta) +   \sum_{v=2}^{v_\tau-1} \log \left( 1 + \frac{6\mathcal{C}_\psi(\kappa(\eta)-1)}{\eta \mu_v^{-1/2}} \right) + \sum_{v=v_\tau}^{\kappa(\eta)} \log \left( 1 + \frac{6\mathcal{C}_\psi(\kappa(\eta)-1)}{\eta \mu_v^{-1/2}} \right)\\&\quad \lesssim \left(\frac{1}{\eta}\right)^{2/(2m-\tau-1)}\log\left(\frac{1}{\eta}\right),
					\end{align} 
					where $\log(1+6/\eta) \asymp \log(1/\eta) \ll({1}/{\eta})^{2/(2m-\tau-1)}\log\left({1}/{\eta}\right)$, \begin{align}
						\sum_{v=2}^{v_\tau-1} \log \left( 1 + \frac{6\mathcal{C}_\psi(\kappa(\eta)-1)}{\eta \mu_v^{-1/2}} \right) &\asymp \log(\kappa(\eta)) + \log\left(\frac{1}{\eta}\right)\\
						&\asymp \left(\frac{2}{2m-\tau-1} + 1\right) \log\left(\frac{1}{\eta}\right) \asymp \log\left(\frac{1}{\eta}\right) \\&\ll \left(\frac{1}{\eta}\right)^{2/(2m-\tau-1)}\log\left(\frac{1}{\eta}\right),
					\end{align} and using $\mu_v^{-1/2} \geq \mathcal{C}_\tau^{1/2} v^{(2m-\tau)/2}  \geq  \mathcal{C}_\tau^{1/2}$ for $v \geq v_\tau$, \begin{align}
						\sum_{v=v_\tau}^{\kappa(\eta)} \log \left( 1 + \frac{6\mathcal{C}_\psi(\kappa(\eta)-1)}{\eta \mu_v^{-1/2}} \right) &\leq (\kappa(\eta)-v_\tau+1) \log \left( 1 + \frac{6\mathcal{C}_\psi(\kappa(\eta)-1)}{\eta \mathcal{C}_\tau^{1/2}} \right)\\
						& \asymp \kappa(\eta)\left( \log(\kappa(\eta)) + \log\left(\frac{1}{\eta}\right) \right)\\&\asymp \left(\frac{1}{\eta}\right)^{2/(2m-\tau-1)}\log\left(\frac{1}{\eta}\right).
			\end{align}\end{proof}}
			
			\section{Additional Details of Simulation Studies}\label{sec:additional_simul}
			Here, we provide additional details of simulation studies in Section \ref{sec:simul}. Figure~\ref{fig:true_components} displays effect functions $g^\ast_S$, and Figure~\ref{fig_intercept} presents the results for the intercept in the effect-wise confidence interval simulations.
			
			\begin{figure}[H]
				\centering
				% Row 1: main effects
				\begin{subfigure}{0.27\textwidth}
					\centering
					\includegraphics[width=0.8\linewidth,height=0.8\linewidth]{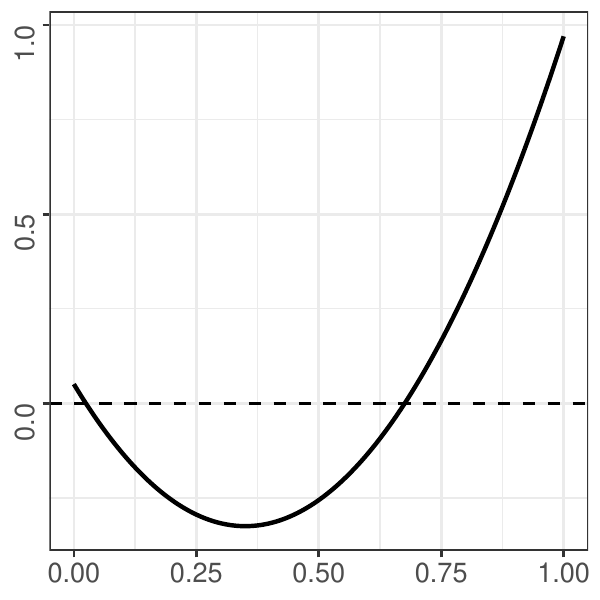}
					\caption{$g^\ast_{\{1\}}$}
				\end{subfigure}
				\begin{subfigure}{0.27\textwidth}
					\centering
					\includegraphics[width=0.8\linewidth,height=0.8\linewidth]{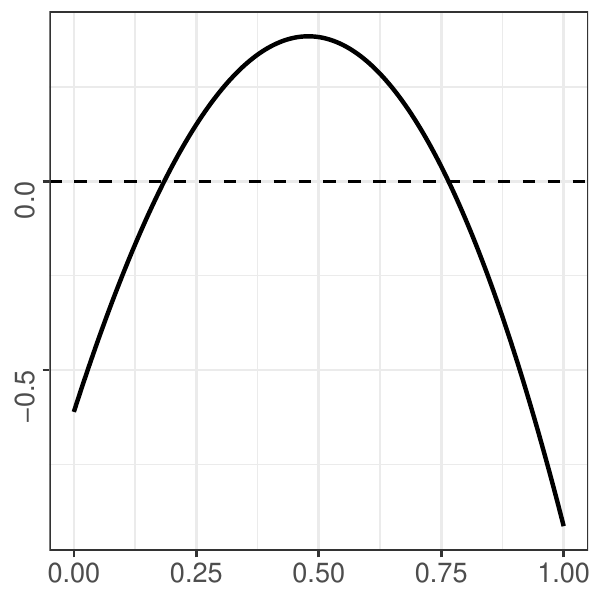}
					\caption{$g^\ast_{\{2\}}$}
				\end{subfigure}
				\begin{subfigure}{0.27\textwidth}
					\centering
					\includegraphics[width=0.8\linewidth,height=0.8\linewidth]{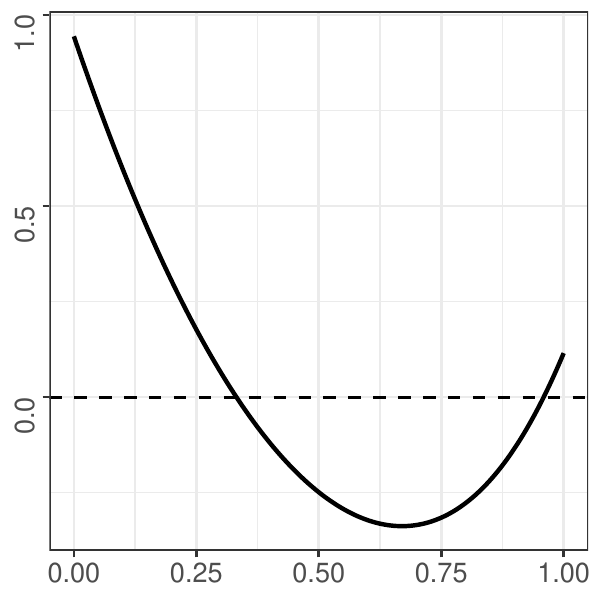}
					\caption{$g^\ast_{\{3\}}$}
				\end{subfigure}
				
				% Row 2: interaction (1,2)
				\vspace{0.3ex}
				\begin{subfigure}{0.81\textwidth}
					\centering
					\includegraphics[width=\linewidth,height=0.3\linewidth]{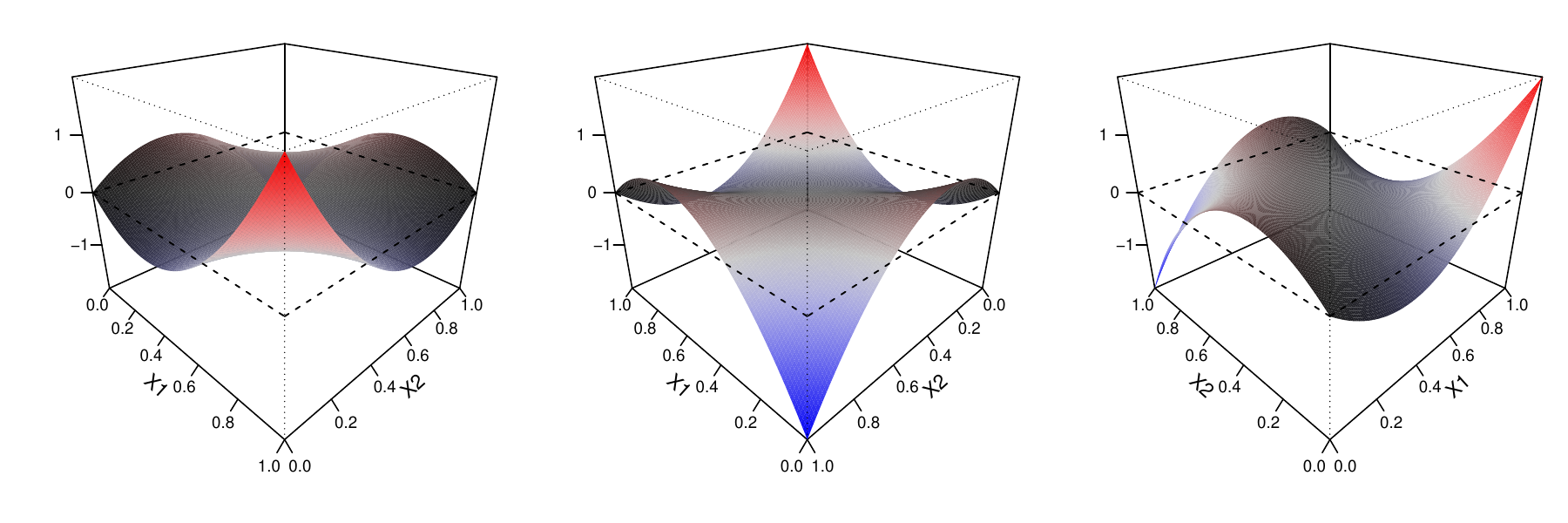}
					\caption{$g^\ast_{\{1,2\}}$}
				\end{subfigure}
				
				% Row 3: interaction (1,3)
				\vspace{0.3ex}
				\begin{subfigure}{0.81\textwidth}
					\centering
					\includegraphics[width=\linewidth,height=0.3\linewidth]{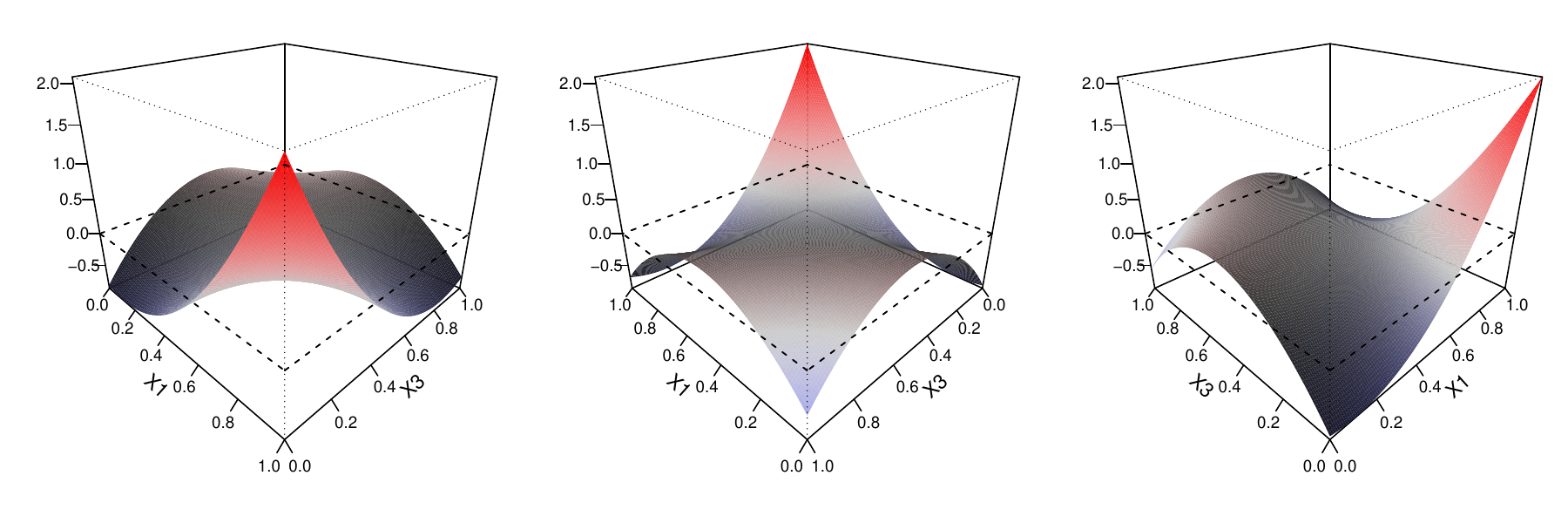}
					\caption{$g^\ast_{\{1,3\}}$}
				\end{subfigure}
				
				% Row 4: interaction (2,3)
				\vspace{0.3ex}
				\begin{subfigure}{0.81\textwidth}
					\centering
					\includegraphics[width=\linewidth,height=0.3\linewidth]{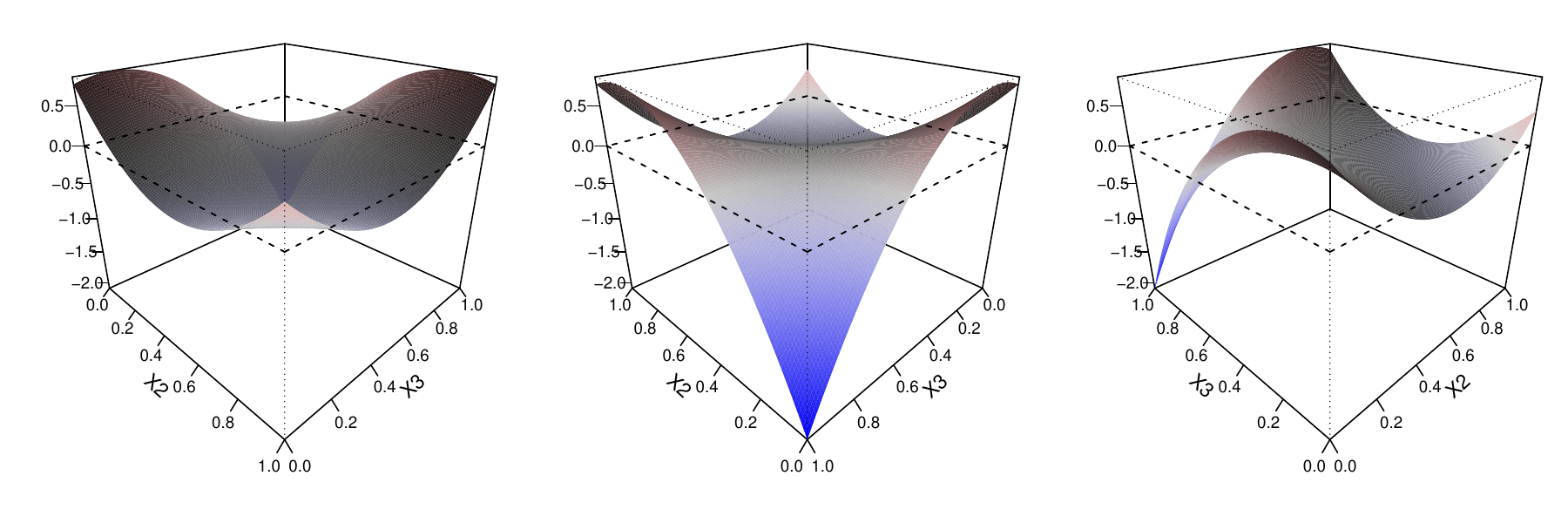}
					\caption{$g^\ast_{\{2,3\}}$}
				\end{subfigure}
				
				\caption{Main and two-factor interaction effects $g^\ast_S$.}\label{fig:true_components}
			\end{figure}
			
			\begin{figure}[H]
				\centering
				\includegraphics[scale=0.35]{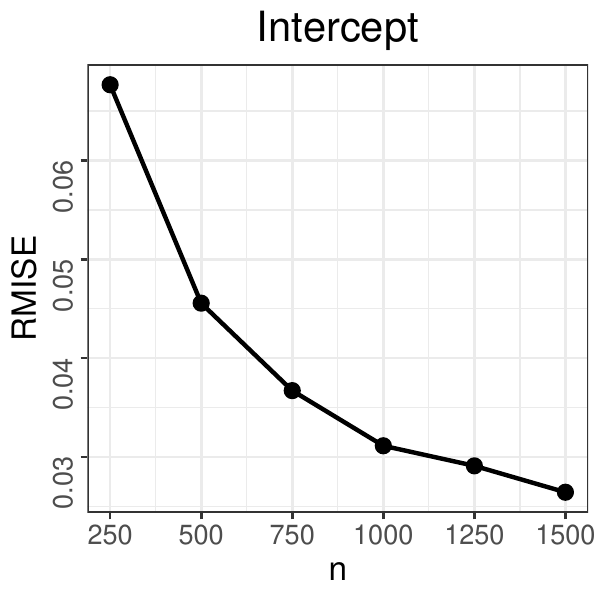}
				\includegraphics[scale=0.35]{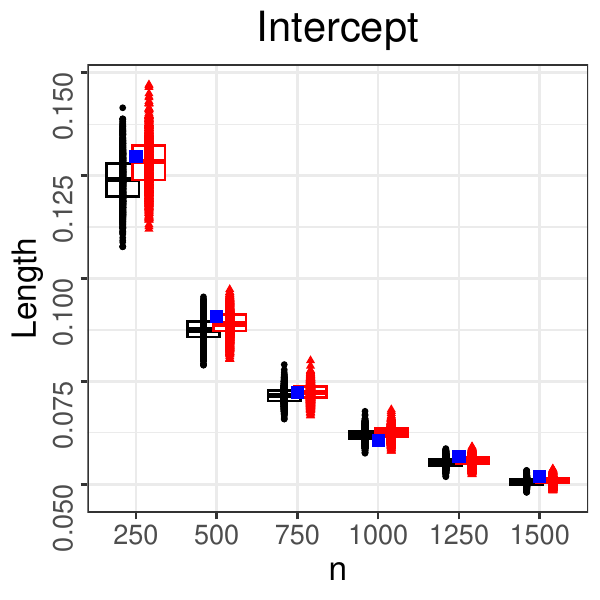}
				\includegraphics[scale=0.35]{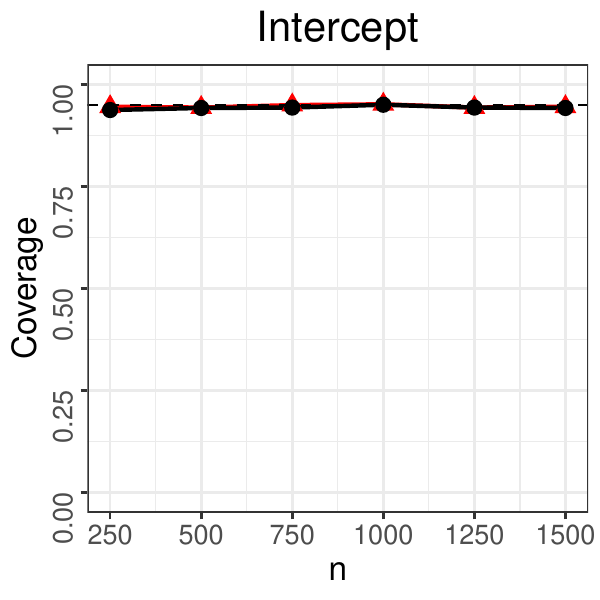}
				\includegraphics[scale=0.7]{Graphics/legend_ssaec_ssaebc_empirical.pdf}
				\caption{RMISE, interval length, and coverage of the intercept. The dashed horizontal line in the coverage panel indicates the nominal level $1-\alpha=0.95$.}\label{fig_intercept}
			\end{figure}
			
		\end{document}